\documentclass[11pt, a4paper, headsepline=false, footsepline=false]{scrartcl}

\usepackage[hmargin={2.5cm,2.5cm}, vmargin={2.5cm,2.25cm}, includeheadfoot]{geometry}
\usepackage[latin1]{inputenc}
\usepackage[automark]{scrpage2}
\usepackage{setspace}
\usepackage{amsmath, amsfonts, amssymb, amsthm}
\usepackage{mathtools}
\usepackage{dsfont}
\usepackage{algorithm}
\usepackage{algpseudocode}
\usepackage{booktabs}
\usepackage{graphicx}
\usepackage[belowskip=-7pt]{caption}
\usepackage{subcaption}
\usepackage[round]{natbib}
\usepackage[colorlinks, citecolor=black, linkcolor=black, urlcolor=black]{hyperref}

\setkomafont{sectioning}{\normalfont\bfseries}
\automark[subsection]{section}

\numberwithin{equation}{section}

\newtheorem{thm}{Theorem}[section]
\newtheorem{lem}[thm]{Lemma}

\theoremstyle{definition}
\newtheorem*{exmp}{Example}

\newcommand{\Nat}{\mathds{N}}
\newcommand{\Real}{\mathds{R}}

\DeclareMathOperator{\EW}{\mathds{E}}
\DeclareMathOperator{\Var}{\mathds{V}ar}

\newcommand{\Edges}{\mathcal{E}}
\newcommand{\edge}{\relbar}

\DeclareMathOperator{\ad}{ad}
\DeclareMathOperator{\an}{an}
\DeclareMathOperator{\An}{An}
\DeclareMathOperator{\de}{de}
\DeclareMathOperator{\nd}{nd}

\DeclareMathOperator{\pa}{pa}

\newcommand{\dperp}{{\perp}\!\!\!{\perp}}
\newcommand{\ndperp}{\not{\hspace*{-.45em}{\dperp}}}
\newcommand{\ind}{\mathrel{\dperp}}
\newcommand{\nind}{\mathrel{\ndperp}}
\newcommand{\given}{\mathrel{|}}
\newcommand{\cind}[3]{#1 \ind #2 \given #3}
\newcommand{\ncind}[3]{#1 \nind #2 \given #3}
\newcommand{\sep}[3]{#1 \perp #2 \, [ #3 ]}
\newcommand{\csep}[4]{#1 \perp #2 \given #3 \, [ #4 ]}

\algrenewcommand\algorithmicrequire{\textbf{Input}}
\algrenewcommand\algorithmicensure{\textbf{Output}}
\algrenewcommand{\algorithmiccomment}[1]{\% #1}

\DeclareMathOperator{\diff}{d\!}
\DeclareMathOperator{\Norm}{N}

\DeclareMathOperator{\dash}{\text{---}}
\newcommand{\abs}[1]{\left\lvert #1 \right\rvert}

\title{Pair-copula Bayesian networks}
\author{Alexander Bauer\thanks{Department of Mathematics, Technische Universität München, Boltzmannstr.\ 3, 85748 Garching, Germany. E-mail: \href{mailto:alexander.bauer@tum.de}{\texttt{alexander.bauer@tum.de}}, \href{mailto:cczado@ma.tum.de}{\texttt{cczado@ma.tum.de}}.}~\thanks{Corresponding author.} \and Claudia Czado\footnotemark[1]}
\date{}

\begin{document}

\onehalfspacing
\parindent0em

\maketitle

\begin{abstract}
\noindent \textbf{Abstract.} Pair-copula Bayesian networks (PCBNs) are a novel class of multivariate statistical models, which combine the distributional flexibility of pair-copula constructions (PCCs) with the parsimony of conditional independence models associated with directed acyclic graphs (DAG). We are first to provide generic algorithms for random sampling and likelihood inference in arbitrary PCBNs as well as for selecting orderings of the parents of the vertices in the underlying graphs. Model selection of the DAG is facilitated using a version of the well-known PC algorithm which is based on a novel test for conditional independence of random variables tailored to the PCC framework. A simulation study shows the PC algorithm's high aptitude for structure estimation in non-Gaussian PCBNs. The proposed methods are finally applied to modelling financial return data.
	
\bigskip
	
\end{abstract}

\section{Introduction}

Graphical models provide a powerful tool in multivariate statistical analysis aimed at modelling the conditional independence structure of a family of random variables. The conditional independence restrictions observed by a graphical model can be conveniently summarised in a graph whose vertices represent the variables and whose edges indicate interrelations between these variables, see \citet{Lauritzen:1996}. We are particularly interested in the graphical models known as Bayesian networks, whose Markov properties can be represented by a directed acyclic graph (DAG). Areas of applications for these Bayesian networks range from artificial intelligence, decision support systems, and engineering to genetics, geology, medicine, and finance, see \citet{Pourret.Naim.Marcot:2008}. Despite the broad scope of applicability, however, graphical modelling of continuous random variables has mainly been limited to the multivariate normal distribution. Accordingly, available structure estimation algorithms for the DAG underlying a Bayesian network are mainly confined to discrete or Gaussian models. We address both the problems of constructing Bayesian networks with non-Gaussian continuous joint distributions, and of estimating the Markov structure underlying such a non-Gaussian Bayesian network.

\bigskip

Our solution to the first problem of deriving non-Gaussian distributions with pre-specified conditional independence properties is based on so-called pair-copula constructions (PCCs). By iterated application of Sklar's theorem on copulas \citep{Sklar:1959}, \citet{Kurowicka.Cooke:2005} and \citet{Bauer.Czado.Klein:2012} have shown that every continuous multivariate distribution associated with a DAG can be decomposed into a family of bivariate, potentially conditional distributions, which correspond to the edges of the underlying graph. An explicit representation of the respective probability density function (pdf) was, however, only derived in examples. We provide a novel algorithm for evaluating the pdf of an arbitrary Bayesian network PCC.

\bigskip

The flexibility of these pair-copula Bayesian networks (PCBNs) allows for the capturing of a wide range of distributional features to be modelled such as heavy-tailedness, tail dependence, and non-linear, asymmetric dependence. Further investigations on PCBNs include \citet{Hanea.Kurowicka.Cooke:2006, Hanea.Kurowicka.Cooke.Ababei:2010} and \citet{Hanea.Kurowicka:2008}. While these authors concentrate on non-parametric statistical inference and elicited expert knowledge, we focus attention to parametric likelihood inference and data-driven structure estimation. We also provide routines for copula selection and enumeration of the parents of the vertices of the underlying DAG.

\bigskip

When expert knowledge on the underlying Markov structure is unavailable, data-driven structure estimation algorithms are frequently used. Two approaches are predominantly found in the literature: the constraint-based and the score-and-search-based approach \citep[Chapter $18$]{Koller.Friedman:2009}. In the former, the DAG is inferred from a series of conditional independence tests, while in the latter, the DAG is found by optimising a given scoring function. We concentrate on the popular constraint-based PC algorithm by \citet{Spirtes.Glymour:1991}, and demonstrate its aptitude for structure estimation in non-Gaussian PCBNs in an extensive simulation study. In particular, we introduce a novel test for conditional independence of continuous random variables which is based on the closely related regular-vine copula models \citep{Bedford.Cooke:2001, Bedford.Cooke:2002}, and which is of interest on its own merits. This novel test will prove to outperform a standard test for zero partial correlation used in the Gaussian setting.

\bigskip

With their focus on conditional independence, PCBNs are generally more parsimonious than regular-vine copula models. Another copula decomposition of a joint distribution associated with a DAG which uses generally higher-variate copulas---and therefore lacks the flexibility of the pair-copula approach---was investigated by \citet{Elidan:2010a, Elidan:2012b}.

\bigskip

The paper is organised as follows. In Section \ref{sec:bn}, we give a short review of Bayesian networks, followed by a review of vine copula models in Section \ref{sec:vine}. In Section \ref{sec:pcbn}, we provide an algorithm for evaluating the pdf of a PCC associated with a DAG as well as routines for simulation, model selection, and likelihood inference in PCBNs. We review the PC algorithm in Section \ref{sec:pc} and introduce a novel test for conditional independence of continuous random variables. The PC algorithm's aptitude for structure estimation in non-Gaussian PCBNs is explored in a simulation study in Section \ref{sec:simstudy}. Section \ref{sec:finance} presents an application of PCBNs to financial return data, and the paper concludes with a brief discussion in Section \ref{sec:conclusion}. The paper is designed to be self-contained and to unify the various non-standard notations on Bayesian networks found in the literature.

\section{Bayesian networks}\label{sec:bn}

We begin by fixing some graph theoretical terminology. Let $V \neq \emptyset$ be a finite set and let $E \subseteq \Edges \coloneqq \bigl\lbrace (v, w) \in V \times V \, \big\vert \, v \neq w \bigr\rbrace$. Then $\mathcal{G} = (V, E)$ denotes a \emph{graph} with \emph{vertex set} $V$ and \emph{edge set} $E$. We say that $\mathcal{G}$ contains the \emph{undirected edge} $v \edge w$ if $(v, w) \in E$ and $(w, v) \in E$. Similarly, we say that $V$ contains the \emph{directed edge} $v \rightarrow w$ if $(v, w) \in E$ but $(w, v) \notin E$. A graph containing only undirected edges is called an \emph{undirected graph} (UG). If $E \equiv \Edges$, we call $\mathcal{G}$ the \emph{complete UG} on $V$. A graph containing only directed edges is called a \emph{directed graph}. By replacing all directed edges of $\mathcal{G}$ with undirected edges, we obtain the \emph{skeleton} $\mathcal{G}^{s}$ of $\mathcal{G}$. We write $v \multimap w$ whenever $(v, w) \in E$, that is $\mathcal{G}$ contains either the directed edge $v \rightarrow w$ or the undirected edge $v \edge w$. A sequence of distinct vertices $v_{1}, \ldots, v_{k} \in V$, $k \ge 2$, is called a \emph{path} from $v_{1}$ to $v_{k}$ if $\mathcal{G}$ contains $v_{i} \multimap v_{i + 1}$ for all $i \in \lbrace 1, \ldots, k - 1 \rbrace$. A path from $v_{1}$ to $v_{k}$ is called \emph{directed} if at least one of the connecting edges is directed. We call a path from $v_{1}$ to $v_{k}$ a \emph{cycle} if $v_{1} = v_{k}$. In particular, we call a directed path from $v_{1}$ to $v_{k}$ a \emph{directed cycle} if $v_{1} = v_{k}$. A graph without directed cycles is called a \emph{chain graph} (CG). A CG containing only directed edges is known as a \emph{directed acyclic graph} (DAG). We define the \emph{adjacency set} of a vertex $v \in V$ as $\ad(v) \coloneqq \bigl\lbrace w \in V \, \big\vert \, (v, w) \in E \text{ or } (w, v) \in E \bigr\rbrace$. If $w \notin \ad(v)$, we say that $v$ and $w$ are \emph{non-adjacent}. A triple of vertices $(u, v, w)$ is called a \emph{v-structure} if $\mathcal{G}$ contains $u \rightarrow v \leftarrow w$ and if $u$ and $w$ are non-adjacent.

\bigskip

Now let $\mathcal{G}$ be a DAG. The \emph{moral graph} $\mathcal{G}^{m}$ of $\mathcal{G}$ is defined as the skeleton of the graph obtained from $\mathcal{G}$ by introducing an undirected edge $u \edge w$ whenever $\mathcal{G}$ contains a v-structure $(u, v, w)$ for $u, v, w \in V$. Since all edges of $\mathcal{G}$ are directed, we can speak of paths instead of directed paths. For $v \in V$, we call $\pa(v) \coloneqq \bigl\lbrace w \in V \, \big\vert \, \mathcal{G} \text{ contains } w \rightarrow v \bigr\rbrace$ the \emph{parents} of $v$, $\an(v) \coloneqq \bigl\lbrace w \in V \, \big\vert \, \mathcal{G} \text{ contains a path from } w \text{ to } v \bigr\rbrace$ the \emph{ancestors} of $v$, $\de(v) \coloneqq \bigl\lbrace w \in V \, \big\vert \, \mathcal{G} \text{ contains a path from } v \text{ to } w \bigr\rbrace$ the \emph{descendants} of $v$, and $\nd(v) \coloneqq V \setminus \bigl( \lbrace v \rbrace \cup \de(v) \bigl)$ the \emph{non-descendants} of $v$. A set $I \subseteq V$ is called \emph{ancestral} if $\pa(v) \subseteq I$ for all $v \in I$. The smallest ancestral set containing $I$ is denoted by $\An(I)$. As is readily verified, $\An(I) = I \cup \bigcup_{v \in I} \an(v)$. The graph $\mathcal{G}_{I} = \bigl( I, E \cap (I \times I) \bigr)$ is called the \emph{subgraph} of $\mathcal{G}$ induced by $I$. A bijection $v_{\bullet} \colon \bigl\lbrace 1, \ldots, \abs{V} \bigr\rbrace \rightarrow V$, $i \mapsto v_{i}$, satisfying $i < j$ whenever $\mathcal{G}$ contains $v_{i} \rightarrow v_{j}$ for some $i, j \le \abs{V}$ is called a \emph{well-ordering} of $\mathcal{G}$. Note that in a well-ordered DAG the set $\lbrace v_{1}, \ldots, v_{k} \rbrace$ is ancestral for all $k \le \abs{V}$.

\bigskip

Finally, let $\mathcal{G}$ be a UG and let $I, J, K \subseteq V$ be pairwise disjoint. A path from $I$ to $J$ is a path from a vertex $v \in I$ to a vertex $w \in J$. We say that $K$ \emph{separates} $I$ from $J$ in $\mathcal{G}$, and write $\csep{I}{J}{K}{\mathcal{G}}$, if every path from $I$ to $J$ contains a vertex in $K$. In particular, we write $\csep{I}{J}{\emptyset}{\mathcal{G}}$, or shortly $\sep{I}{J}{\mathcal{G}}$, if there exists no path between $I$ and $J$. We call $\mathcal{G}$ \emph{connected} if for every distinct $v, w \in V$ there is a path from $v$ to $w$. A connected UG without cycles is a \emph{tree}. If there is a vertex $w \in V$ such that $\ad(w) = V \setminus \lbrace w \rbrace$ and $\ad(v) = \lbrace w \rbrace$ for all $v \in V \setminus \lbrace w \rbrace$, that is all vertices are solely adjacent to $w$, then $\mathcal{G}$ is called a \emph{star} and $w$ is called its \emph{root vertex}. Note that above terminology is not used consistently throughout the literature. 

\subsubsection*{Markovian probability measures}

In graphical probability modelling, graphs are used to represent conditional independence properties of corresponding families of probability measures. Let $\mathcal{D} = (V, E)$ be a DAG on $d \coloneqq \abs{V}$ vertices and let $P$ be a probability measure on $\Real^{d}$. Moreover, let $\boldsymbol{X}$ be an $\Real^{d}$-valued random variable distributed as $P$. For $I \subseteq V$, we write $\boldsymbol{X}_{I} \coloneqq (X_{v})_{v \in I}$ and denote the corresponding $I$-margin of $P$ by $P_{I}$. If $I = \lbrace v \rbrace$ for some $v \in V$, we write $X_{v}$ and $P_{v}$ instead of $X_{\lbrace v \rbrace}$ and $P_{\lbrace v \rbrace}$. Furthermore, we write $\cind{\boldsymbol{X}_{I}}{\boldsymbol{X}_{J}}{\boldsymbol{X}_{K}}$ whenever $\boldsymbol{X}_{I}$ and $\boldsymbol{X}_{J}$ are \emph{conditionally independent} given $\boldsymbol{X}_{K}$ for pairwise disjoint sets $I, J, K \subseteq V$. By convention, $\cind{\boldsymbol{X}_{I}}{\boldsymbol{X}_{J}}{\boldsymbol{X}_{\emptyset}}$ is understood as $\boldsymbol{X}_{I} \ind \boldsymbol{X}_{J}$. $P$ is said to possess the \emph{local $\mathcal{D}$-Markov} property if
\begin{equation}\label{eq:localmarkov}
	\cind{X_{v}}{\boldsymbol{X}_{\nd(v) \setminus \pa(v)}}{\boldsymbol{X}_{\pa(v)}} \quad \text{for all } v \in V.
\end{equation}
Correspondingly, $P$ is said to possess the \emph{global $\mathcal{D}$-Markov} property if
\begin{equation}\label{eq:globalmarkov}
	\csep{I}{J}{K}{(\mathcal{D}_{\An(I \cup J \cup K)})^{m}} \ \Rightarrow \ \cind{\boldsymbol{X}_{I}}{\boldsymbol{X}_{J}}{\boldsymbol{X}_{K}} \quad \text{for all pairwise disjoint } I, J, K \subseteq V.
\end{equation}
Equations \eqref{eq:localmarkov} and \eqref{eq:globalmarkov} relate (conditional) independence properties of $P$ to graph separation properties of $\mathcal{D}$. Since $\ad(v) \cap \bigl( \nd(v) \setminus \pa(v) \bigr) = \emptyset$ for every $v \in V$, it can be easily seen that the conditional independence restrictions obtained from Equation \eqref{eq:localmarkov} correspond to missing edges in $\mathcal{D}$. One can show that $P$ has the local $\mathcal{D}$-Markov property if and only if $P$ has the global $\mathcal{D}$-Markov property, see \citet[p.\ $51$]{Lauritzen:1996}. A probability measure satisfying Equations \eqref{eq:localmarkov} and \eqref{eq:globalmarkov} is thus simply called \emph{$\mathcal{D}$-Markovian}. Despite the aforementioned equivalence, the lists of \emph{explicit} conditional independence restrictions obtained from Equations \eqref{eq:localmarkov} and \eqref{eq:globalmarkov} may, however, be of different lengths. Note that a $\mathcal{D}$-Markovian probability measure can exhibit further conditional independence properties apart from those represented by $\mathcal{D}$. If, however, $P$ exhibits no conditional independence properties other than those represented by $\mathcal{D}$, then $P$ is called \emph{faithful} to $\mathcal{D}$. Now let $P$ have Lebesgue-density $f$. One can show that $P$ is $\mathcal{D}$-Markovian if and only if $f$ has a so-called \emph{$\mathcal{D}$-recursive factorisation}, that is
\begin{equation*}
	f(\boldsymbol{x}) = \prod_{v \in V} f_{v \vert \! \pa(v)} \bigl( x_{v} \, \big\vert \, \boldsymbol{x}_{\pa(v)} \bigr) \quad \text{for all } \boldsymbol{x} = (x_{1}, \ldots, x_{d}) \in \Real^{d},
\end{equation*}
where $f_{v \vert \! \pa(v)}( \, \cdot \given \boldsymbol{x}_{\pa(v)})$ denotes the conditional probability density function (pdf) of $X_{v}$ given $\boldsymbol{X}_{\pa(v)} = \boldsymbol{x}_{\pa(v)}$, see again \citet[p.\ $51$]{Lauritzen:1996}. Note that there may be more than one DAG representing the same set of conditional independence restrictions. We call the set of DAGs representing the same conditional independence restrictions as $\mathcal{D}$ the \emph{Markov-equivalence class} of $\mathcal{D}$, and denote it by $[\mathcal{D}]$. Two DAGs $\mathcal{D}_{1} = (V, E_{1})$ and $\mathcal{D}_{2} = (V, E_{2})$ are called \emph{Markov equivalent} if $[\mathcal{D}_{1}] = [\mathcal{D}_{2}]$. By \citet{Verma.Pearl:1991}, $\mathcal{D}_{1}$ and $\mathcal{D}_{2}$ are Markov equivalent if and only if they have the same skeleton and the same v-structures. The Markov-equivalence class of $\mathcal{D}$ can be represented by a CG, the so-called \emph{essential graph} $\mathcal{D}^{e}$ associated with $[\mathcal{D}]$, which has the same skeleton as $\mathcal{D}$ and contains a directed edge $v \rightarrow w$ if and only if all members of $[\mathcal{D}]$ contain $v \rightarrow w$, see \citet{Andersson.Madigan.Perlman:1997}. A DAG in $[\mathcal{D}]$ can be obtained from $\mathcal{D}^{e}$ by directing all undirected edges of $\mathcal{D}^{e}$ such that no new v-structures and no directed cycles are introduced. Figure \ref{fig:dagessgraph} gives an example of a DAG on four vertices together with the essential graph associated with the corresponding Markov-equivalence class.

\begin{figure}[htb]
	\centering
		\begin{subfigure}[t]{.45\textwidth}
			\centering \includegraphics[width=.5\textwidth]{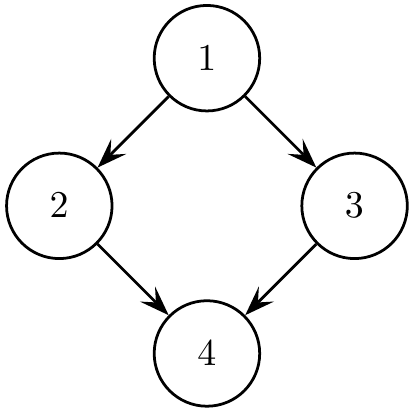}
		\end{subfigure}
		\begin{subfigure}[t]{.45\textwidth}
			\centering \includegraphics[width=.5\textwidth]{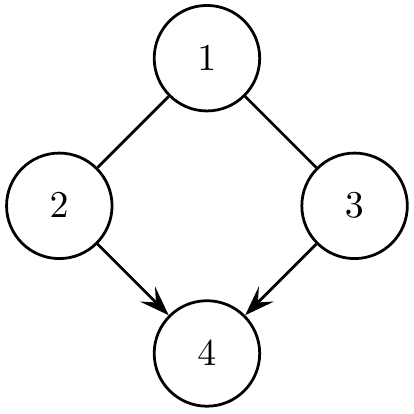}
		\end{subfigure}
	\caption{A DAG $\mathcal{D}$ (left) specifying the conditional independence restrictions $\cind{X_{1}}{X_{4}}{\boldsymbol{X}_{23}}$ and $\cind{X_{2}}{X_{3}}{X_{1}}$, and the essential graph $\mathcal{D}^{e}$ (right) associated with the corresponding Markov-equivalence class $[\mathcal{D}]$.}
	\label{fig:dagessgraph}
\end{figure}

\subsubsection*{Graphical models}

A \emph{Bayesian network} or \emph{(directed) graphical model} based on $\mathcal{D}$ is a family of $\mathcal{D}$-Markovian probability measures. A comprehensive introduction to graphical models, and Bayesian networks in particular, is found in \citet{Lauritzen:1996} and \citet{Cowell.Dawid.Lauritzen.Spiegelhalter:2003}, see also \citet{Pourret.Naim.Marcot:2008} for examples of applications. For lack of tractable continuous probability measures, statistical modelling with Bayesian networks has mostly been limited to multivariate discrete or normal distributions. \citet{Kurowicka.Cooke:2005} therefore used copulas to derive a rich and tractable class of continuous Bayesian networks, which we will investigate in Section \ref{sec:pcbn}.

\section{Vine copula models}\label{sec:vine}

A $d$-variate \emph{copula}, $d \in \Nat$, is a cumulative distribution function (cdf) on $[0, 1]^{d}$ such that all univariate marginals are uniform on the interval $[0, 1]$. By \emph{Sklar's theorem} \citep{Sklar:1959}, every cdf $F$ on $\Real^{d}$ with marginals $F_{1}, \ldots, F_{d}$ can be written as
\begin{equation*}
	F(\boldsymbol{x}) = C \bigl( F_{1}(x_{1}), \ldots, F_{d}(x_{d}) \bigr), \quad \boldsymbol{x} = (x_{1}, \ldots, x_{d}) \in \Real^{d},
\end{equation*}
for some suitable copula $C$. If $F$ is absolutely continuous and $F_{1}, \ldots, F_{d}$ are strictly increasing, a similar relationship holds for the pdf $f$ of $F$, namely
\begin{equation*}
	f(\boldsymbol{x}) = c \bigl( F_{1}(x_{1}), \ldots, F_{d}(x_{d}) \bigr) \prod_{i = 1}^{d} f_{i}(x_{i}), \quad \boldsymbol{x} = (x_{1}, \ldots, x_{d}) \in \Real^{d},
\end{equation*}
where the \emph{copula density} $c$ is uniquely determined. A comprehensive introduction to copulas is found in \citet{Joe:1997} and \citet{Nelsen:2006}.

\subsection{Pair-copula constructions and regular vines}

While in recent years a vast catalogue of bivariate copula families (also known as \emph{pair-copula} families) has accumulated in the literature, many of these bivariate families have no straightforward multivariate extension. Based on \citet{Joe:1996}, \citet{Bedford.Cooke:2001, Bedford.Cooke:2002} introduced a rich and flexible class of multivariate copulas that uses bivariate (conditional) copulas as building blocks only. The corresponding decomposition of a multivariate copula into bivariate copulas is called a \emph{pair-copula construction} (PCC). The most widely researched copulas arising from PCCs are the \emph{vine copulas}. These vine copulas admit a graphical representation called a \emph{regular vine} (R-vine), which essentially consists of a sequence of trees, each edge of which is associated with a certain pair copula in the corresponding PCC. More precisely, let $V \neq \emptyset$ be a finite set and let $d \coloneqq \abs{V}$. An R-vine on $V$ is a sequence $\mathcal{V} \coloneqq (T_{1}, \ldots, T_{d - 1})$ of trees $T_{1} = (V_{1}, E_{1}), \ldots, T_{d - 1} = (V_{d - 1}, E_{d  - 1})$ such that $V_{1} = V$ and $V_{i} = E_{i - 1}$ for $i \ge 2$, that is the vertices of tree $T_{i}$ are the edges of tree $T_{i - 1}$. We here represent an edge $v \edge w$ in tree $T_{i}$, $i \in \lbrace 1, \ldots, d - 1 \rbrace$, by the doubleton $\lbrace v, w \rbrace$ instead of by the pairs $(v, w)$ and $(w, v)$, that is $E_{i} \subseteq \bigl\lbrace \lbrace v, w \rbrace \, \big\vert \, v \neq w \in V_{i} \bigr\rbrace$. Moreover, every tree $T_{i}$, $i \ge 2$, of $\mathcal{V}$ has to satisfy a \emph{proximity condition} requiring that $\abs{v \vartriangle w} = 2$ for every edge $\lbrace v, w \rbrace \in E_{i}$, where $u \vartriangle v = (u \cup v) \setminus (u \cap v)$. Two vertices in tree $T_{i}$, $i \ge 2$, can hence only be adjacent if the corresponding edges in tree $T_{i - 1}$ share a common vertex. Last, every edge $\lbrace v, w \rbrace \in E \coloneqq E_{1} \cup \cdots \cup E_{d - 1}$ carries a label $v \vartriangle w \, \vert \, v \cap w$ representing the (conditional) pair copula $C_{v \vartriangle w \vert v \cap w}$, where $v \vartriangle w \, \vert \,\emptyset$ is conveniently replaced by $v \vartriangle w$. Instead of $C_{v \vartriangle w \vert v \cap w}$ we also write $C_{v_{\vartriangle}, w_{\vartriangle} \vert v \cap w}$, where $v_{\vartriangle} \coloneqq v \setminus (v \cap w)$ and $w_{\vartriangle} \coloneqq w \setminus (v \cap w)$. The pdf $f$ of a $d$-variate probability measure with univariate marginals $F_{v}$, $v \in V$, and copula $C_{\mathcal{V}}$ corresponding to $\mathcal{V}$ then takes the form
\begin{equation}\label{eq:rvinepdf}
	f(\boldsymbol{x}) = \prod_{\lbrace v, w \rbrace \in E} c_{v_{\vartriangle}, w_{\vartriangle} \vert v \cap w} \bigl( F_{v_{\vartriangle} \vert v \cap w}(x_{v_{\vartriangle}} \, \vert \, \boldsymbol{x}_{v \cap w}), F_{w_{\vartriangle} \vert v \cap w}(x_{w_{\vartriangle}} \, \vert \, \boldsymbol{x}_{v \cap w}) \, \big\vert \, \boldsymbol{x}_{v \cap w} \bigr) \prod_{v \in V} f_{v}(x_{v}),
\end{equation}
where $\boldsymbol{x} = (x_{v})_{v \in V} \in \Real^{d}$. Note that---similar to DAGs---the vertices in the first tree of $\mathcal{V}$ represent the univariate margins of $C_{\mathcal{V}}$. In contrast to DAGs, however, $\mathcal{V}$ does not have an interpretation in terms of Markov properties of $C_{\mathcal{V}}$. An example of an R-vine representing a five-variate vine copula is given in Figure \ref{fig:rvine}.

\bigskip

\begin{figure}[htb]
	\centering
	\includegraphics[width=.75\textwidth]{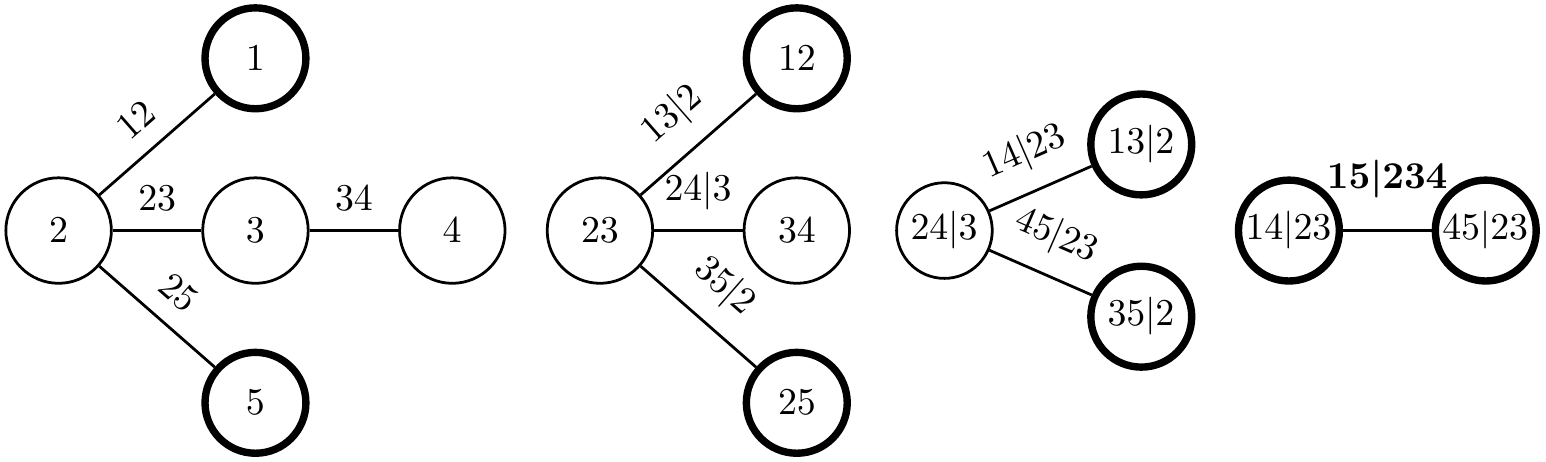}
	\caption{An R-vine specifying the pair copulas $C_{12}$, $C_{23}$, $C_{25}$, $C_{34}$, $C_{13 \vert 2}$, $C_{24 \vert 3}$, $C_{35 \vert 2}$, $C_{14 \vert 23}$, $C_{45 \vert 23}$, and $C_{15 \vert 234}$ (edge labels). Boundaries of vertices including either $1$ or $5$ appear in bold, see Section \ref{sec:pc}.}
	\label{fig:rvine}
\end{figure}

For every $K \subseteq V$ and $v \in K$, we define $K_{-v} \coloneqq K \setminus \lbrace v \rbrace$.  The conditional cdfs in Equation \eqref{eq:rvinepdf} can be evaluated tree-by-tree using a recursive formula derived in \citet{Joe:1996}, which says that for every $v \in V$, every $K \subseteq V_{-v}$, and an arbitrary $w \in K$
\begin{equation}\label{eq:condcdf}
	F_{v \vert K}(x_{v} \, \vert \, \boldsymbol{x}_{K}) = \frac{\partial C_{v, w \vert K_{-w}}{\bigl( F_{v \vert K_{-w}}(x_{v} \, \vert \, \boldsymbol{x}_{K_{-w}}), F_{w \vert K_{-w}}(x_{w} \, \vert \, \boldsymbol{x}_{K_{-w}}) \, \big| \, \boldsymbol{x}_{K_{-w}} \bigr)}}{\partial F_{w \vert K_{-w}}(x_{w} \, \vert \, \boldsymbol{x}_{K_{-w}})}.
\end{equation}
An iterative algorithm for evaluating the pdf in Equation \eqref{eq:rvinepdf} under a simplifying assumption of constant conditional copulas introduced below is given in \citet{Dissmann.Brechmann.Czado.Kurowicka:2012}. The first partial derivatives of a pair copula $C_{v, w}$ are also known as \emph{h-functions}. We write
\begin{equation*}
	h_{\underline{v}, w}(u_{v}, u_{w}) \coloneqq \frac{\partial C_{v, w}(u_{v}, u_{w})}{\partial u_{w}} \quad \text{and} \quad h_{v, \underline{w}}(u_{v}, u_{w}) \coloneqq \frac{\partial C_{v, w}(u_{v}, u_{w})}{\partial u_{v}}, \quad (u_{v}, u_{w}) \in [0, 1]^{2}.
\end{equation*}
Many popular pair-copula families exhibit closed-form expressions for these h-functions, see for instance \citet{Aas.Czado.Frigessi.Bakken:2009}. Note that by Equation \eqref{eq:condcdf} we have
\begin{equation*}
	h_{\underline{v}, w} \bigl( F_{v}(x_{v}), F_{w}(x_{w}) \bigr) = F_{v \vert w}(x_{v} \, \vert \, x_{w}) \quad \text{and} \quad h_{v, \underline{w}} \bigl( F_{v}(x_{v}), F_{w}(x_{w}) \bigr) = F_{w \vert v}(x_{w} \, \vert \, x_{v}) ,
\end{equation*}
where $(x_{v}, x_{w}) \in \Real^{2}$. Hence, we can extend the notion of h-functions to conditional pair copulas, and express the right hand side of Equation \eqref{eq:condcdf} by
\begin{equation*}
	h_{\underline{v}, w \vert K_{-w}} \! \bigl( F_{v \vert K_{-w}}(x_{v} \, \vert \, \boldsymbol{x}_{K_{-w}}), F_{w \vert K_{-w}}(x_{w} \, \vert \, \boldsymbol{x}_{K_{-w}}) \, \big| \, \boldsymbol{x}_{K_{-w}} \bigr).
\end{equation*}
Assume $p \coloneqq \abs{K} \ge 1$, and write $K = \lbrace w_{1}, \ldots, w_{p} \rbrace$ such that $w_{i} \neq w_{j}$ for $i \neq j$. We define $K_{-i} \coloneqq \lbrace w_{i + 1}, \ldots, w_{p} \rbrace$ for every $i \in \lbrace 1, \ldots, p \rbrace$. Observing that $f_{v \vert K}(x_{v} \, \vert \, \boldsymbol{x}_{K}) = \frac{\diff{}}{\diff x_{v}} F_{v \vert K}(x_{v} \, \vert \, \boldsymbol{x}_{K})$, we obtain by the chain rule of differentiation
\begin{equation}\label{eq:chainrule}
	f_{v \vert K}(x_{v} \, \vert \, \boldsymbol{x}_{K}) = f_{v}(x_{v}) \prod_{i = 1}^{p} c_{v, w_{i} \vert K_{-i}}{\bigl( F_{v \vert K_{-i}}(x_{v} \, \vert \, \boldsymbol{x}_{K_{-i}}), F_{w_{i} \vert K_{-i}}(x_{w_{i}} \, \vert \, \boldsymbol{x}_{K_{-i}}) \, \big| \, \boldsymbol{x}_{K_{-i}} \bigr)}.
\end{equation}

\subsection{ML estimation and model selection in vine copula models}\label{sec:vinestat}

A \emph{vine copula model} is a family of vine copulas together with families of univariate marginals. Maximum likelihood (ML) estimation in vine copula models was first considered in \citet{Aas.Czado.Frigessi.Bakken:2009}. The findings therein were, however, restricted to vine copula models represented by C- and D-vines. A \emph{C-vine} is an R-vine whose trees are all stars. Conversely, an R-vine is called a \emph{D-vine} if all vertices in tree $T_{1}$ are adjacent to at most two other vertices. ML estimation in vine copula models based on general R-vines was considered in \citet{Dissmann.Brechmann.Czado.Kurowicka:2012}.

\bigskip

Let $\mathcal{V}$ be an R-vine on $V$ with edge set $E$, and let $C_{v_{\vartriangle}, w_{\vartriangle} \vert v \cap w}(\, \cdot \, , \cdot \, ; \boldsymbol{\theta}_{v_{\vartriangle}, w_{\vartriangle} \vert v \cap w})$, $\lbrace v, w \rbrace \in E$, be given (conditional) pair copulas with joint parameter vector $\boldsymbol{\theta} \coloneqq (\boldsymbol{\theta}_{v_{\vartriangle}, w_{\vartriangle} \vert v \cap w})_{\lbrace v, w \rbrace \in E} \in \boldsymbol{\Theta}$. We denote the corresponding vine copula family by $\lbrace C_{\mathcal{V}, \boldsymbol{\theta}} \, \vert \, \boldsymbol{\theta} \in \boldsymbol{\Theta} \rbrace$. Note that we dropped the values $\boldsymbol{x}_{u \cap v}$ of the conditioning variables from the pair copulas $C_{v_{\vartriangle}, w_{\vartriangle} \vert v \cap w}$, thus assuming that the corresponding copula family and parameter vector $\boldsymbol{\theta}_{v_{\vartriangle}, w_{\vartriangle} \vert v \cap w}$ remain constant for all $\boldsymbol{x}_{u \cap v} \in \Real^{\abs{u \cap v}}$. This \emph{simplifying assumption} is made for computational convenience and has become common practice in likelihood inference for vine copula models, see \citet{Haff.Aas.Frigessi:2010} and \citet{Acar.Genest.Neslehova:2012} for a critical assessment. Furthermore, let $\boldsymbol{u} = \left( \boldsymbol{u}^{1}, \ldots, \boldsymbol{u}^{n} \right)$, $n \in \Nat$, be a realisation of a sample of i.i.d.\ observations $\boldsymbol{U}^{1}, \ldots, \boldsymbol{U}^{n}$ from a random variable $\boldsymbol{U}$ on $[0, 1]^{d}$ with copula family $\lbrace C_{\mathcal{V}, \boldsymbol{\theta}} \, \vert \, \boldsymbol{\theta} \in \boldsymbol{\Theta} \rbrace$ and uniform univariate margins. Equation \eqref{eq:rvinepdf} yields the log-likelihood function
\begin{equation}\label{eq:rvinell}
	l(\boldsymbol{\theta}; \boldsymbol{u}) = \sum_{k = 1}^{n} \sum_{\lbrace v, w \rbrace \in E} \!\! \log c_{v_{\vartriangle}, w_{\vartriangle} \vert v \cap w} \Bigl( F_{v_{\vartriangle} \vert v \cap w} \bigl( u_{v_{\vartriangle}}^{k} \, \big\vert \, \boldsymbol{u}_{v \cap w}^{k}; \boldsymbol{\theta} \bigr), F_{w_{\vartriangle} \vert v \cap w} \bigl( u_{w_{\vartriangle}}^{k} \, \big\vert \, \boldsymbol{u}_{v \cap w}^{k}; \boldsymbol{\theta} \bigr); \boldsymbol{\theta} \Bigr).
\end{equation}
The restriction to uniform univariate margins is made for computational convenience, see~below.

\subsubsection*{ML estimation} \enlargethispage{\baselineskip}% Remove!

Since a joint estimation of the parameters of the univariate marginal distributions and the copula can become computationally demanding in high dimensions, a two-step estimation approach known as the \emph{inference functions for margins} method \citep{Joe.Xu:1996} is frequently applied. First, the marginal parameters are estimated and second, given the estimates of the marginal parameters, the copula parameters are inferred. In a similar vein, \citet{Genest.Ghoudi.Rivest:1995} proposed a semiparametric approach in which the empirical cdf is used to transform the univariate marginals to uniform $[0, 1]$ distributions before estimating the parameters of the copula model, see \citet{Kim.Silvapulle.Silvapulle:2007} for a comparison. ML estimation of the parameters in Equation \eqref{eq:rvinell} is frequently performed using a stepwise approach as first described in \citet{Aas.Czado.Frigessi.Bakken:2009}. In a first step, ML estimates of the parameters of each pair-copula family are computed separately. Due to the recursive structure of the log-likelihood function outlined above, this estimation step is carried out tree-by-tree. We refer to the obtained parameter estimates as \emph{sequential ML estimates}. In a second step, the full log-likelihood function is maximised jointly using the sequential ML estimates as starting values, yielding the so-called \emph{joint ML estimates} $\boldsymbol{\widehat{\theta}}_{v_{\vartriangle}, w_{\vartriangle} \vert v \cap w}$, $\lbrace v, w \rbrace \in E$. Large and small sample applications of the stepwise estimation procedure have shown that the sequential ML estimates also provide a good approximation of their joint counterparts, see \citet{Haff:2012a, Haff:2012b} for consistency results and a simulation study. One might hence consider omitting the second estimation step in a given situation to reduce computational complexity.

\subsubsection*{Model selection} \enlargethispage{\baselineskip}% Remove!

Model selection for vine copula models comprises an estimation of the R-vine $\mathcal{V}$ and a selection of the pair-copula families for $C_{v_{\vartriangle}, w_{\vartriangle} \vert v \cap w}$, $\lbrace v, w \rbrace \in E$. Given $\mathcal{V}$, the latter task of selecting pair-copula families can be performed tree-by-tree, choosing for each edge $\lbrace v, w \rbrace \in E$ the one pair-copula family among a given set of candidate families that optimises a given selection criterion like Akaike's information criterion (AIC) or the Bayesian information criterion (BIC). \citet{Dissmann.Brechmann.Czado.Kurowicka:2012} presented a greedy-type algorithm for the estimation of $\mathcal{V}$, which estimates the trees $T_{1}, \ldots, T_{d - 2}$ sequentially, that is again tree-by-tree. Note that estimating tree $T_{d - 2}$ also fixes tree $T_{d - 1}$. Structure estimation for tree $T_{i} = (V_{i}, E_{i})$, $i \in \lbrace 1, \ldots, d - 2 \rbrace$, is carried out in three steps. In a first step, a weight $\omega_{v, w}$ is assigned to every pair of vertices $v, w \in V_{i}$ with $\abs{v \vartriangle w} = 2$. Suitable weights given the data are, for instance, the absolute values of estimates of Kendall's $\tau$, or AIC or BIC values of selected pair-copula families with estimated parameters. In a second step, $T_{i}$ is set to be a tree on $V_{i}$ optimising the sum of edge weights $\sum_{\lbrace v, w \rbrace \in E_{i}} \omega_{v, w}$, where $\abs{v \vartriangle w} = 2$ for all $\lbrace v, w \rbrace \in E_{i}$ to ensure the proximity condition. Such an \emph{optimal spanning tree} can be found using the algorithms by \citet{Kruskal:1956} or \citet{Prim:1957}. In a last step, a pair-copula family is assigned to each edge $\lbrace v, w \rbrace \in E_{i}$, as described above, and an ML estimate of the corresponding parameter(s) is computed. This last step may have already been performed when computing the edge weights $\omega_{v, w}$. Note that due to the greedy nature of the algorithm, the resulting R-vine need not optimise the sum of all edge weights $\sum_{\lbrace v, w \rbrace \in E} \omega_{v, w}$. The search for optimal spanning trees reduces to a search for root vertices when only considering C-vines instead of the more general R-vines, cf.\ \citet{Czado.Schepsmeier.Min:2012}. Since a D-vine is completely determined by tree $T_{1}$, only one tree has to be specified when restricting the class of R-vines to D-vines. Due to the particular structure of D-vines, however, finding tree $T_{1}$ by the above method leads to a \emph{travelling salesman problem} (TSP) \citep{Applegate.Bixby.Chvatal.Cook:2007}, which is NP-hard. \citet{Kurowicka:2011} proposed an alternative structure selection algorithm, in which $\mathcal{V}$ is built in reverse order from tree $T_{d - 1}$ to tree $T_{1}$ using partial correlation estimates as weights. Bayesian approaches to structure estimation have been considered in \citet{Smith.Min.Almeida.Czado:2010}, \citet{Min.Czado:2011}, and \citet{Gruber.Czado.Stoeber:2012}. A more detailed exposition of vine copula models is found in \citet{Kurowicka.Joe:2011}. Implementations of model selection and ML estimation procedures for vine copula models are available in the R package \texttt{VineCopula} \nolinebreak[4]% Remove!
\citep{Schepsmeier.Stoeber.Brechmann:2012}. \nolinebreak[4]% Remove!

\bigskip

The construction of a $d$-variate vine copula model requires the specification of $\binom{d}{2}$ pair-copula families, a number growing quadratically in $d$. The actual number of decisions to make in practical applications may, however, be lower if we happen to discover (conditional) independences in the analysed data. In that case, the corresponding pair copulas are set to be independence copulas. Since above structure estimation algorithm is based on the idea of modelling strongest dependences in the first trees, \citet{Brechmann.Czado.Aas:2012} proposed to set all pair copulas in the later trees to independence copulas, which leads to so-called \emph{truncated R-vines}. Instead of leaving the detection of (conditional) independences to chance, one may, however, consider modelling these independences in the first place to obtain more parsimonious models. Unfortunately, the construction of vine copula models satisfying pre-specified conditional independence restrictions is a hard problem in general. A class of models suited for this task are the Bayesian networks discussed in Section \ref{sec:bn}. \citet{Kurowicka.Cooke:2005} hence joined graphical and copula modelling to introduce PCCs for Bayesian networks, which we will investigate in the next section.

\section{Pair-copula Bayesian networks (PCBNs)}\label{sec:pcbn}

Let $\mathcal{D} = (V, E)$ be a DAG, and let $P$ be an absolutely continuous $\mathcal{D}$-Markovian probability measure on $\Real^{d}$, $d \coloneqq \abs{V}$, with strictly increasing univariate marginal cdfs. Moreover, let $w_{v} \colon \bigl\lbrace 1, \ldots, \abs{\pa(v)} \bigr\rbrace \rightarrow \pa(v)$, $i \mapsto w_{i} \coloneqq w_{v}(i)$, be a bijection for every $v \in V$ with $\abs{\pa(v)} \ge 1$. We introduce a total order $<_{v}$ on $\pa(v)$ for every $v \in V$ such that whenever $\abs{\pa(v)} \ge 1$ we have $w_{i} <_{v} w_{j}$ if and only if $i < j$ for all $i, j \in \bigl\lbrace 1, \ldots, \abs{\pa(v)} \bigr\rbrace$. Note that there are $\abs{\pa(v)} !$ permutations of $\pa(v)$ (up to isomorphism). We call $\mathcal{O} \coloneqq \lbrace <_{v} \! \vert \, v \in V \rbrace$ a set of \emph{parent orderings} for $\mathcal{D}$. For every $v \in V$ and $w \in \pa(v)$, we set
\begin{equation*}
	\pa(v; w) \coloneqq \bigl\lbrace u \in \pa(v) \, \big\vert \, u <_{v} w \bigr\rbrace = \bigl\lbrace w_{i} \in \pa(v) \, \big\vert \, i < w_{v}^{-1}(w) \bigr\rbrace.
\end{equation*}
By Sklar's theorem, we know that the cdf of $P$ can be uniquely decomposed into the univariate marginals $F_{1}, \ldots, F_{d}$ and a copula $C$. \citet{Bauer.Czado.Klein:2012} have shown that $C$ can be further decomposed into the (conditional) pair copulas $C_{v, w \vert \! \pa(v; w)}$, $v \in V$, $w \in \pa(v)$, which yields a PCC for $C$ in which each (conditional) pair copula corresponds to exactly one edge $w \rightarrow v$ in $\mathcal{D}$. The pdf $f$ of $P$ can hence be written as
\begin{equation}\label{eq:dagpdf}
	f(\boldsymbol{x}) = \prod_{v \in V} \, f_{v}(x_{v}) \!\! \prod_{w \in \pa(v)} \!\! c_{v, w \vert \! \pa(v; w)}{\bigl( F_{v \vert \! \pa(v; w)}(x_{v} \, \vert \, \boldsymbol{x}_{\pa(v; w)}), F_{w \vert \! \pa(v; w)}(x_{w} \, \vert \, \boldsymbol{x}_{\pa(v; w)}) \, \big\vert \, \boldsymbol{x}_{\pa(v; w)} \bigr)},
\end{equation}
where $\boldsymbol{x} = (x_{v})_{v \in V} \in \Real^{d}$. As an example consider the DAG in Figure \ref{fig:dagessgraph} with ordering $2 <_{4} 3$ of $\pa(4) = \lbrace 2, 3 \rbrace$. Equation \eqref{eq:dagpdf} yields
\begin{align*}
		f(\boldsymbol{x}) & = f_{1}(x_{1}) \cdots f_{4}(x_{4}) \cdot c_{21} \bigl( F_{2}(x_{2}), F_{1}(x_{1}) \bigr) \cdot c_{31} \bigl( F_{3}(x_{3}), F_{1}(x_{1}) \bigr) \cdot c_{42} \bigl( F_{4}(x_{4}), F_{2}(x_{2}) \bigr)\\
		& \quad \cdot c_{43 \vert 2} \bigl( F_{4 \vert 2}(x_{4} \, \vert \, x_{2}), F_{3 \vert 2}(x_{3} \, \vert \, x_{2}) \, \big\vert \, x_{2} \bigr), \quad \boldsymbol{x} = (x_{1}, \ldots, x_{4}) \in \Real^{4},
\end{align*}
where $F_{4 \vert 2}(x_{4} \, \vert \, x_{2}) = h_{\underline{4}2} \bigl( F_{4}(x_{4}), F_{2}(x_{2}) \bigr)$ by Equation \eqref{eq:condcdf} and
\begin{equation*}
	F_{3 \vert 2}(x_{3} \, \vert \, x_{2}) = \int_{0}^{1} c_{21} \bigl( F_{2}(x_{2}), u_{1} \bigr) \, h_{\underline{3}1} \bigl( F_{3}(x_{3}), u_{1} \bigr) \diff u_{1} ,
\end{equation*}
see \citet{Bauer.Czado.Klein:2012} for details. If we instead choose the ordering $3 <_{4} 2$ for $\pa(4) = \lbrace 2, 3 \rbrace$, we obtain the same decomposition as above with the roles of vertices $2$ and $3$ interchanged. Due to the appearing integral, the pair-copula decomposition in the example cannot be represented by an R-vine. There are, however, DAG PCCs representable by R-vines, see for instance \citet{Bauer.Czado.Klein:2012} for a four-variate DAG PCC which coincides with a D-vine PCC.

\subsection{Evaluating conditional cdfs in PCBNs}

Similar to vine copulas, the challenge in Equation \eqref{eq:dagpdf} lies in the evaluation of the conditional cdfs. Assume without loss of generality that $\mathcal{D}$ is well-ordered. Let $v \in V$ and let $J \subseteq V \setminus \lbrace v \rbrace$ be non-empty. We will now derive a pair-copula decomposition for the conditional cdf $F_{v \vert J}(\, \cdot \, \vert \, \boldsymbol{x}_{J})$. We begin by exploiting the (conditional) independence restrictions represented by $\mathcal{D}$. To this end, consider the moral graph $\mathcal{G} \coloneqq (\mathcal{D}_{\An(\lbrace v \rbrace \cup J)})^{m}$. If $\csep{\lbrace v \rbrace}{I}{(J \setminus I)}{\mathcal{G}}$ for some non-empty $I \subseteq J$, then the global $\mathcal{D}$-Markov property in Equation \eqref{eq:globalmarkov} yields with $K \coloneqq J \setminus I$
\begin{equation*}
	f_{v \vert J}(x_{v} \, \vert \, \boldsymbol{x}_{J}) = \frac{f_{\lbrace v \rbrace \cup J}(\boldsymbol{x}_{\lbrace v \rbrace \cup J})}{f_{J}(\boldsymbol{x}_{J})} = \frac{f_{v \vert K}(x_{v} \, \vert \, \boldsymbol{x}_{K}) \, f_{I \vert K}(\boldsymbol{x}_{I} \, \vert \, \boldsymbol{x}_{K}) \, f_{K}(\boldsymbol{x}_{K})}{f_{I \vert K}(\boldsymbol{x}_{I} \, \vert \, \boldsymbol{x}_{K}) \, f_{K}(\boldsymbol{x}_{K})} = f_{v \vert K}(x_{v} \, \vert \, \boldsymbol{x}_{K}),
\end{equation*}
where by convention $f_{W \vert \emptyset}(\boldsymbol{x}_{W} \, \vert \, \boldsymbol{x}_{\emptyset}) \coloneqq f_{W}(\boldsymbol{x}_{W})$ for every $W \subseteq V$, and $f_{\emptyset}(\boldsymbol{x}_{\emptyset}) \coloneqq 1$. Thus, $F_{v \vert J}(\, \cdot \, \vert \, \boldsymbol{x}_{ J}) = F_{v \vert K}(\, \cdot \, \vert \, \boldsymbol{x}_{K})$, and we can continue with the conditioning set $K$. The case $K = \emptyset$ is trivial. Assume $K \neq \emptyset$. Observing that
\begin{equation}\label{eq:dagcondcdf}
	F_{v \vert K}(y \, \vert \, \boldsymbol{x}_{K}) = \frac{\int_{-\infty}^{y} f_{\lbrace v \rbrace \cup K}(\boldsymbol{x}_{\lbrace v \rbrace \cup K}) \diff x_{v}}{f_{K}(\boldsymbol{x}_{K})},
\end{equation}
we next need to find pair-copula decompositions for $f_{\lbrace v \rbrace \cup K}$ and $f_{K}$.

\subsubsection*{Pair-copula decompositions for marginal pdfs}

More generally, let $I \subseteq V$ be non-empty and consider the (marginal) pdf $f_{I}$. For every $v \in I$, we set $I_{-v} \coloneqq I \setminus \lbrace v \rbrace$ and obtain the following lemma.

\begin{lem}\label{lem:dagpdf}
Let $\mathcal{D} = (V, E)$ be a well-ordered DAG on $d \coloneqq \abs{V}$ vertices, and let $P$ be an absolutely continuous $\mathcal{D}$-Markovian probability measure on $\Real^{d}$ with pdf $f$. Let $I \subseteq V$ be non-empty and let $v$ denote the maximal vertex in $I$ by the well-ordering of $\mathcal{D}$. Moreover, define $S_{v} \coloneqq \bigl\lbrace u \in \pa(v) \, \big\vert \, \sep{\lbrace u \rbrace}{I_{-v}}{(\mathcal{D}_{\An(\lbrace u \rbrace \cup I_{-v})})^{m}} \bigr\rbrace$ and
\begin{equation*}
	W_{v} \coloneqq 
		\begin{cases}
			\emptyset & \text{if } I_{-v} = \emptyset \text{ or } S_{v} = \pa(v),\\
			\lbrace w_{1} \rbrace \cup \pa(v; w_{1}) & \text{if } I_{-v} \subseteq \pa(v) \text{ and } I_{-v} \neq \emptyset,\\
			\lbrace w_{2} \rbrace \cup \pa(v; w_{2}) & \text{else},
		\end{cases}
\end{equation*}
where $w_{1}$ and $w_{2}$ denote the maximal vertices in $I_{-v}$ and $\pa(v) \setminus S_{v}$, respectively, by a given parent ordering $<_{v}$. Then for all $\boldsymbol{x}_{I} = (x_{u})_{u \in I} \in R^{\abs{I}}$,
\begin{equation}\label{eq:dagpdflem}
	f_{I}(\boldsymbol{x}_{I}) = \int_{\Real^{\abs{W_{v} \setminus I}}} f_{v \vert W_{v}}(x_{v} \, \vert \, \boldsymbol{x}_{W_{v}}) \, f_{W_{v} \cup I_{-v}}(\boldsymbol{x}_{W_{v} \cup I_{-v}}) \diff \boldsymbol{x}_{W_{v} \setminus I} .
\end{equation}
\end{lem}

Note that by convention, $\int_{\Real^{0}} g(\boldsymbol{x}) \diff \boldsymbol{x}_{\emptyset} \coloneqq g(\boldsymbol{x})$ for every integrable function $g \colon \Real^{k} \rightarrow \Real$, $k \in \Nat$. Also note that the parent ordering $<_{v}$ need not concur with the well-ordering of $\mathcal{D}$.

\begin{proof}
As can be seen from the definition of $W_{v}$, the decomposition of $f_{I}$ in the lemma's claim depends on the relation between the sets $I_{-v}$ and $\pa(v)$. Assume first that $I_{-v} = \emptyset$. Then $f_{I} = f_{v}$ and the claim is trivial.

\bigskip

Next, assume $I_{-v} \neq \emptyset$ but $\pa(v) = \emptyset$. Then $S_{v} = \emptyset$. Since $v$ is maximal in $I$ by the well-ordering of $\mathcal{D}$, $v$ has no descendants in $I_{-v}$, and we have $\sep{\lbrace v \rbrace}{I_{-v}}{(\mathcal{D}_{\An(I)})^{m}}$. The global $\mathcal{D}$-Markov property thus yields $f_{I}(\boldsymbol{x}_{I}) = f_{v}(x_{v}) \, f_{I_{-v}}(\boldsymbol{x}_{I_{-v}})$, that is Equation \eqref{eq:dagpdflem} for $W_{v} \coloneqq \emptyset$.

\bigskip

From now on assume $I_{-v} \neq \emptyset$ and $\pa(v) \neq \emptyset$. The possible relations between $I_{-v}$ and $\pa(v)$ are illustrated in Figure \ref{fig:venn}. If $I_{-v} \subseteq \pa(v)$ (Figure \ref{fig:venn1}), we extend $I_{-v}$ to $W_{v} \coloneqq \lbrace w_{1} \rbrace \cup \pa(v; w_{1}) \supseteq I_{-v}$ and obtain as claimed
\begin{equation*}
	f_{I}(\boldsymbol{x}_{I}) = \int_{\Real^{\abs{W_{v} \setminus I}}} f_{\lbrace v \rbrace \cup W_{v}}(\boldsymbol{x}_{\lbrace v \rbrace \cup W_{v}}) \diff \boldsymbol{x}_{W_{v} \setminus I} = \int_{\Real^{\abs{W_{v} \setminus I}}} f_{v \vert W_{v}}(x_{v} \, \vert \, \boldsymbol{x}_{W_{v}}) \, f_{W_{v}}(\boldsymbol{x}_{W_{v}}) \diff \boldsymbol{x}_{W_{v} \setminus I} .
\end{equation*}
Note that in case $I_{-v} = \lbrace w_{1} \rbrace \cup \pa(v; w_{1})$, no integration is required since then $W_{v} \setminus I = \emptyset$.

\bigskip

\begin{figure}[htb]
	\centering
		\begin{subfigure}[htb]{.205\textwidth}
			\captionsetup{font=scriptsize}
			\centering \includegraphics[scale=.75]{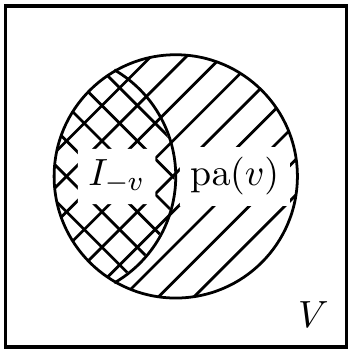}
			\caption{$W_{v} \! = \! \lbrace w_{1} \rbrace \cup \pa(v; w_{1})$}
			\label{fig:venn1}
		\end{subfigure}
		\begin{subfigure}[htb]{.32\textwidth}
			\captionsetup{font=scriptsize}
			\centering \includegraphics[scale=.75]{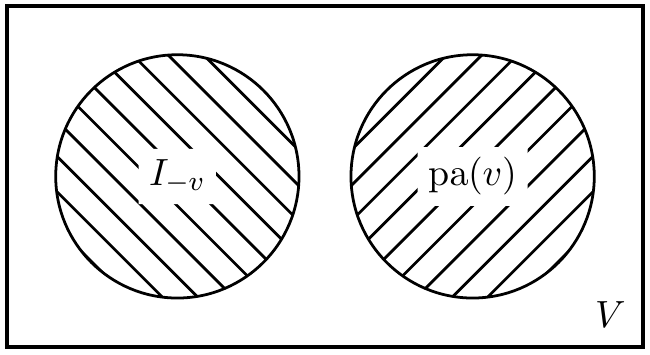}
			\caption{$W_{v} = \emptyset$ or $W_{v} = \lbrace w_{2} \rbrace \cup \pa(v; w_{2})$}
			\label{fig:venn2}
		\end{subfigure}%
		\begin{subfigure}[htb]{.25\textwidth}
			\captionsetup{font=scriptsize}
			\centering \includegraphics[scale=.75]{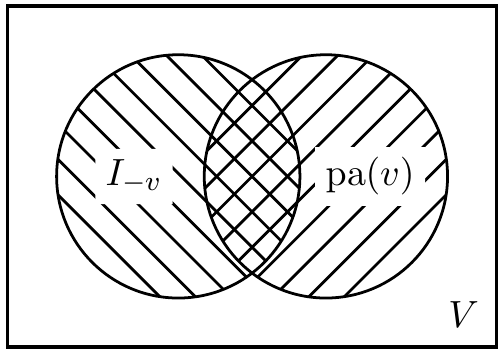}
			\caption{$W_{v} = \lbrace w_{2} \rbrace \cup \pa(v; w_{2})$}
			\label{fig:venn3}
		\end{subfigure}
		\begin{subfigure}[htb]{.205\textwidth}
			\captionsetup{font=scriptsize}
			\centering \includegraphics[scale=.75]{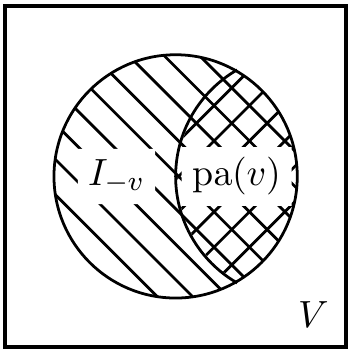}
			\caption{$W_{v} = \pa(v)$}
			\label{fig:venn4}
		\end{subfigure}
	\caption{Venn diagrams of the sets $I_{-v} \neq \emptyset$ and $\pa(v) \neq \emptyset$, and corresponding definitions of $W_{v}$ (see lower captions).}
	\label{fig:venn}
\end{figure}

Next, let $I_{-v} \cap \pa(v) = \emptyset$ (Figure \ref{fig:venn2}). If $S_{v} = \pa(v)$, then $\sep{\lbrace v \rbrace}{I_{-v}}{(\mathcal{D}_{\An(I)})^{m}}$ since $v$ has no descendants in $I_{-v}$. Hence, we again have $f_{I}(\boldsymbol{x}_{I}) = f_{v}(x_{v}) \, f_{I_{-v}}(\boldsymbol{x}_{I_{-v}})$, that is Equation \eqref{eq:dagpdflem} for $W_{v} \coloneqq \emptyset$. If, however, $S_{v} \neq \pa(v)$, then $\csep{\lbrace v \rbrace}{I_{-v}}{\pa(v; w_{2})}{(\mathcal{D}_{\An(\pa(v; w_{2}) \cup I)})^{m}}$, and with $W_{v} \coloneqq \lbrace w_{2} \rbrace \cup \pa(v; w_{2})$ the global $\mathcal{D}$-Markov property yields
\begin{equation}\label{eq:lemglobalmarkov}
	f_{I \cup W_{v} }(\boldsymbol{x}_{I \cup W_{v}}) = f_{v \vert W_{v}}(x_{v} \, \vert \, \boldsymbol{x}_{W_{v}}) \, f_{I_{-v} \vert W_{v}}(\boldsymbol{x}_{I_{-v}} \, \vert \, \boldsymbol{x}_{W_{v}}) \, f_{W_{v}}(\boldsymbol{x}_{W_{v}}) .
\end{equation}
Since $I_{-v} \cap W_{v} = \emptyset$, we thus get
\begin{equation*}
	f_{I}(\boldsymbol{x}_{I}) = \int_{\Real^{\abs{W_{v}}}} f_{v \vert W_{v}}(x_{v} \, \vert \, \boldsymbol{x}_{W_{v}}) \, f_{W_{v} \cup I_{-v}}(\boldsymbol{x}_{W_{v} \cup I_{-v}}) \diff \boldsymbol{x}_{W_{v}} .
\end{equation*}
Note that in case $S_{v} = \emptyset$, we have $W_{v} = \pa(v)$.

\bigskip

Finally, assume $I_{-v} \cap \pa(v) \neq \emptyset$ such that $I_{-v} \nsubseteq \pa(v)$ (Figure \ref{fig:venn3}). Similarly to Equation \eqref{eq:lemglobalmarkov}, we obtain with $W_{v} \coloneqq \lbrace w_{2} \rbrace \cup \pa(v; w_{2})$
\begin{equation*}
	f_{I \cup W_{v} }(\boldsymbol{x}_{I \cup W_{v}}) = f_{v \vert W_{v}}(x_{v} \, \vert \, \boldsymbol{x}_{W_{v}}) \, f_{(I_{-v} \setminus W_{v}) \vert W_{v}}(\boldsymbol{x}_{I_{-v} \setminus W_{v}} \, \vert \, \boldsymbol{x}_{W_{v}}) \, f_{W_{v}}(\boldsymbol{x}_{W_{v}})
\end{equation*}
by the global $\mathcal{D}$-Markov property, and hence
\begin{equation*}
	f_{I}(\boldsymbol{x}_{I}) = \int_{\Real^{\abs{W_{v} \setminus I}}} f_{v \vert W_{v}}(x_{v} \, \vert \, \boldsymbol{x}_{W_{v}}) \, f_{W_{v} \cup I_{-v}}(\boldsymbol{x}_{W_{v} \cup I_{-v}}) \diff \boldsymbol{x}_{W_{v} \setminus I} .
\end{equation*}
Note again that in case $S_{v} = \emptyset$, we have $W_{v} = \pa(v)$. Also, note that in case $\pa(v) \subseteq I_{-v}$ (Figure \ref{fig:venn4}), no integration is required. This establishes the claim.
\end{proof}

The set $W_{v}$ in Lemma \ref{lem:dagpdf} is either empty or of the form $\lbrace w \rbrace \cup \pa(v; w)$ for some $w \in \pa(v)$. In the latter case, we can express the conditional pdf $f_{v \vert \lbrace w \rbrace \cup \pa(v; w)}(\, \cdot \, \vert \, \boldsymbol{x}_{\lbrace w \rbrace \cup \pa(v; w)})$ on the right hand side of Equation \eqref{eq:dagpdflem} in terms of the univariate marginals $F_{u}$, $u \in V$, and the (conditional) pair copulas $C_{v, u \vert \pa(v; u)}$, $u \in \pa(v)$, as follows.

\begin{lem}\label{lem:dagcondpdf}
Let the notation be as in Lemma \ref{lem:dagpdf} and let $P$ have strictly increasing univariate marginal cdfs. Let $I_{-v} \neq \emptyset$ and let $S_{v} \neq \pa(v)$. Then
\begin{equation*}
	f_{v \vert W_{v}}(x_{v} \hspace{.2ex} \vert \hspace{.2ex} \boldsymbol{x}_{W_{v}}) = f_{v}(x_{v}) \!\! \prod_{w \in W_{v}} \!\! c_{v, w \vert \! \pa(v; w)} \bigl( F_{v \vert\! \pa(v; w)}(x_{v} \hspace{.2ex} \vert \hspace{.2ex} \boldsymbol{x}_{\pa(v; w)}), F_{w \vert\! \pa(v; w)}(x_{w} \hspace{.2ex} \vert \hspace{.2ex} \boldsymbol{x}_{\pa(v; w)}) \hspace{.2ex} \big\vert \hspace{.2ex} \boldsymbol{x}_{\pa(v; w)} \bigr)
\end{equation*}
for all $x_{v} \in \Real$ and $\boldsymbol{x}_{W_{v}} = (x_{w})_{w \in W_{v}} \in \Real^{\abs{W_{v}}}$. 
\end{lem}

\begin{proof}
Since $I_{-v} \neq \emptyset$ and $S_{v} \neq \pa(v)$, $W_{v}$ is non-empty and thus $W_{v} = \lbrace u \rbrace \cup \pa(v; u)$ for some $u \in \pa(v)$. By Equation \eqref{eq:chainrule}, we can hence write
\begin{equation*}
	f_{v \vert W_{v}}(x_{v} \hspace{.15ex} \vert \hspace{.15ex} \boldsymbol{x}_{W_{v}}) = f_{v}(x_{v}) \!\! \prod_{w \in W_{v}} \!\! c_{v, w \vert \! \pa(v; w)} \bigl( F_{v \vert\! \pa(v; w)}(x_{v} \hspace{.15ex} \vert \hspace{.15ex} \boldsymbol{x}_{\pa(v; w)}), F_{w \vert\! \pa(v; w)}(x_{w} \hspace{.15ex} \vert \hspace{.15ex} \boldsymbol{x}_{\pa(v; w)}) \hspace{.15ex} \big\vert \hspace{.15ex} \boldsymbol{x}_{\pa(v; w)} \bigr),
\end{equation*}
and the claim is proven.
\end{proof}

\begin{algorithm}[!htb]
	\begin{algorithmic}[1]
		\Require Well-ordered DAG $\mathcal{D}$; set of parent orderings $\mathcal{O}$; non-empty vertex set $I \subseteq V$.
		\Ensure Factorisation $f$. \Comment{(marginal) pdf $f_{I}(\boldsymbol{x}_{I})$}
		\State $f \leftarrow 1$;
		\State $J \leftarrow \emptyset$; \Comment{indices of integration variables}
		\While{$\abs{I} \ge 1$}
			\State \Comment{Select maximal vertex:}
			\State $v \leftarrow$ maximal vertex in $I$ by the well-ordering of $\mathcal{D}$;
			\State $f \leftarrow f \cdot f_{v}(x_{v})$;
			\State $I \leftarrow I_{-v}$;
			\State \Comment{Determine the set $W_{v}$:}
			\State $W \leftarrow \emptyset$;
			\State $S \leftarrow \bigl\lbrace w \in \pa(v) \, \big\vert \, \sep{\lbrace w \rbrace}{I}{(\mathcal{D}_{\An(\lbrace w \rbrace \cup I)})^{m}} \bigr\rbrace$;
			\If{$I \neq \emptyset$ \textbf{and} $S \neq pa(v)$}
				\If{$I \subseteq \pa(v)$}
					\State $w \leftarrow$ maximal vertex in $I$ by the parent ordering $<_{v}$;
					\State $W \leftarrow \lbrace w \rbrace \cup \pa(v; w)$;
				\Else
					\State $w \leftarrow$ maximal vertex in $\pa(v) \setminus S$ by the parent ordering $<_{v}$;
					\State $W \leftarrow \lbrace w \rbrace \cup \pa(v; w)$;
				\EndIf
			\EndIf
			\State \Comment{Introduce corresponding pair copulas and integration variables:}
			\For{$w \in W$}
				\State $f \leftarrow f \cdot c_{v, w \vert \! \pa(v; w)} \bigl( F_{v \vert\! \pa(v; w)}(x_{v} \hspace{.1ex} \vert \hspace{.1ex} \boldsymbol{x}_{\pa(v; w)}), F_{w \vert\! \pa(v; w)}(x_{w} \hspace{.1ex} \vert \hspace{.1ex} \boldsymbol{x}_{\pa(v; w)}) \hspace{.1ex} \big\vert \hspace{.1ex} \boldsymbol{x}_{\pa(v; w)} \bigr)$;
				\If{$w \notin I$}
					\State $I \leftarrow I \cup \lbrace w \rbrace$;
					\State $J \leftarrow J \cup \lbrace w \rbrace$;
				\EndIf
			\EndFor
		\EndWhile
		\State $f \leftarrow \int_{\Real^{\abs{J}}} f \diff \boldsymbol{x}_{J}$;
	\end{algorithmic}
	\caption{Pair-copula decomposition of a (marginal) pdf.}
	\label{alg:dagpdf}
\end{algorithm}

Since all vertices in $W_{v} \cup I_{-v}$ are smaller than $v$ by the well-ordering of $\mathcal{D}$ and since $V$ is finite, we can inductively apply Lemmas \ref{lem:dagpdf} and \ref{lem:dagcondpdf} to the pdf $f_{W_{v} \cup I_{-v}}$ in Equation \eqref{eq:dagpdflem} until no unconditional pdfs of dimension higher than one remain. Let $J$ denote the set of vertices corresponding to the integration variables added during this iterative procedure (and including $W_{v} \setminus I$). Given a set $\mathcal{O}$ of parent orderings for $\mathcal{D}$, Lemma \ref{lem:dagpdf} yields a set $W_{u}$ for every $u \in I \cup J$. We have hence established the following theorem.

\begin{thm}\label{thm:dagpdf}
Let the notation be as in Lemma \ref{lem:dagcondpdf}. Then $f_{I}(\boldsymbol{x}_{I})$ takes the form
\begin{equation*}
	\int_{\Real^{\abs{J}}} \prod_{v \in (I \cup J)} \hspace{-1.5ex} f_{v}(x_{v}) \hspace{-.75ex} \prod_{w \in W_{v}} \hspace{-.75ex} c_{v, w \vert \! \pa(v; w)} \bigl( F_{v \vert\! \pa(v; w)}(x_{v} \hspace{.2ex} \vert \hspace{.2ex} \boldsymbol{x}_{\pa(v; w)}), F_{w \vert\! \pa(v; w)}(x_{w} \hspace{.2ex} \vert \hspace{.2ex} \boldsymbol{x}_{\pa(v; w)}) \hspace{.2ex} \big\vert \hspace{.2ex} \boldsymbol{x}_{\pa(v; w)} \bigr) \diff \boldsymbol{x}_{J}
\end{equation*}
for all $\boldsymbol{x}_{I} = (x_{v})_{v \in I} \in R^{\abs{I}}$. \hfill\qed
\end{thm}

Note that in the special case $I = V$, Theorem \ref{thm:dagpdf} yields Equation \eqref{eq:dagpdf}. Above procedure for deriving a pair-copula decomposition of $f_{I}$ as given in Theorem \ref{thm:dagpdf} is summarised in Algorithm~\ref{alg:dagpdf}.

\smallskip

% Place Algorithm \ref{alg:dagpdf} here!

\begin{exmp}
We consider the well-ordered DAG $\mathcal{D}$ in Figure \ref{fig:dagcdf}. The edges and parent orderings of $\mathcal{D}$ can be summarised in a matrix $A_{\mathcal{D}} = (a_{i j})_{1 \le i, j \le 7}$ whose elements satisfy $a_{i j} = k$, $k  \le \abs{\pa(j)}$, if $\mathcal{D}$ contains the edge $i \rightarrow j$ and if $i$ is the $k$-th smallest parent of $j$ by $<_{j}$, and $a_{i j} = 0$ otherwise, see Figure \ref{fig:dagcdf}. For the reader's convenience we will omit function arguments.
Equation \eqref{eq:dagpdf} yields
\begin{align*}
	f & = f_{1} \cdots f_{7} \, c_{21}(F_{2}, F_{1}) \, c_{31}(F_{3}, F_{1}) \, c_{42}(F_{4}, F_{2}) \, c_{41 \vert 2}(F_{4 \vert 2}, F_{1 \vert 2}) \, c_{54}(F_{5}, F_{4}) \, c_{53 \vert 4}(F_{5 \vert 4}, F_{3 \vert 4})\\
	& \quad \cdot c_{65}(F_{6}, F_{5}) \, c_{64 \vert 5}(F_{6 \vert 5}, F_{4 \vert 5}) \, c_{63 \vert 54}(F_{6 \vert 54}, F_{3 \vert 54}) \, c_{62 \vert 543}(F_{6 \vert 543}, F_{2 \vert 543}) \, c_{75}(F_{7}, F_{5})\\
	& \quad \cdot c_{76 \vert 5}(F_{7 \vert 5}, F_{6 \vert 5}) \, c_{73 \vert 56}(F_{7 \vert 56}, F_{3 \vert 56}) .
\end{align*}
We will later derive a pair-copula decomposition for $F_{3 \vert 56}$. In preparation, we now use Algorithm~\ref{alg:dagpdf} to derive pair-copula decompositions for $f_{356}$ and $f_{56}$.

\bigskip

\begin{figure}[htb]
	\centering
		\begin{subfigure}[c]{.45\textwidth}
			\centering \includegraphics[width=.75\textwidth]{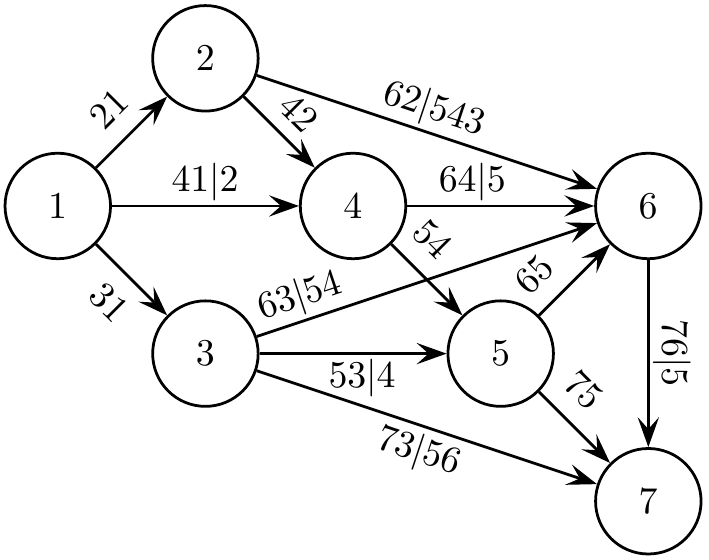}
		\end{subfigure}
		\begin{subfigure}[c]{.45\textwidth}
			\centering \includegraphics[width=.7\textwidth]{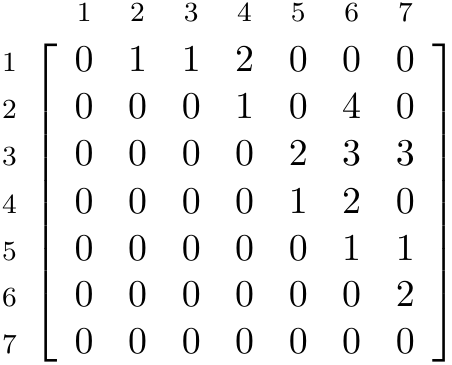}
		\end{subfigure}
	\caption{A well-ordered (vertex labels) DAG $\mathcal{D}$ (left) with parent orderings $2 <_{4} 1$, $4 <_{5} 3$, $5 <_{6} 4 <_{6} 3 <_{6} 2$, $5 <_{7} 6 <_{7} 3$ specifying the pair copulas $C_{21}$, $C_{31}$, $C_{42}$, $C_{41 \vert 2}$, $C_{54}$, $C_{53 \vert 4}$, $C_{65}$, $C_{64 \vert 5}$, $C_{63 \vert 54}$, $C_{62 \vert 542}$, $C_{75}$, $C_{76 \vert 5}$, $C_{73 \vert 56}$ (edge labels), and corresponding representation matrix $A_{\mathcal{D}}$ (right).}
	\label{fig:dagcdf}
\end{figure}

As a result of applying Algorithm \ref{alg:dagpdf} to $f_{356}$ and $f_{56}$, we obtain
\begin{align}\label{eq:dagexample1}
	f_{356} & = \int_{\Real^{3}}  f_{6 \vert 543} \, f_{5 \vert 43} \, f_{4 \vert 21} \, f_{3 \vert 1} \, f_{2 \vert 1} \, f_{1} \diff \boldsymbol{x}_{124}\notag\\
	& = \int_{\Real^{3}} f_{6} \, c_{63 \vert 54}(F_{6 \vert 54}, F_{3 \vert 54}) \, c_{64 \vert 5}(F_{6 \vert 5}, F_{4 \vert 5}) \, c_{65}(F_{6}, F_{5}) \, f_{5} \, c_{53 \vert 4}(F_{5 \vert 4}, F_{3 \vert 4}) \, c_{54}(F_{5}, F_{4})\\
	& \quad \cdot f_{4} \, c_{41 \vert 2}(F_{4 \vert 2}, F_{1 \vert 2}) \, c_{42}(F_{4}, F_{2}) \, f_{3} \, c_{31}(F_{3}, F_{1}) \, f_{2} \, c_{21}(F_{2}, F_{1}) \, f_{1} \diff \boldsymbol{x}_{124}\notag
\end{align}
and $f_{56} = f_{6 \vert 5} \, f_{5} = f_{6} \, c_{65}(F_{6}, F_{5}) \, f_{5}$, respectively, see Table \ref{tab:dagexample}.

\smallskip

\begin{table}[htb]
	\centering
	\begin{tabular*}{\textwidth}{@{\extracolsep{\fill}}lcccccccc}
		\toprule
		$f_{356}$ & $I$ & $v$ & $S$ & $I_{-v} \subseteq \pa(v)$? & $w$ & $W$ & $J$\\
		\midrule
		& $\lbrace 3, 5, 6 \rbrace$ & $6$ & $\emptyset$ & $\checkmark$ & $3$ & $\lbrace 3, 4, 5 \rbrace$ & $\lbrace 4 \rbrace$\\
		& $\lbrace 3, 4, 5 \rbrace$ & $5$ & $\emptyset$ & $\checkmark$ & $3$ & $\lbrace 3, 4 \rbrace$ & $\lbrace 4 \rbrace$\\
		& $\lbrace 3, 4 \rbrace$ & $4$ & $\emptyset$ & \textsf{X} & $1$ & $\lbrace 1, 2 \rbrace$ & $\lbrace 1, 2, 4 \rbrace$\\
		& $\lbrace 1, 2, 3 \rbrace$ & $3$ & $\emptyset$ & \textsf{X} & $1$ & $\lbrace 1 \rbrace$ & $\lbrace 1, 2, 4 \rbrace$\\
		& $\lbrace 1, 2 \rbrace$ & $2$ & $\emptyset$ & $\checkmark$ & $1$ & $\lbrace 1 \rbrace$ & $\lbrace 1, 2, 4 \rbrace$\\
		& $\lbrace 1 \rbrace$ & $1$ & $\emptyset$ & $-$ & $-$ & $\emptyset$ & $\lbrace 1, 2, 4 \rbrace$\\
		\midrule
		$f_{56}$ & $I$ & $v$ & $S$ & $I_{-v} \subseteq \pa(v)$? & $w$ & $W$ & $J$\\
		\midrule
		& $\lbrace 5, 6 \rbrace$ & $6$ & $\emptyset$ & $\checkmark$ & $5$ & $\lbrace 5 \rbrace$ & $\emptyset$\\
		& $\lbrace 5 \rbrace$ & $5$ & $\lbrace 3, 4 \rbrace$ & $-$ & $-$ & $\emptyset$ & $\emptyset$\\
		\bottomrule
	\end{tabular*}
	\caption{Vertices and vertex sets obtained during the application of Algorithm \ref{alg:dagpdf} to the pdfs $f_{356}$ and $f_{56}$ corresponding to the DAG $\mathcal{D}$ in Figure \ref{fig:dagcdf}.} 
	\label{tab:dagexample}
\end{table}
\end{exmp}

\subsubsection*{When can $\int_{-\infty}^{y} f_{\lbrace v \rbrace \cup K}(\boldsymbol{x}_{\lbrace v \rbrace \cup K}) \diff x_{v}$ in Equation \eqref{eq:dagcondcdf} be further simplified?}

Let us now return to the conditional cdf in Equation \eqref{eq:dagcondcdf}. Setting $I \coloneqq \lbrace v \rbrace \cup K$, the numerator on the right hand side of Equation \eqref{eq:dagcondcdf} takes the form $\int_{-\infty}^{y} f_{I}(\boldsymbol{x}_{I}) \diff x_{v}$. Decompose $f_{I}(\boldsymbol{x}_{I})$ according to Theorem \ref{thm:dagpdf}, and let $J$ denote the set of vertices corresponding to the newly added integration variables. Clearly, $J \subseteq \An(I) \setminus I$. If the old integration variable $x_{v}$ does not appear as a conditioning variable in one of the pair copulas $C_{v, w \vert \! \pa(v; w)}$, $v \in I \cup J$, $w \in W_{v}$, in the decomposition of $f_{I}$, it may be possible to solve the integral with respect to $x_{v}$ analytically. More precisely, let $J^{\prime} \subseteq J$ and let $W^{\prime} \subseteq W \coloneqq I_{-v} \cup J^{\prime}$ be non-empty. Let $k \coloneqq \abs{W^{\prime}}$ and write $W^{\prime} = \lbrace w_{1}, \ldots, w_{k} \rbrace$. Moreover, set $W_{-i}^{\prime} \coloneqq \lbrace w_{1}, \ldots, w_{i - 1} \rbrace$ for all $i \in \lbrace 1, \ldots k \rbrace$. Assume that the pair copula $C_{v, w_{k} \vert W_{-k}^{\prime}}$ is available in the pair-copula decomposition of $f$, that is $W_{-k}^{\prime} = \pa(v; w_{k})$ or $W_{-k}^{\prime} = \pa(w_{k}; v)$, and that (after possible algebraic manipulation) $f_{I}$ takes the form
\begin{equation}\label{eq:dagchainrulec}
	\int_{\Real^{\abs{J^{\prime}}}} f_{v}(x_{v}) \prod_{i = 1}^{k} c_{v, w_{i} \vert W_{-i}^{\prime}} \bigl( F_{v \vert W_{-i}^{\prime}}(x_{v} \, \vert \, \boldsymbol{x}_{W_{-i}^{\prime}}), F_{w_{i} \vert W_{-i}^{\prime}}(x_{w_{i}} \, \vert \, \boldsymbol{x}_{W_{-i}^{\prime}}) \, \big\vert \, \boldsymbol{x}_{W_{-i}^{\prime}} \bigr) \, f_{W}(\boldsymbol{x}_{W}) \diff \boldsymbol{x}_{J^{\prime}} .
\end{equation}
Then Fubini's theorem and Equation \eqref{eq:chainrule} yield that $\int_{-\infty}^{y} f_{I}(\boldsymbol{x}_{I}) \diff x_{v}$ takes the form
\begin{equation}\label{eq:dagchainruleh}
	\int_{\Real^{\abs{J^{\prime}}}} h_{\underline{v}, w_{k} \vert W_{-k}^{\prime}} \bigl( F_{v \vert W_{-k}^{\prime}}(y \, \vert \, \boldsymbol{x}_{W_{-k}^{\prime}}), F_{w_{k} \vert W_{-k}^{\prime}}(x_{w_{k}} \, \vert \, \boldsymbol{x}_{W_{-k}^{\prime}}) \, \big\vert \, \boldsymbol{x}_{W_{-k}^{\prime}} \bigr) \, f_{W}(\boldsymbol{x}_{W}) \diff \boldsymbol{x}_{J^{\prime}} ,
\end{equation}
where the integral with respect to $x_{v}$ was replaced by an h-function which, by assumption, is available in the pair-copula decomposition of $f$. Note that some of the copula pdfs $c_{v, w_{i} \vert W_{-i}^{\prime}}$, $i \in \lbrace 1, \ldots, k - 1 \rbrace$, in Equation \eqref{eq:dagchainrulec} may not correspond to an edge in $\mathcal{D}$, but may instead be given implicitly by an integral over further variables, or may be equal to $1$ due to a related Markov property of $P$, see also the example below. We need to take these special cases into account when checking the applicability of the inverse chain rule algorithmically.

\bigskip

It may sometimes also be useful to substitute $u_{w} \coloneqq F_{w}(x_{w})$, that is $\diff u_{w} = f_{w}(x_{w}) \diff x_{w}$, for all $w \in J$ in the pair-copula decomposition of $f_{I}$, and thus to write
\begin{equation*}
	f_{I}(\boldsymbol{x}_{I}) = \int_{[0, 1]^{\abs{J}}} c_{I \cup J} \Bigl( \bigl( F_{w}(x_{w}) \bigr)_{w \in I}, \boldsymbol{u}_{J} \Bigr) \! \prod_{w \in I} \! f_{w}(x_{w}) \diff \boldsymbol{u}_{J} .
\end{equation*}
A similar transformation can be applied to the denominator in Equation \eqref{eq:dagcondcdf} if integration variables are present.

\begin{exmp}[continued]
Consider the integral $\int_{-\infty}^{\cdot} f_{356} \diff x_{3}$ associated to the DAG $\mathcal{D}$ in Figure~\ref{fig:dagcdf}, which will later appear when deriving a pair-copula decomposition for $F_{3 \vert 56}$. Observing that
\begin{align*}
	\int_{\Real^{2}} f_{4} \, c_{41 \vert 2}(F_{4 \vert 2}, F_{1 \vert 2}) \, c_{42}(F_{4}, F_{2}) \, f_{3} \, c_{31}(F_{3}, F_{1}) \, f_{2} \, c_{21}(F_{2}, F_{1}) \, f_{1} \diff \boldsymbol{x}_{12} & = \int_{\Real^{2}} f_{1234} \diff \boldsymbol{x}_{12} = f_{34}\\
	& = f_{4} \, f_{3} \, c_{43}(F_{4}, F_{3}) ,
\end{align*}
Equation \eqref{eq:dagexample1} yields
\begin{align*}
	\int_{-\infty}^{\cdot} f_{356} \diff x_{3} & = \int_{-\infty}^{\cdot} \int_{\Real} f_{3} \, c_{63 \vert 54}(F_{6 \vert 54}, F_{3 \vert 54}) \, c_{53 \vert 4}(F_{5 \vert 4}, F_{3 \vert 4}) \, c_{43}(F_{4}, F_{3})\\
	& \quad \cdot f_{6} \, c_{64 \vert 5}(F_{6 \vert 5}, F_{4 \vert 5}) \, c_{65}(F_{6}, F_{5}) \, f_{5} \, c_{54}(F_{5}, F_{4}) \, f_{4} \diff x_{4} \diff x_{3} .
\end{align*}
Note that $c_{43}$ is not available in the pair-copula decomposition of $f$. Since by Equation~\eqref{eq:chainrule}
\begin{equation*}
	\int_{-\infty}^{\cdot} f_{3} \, c_{63 \vert 54}(F_{6 \vert 54}, F_{3 \vert 54}) \, c_{53 \vert 4}(F_{5 \vert 4}, F_{3 \vert 4}) \, c_{43}(F_{4}, F_{3}) \diff x_{3} = h_{6\underline{3} \vert 54}(F_{6 \vert 54}, F_{3 \vert 54}) ,
\end{equation*}
we can, however, simplify the integral with respect to $x_{3}$, and $c_{43}$ vanishes. We obtain
\begin{equation}\label{eq:dagexample2}
	\int_{-\infty}^{\cdot} f_{356} \diff x_{3} = \int_{\Real} f_{6} \, h_{6\underline{3} \vert 54}(F_{6 \vert 54}, F_{3 \vert 54}) \, c_{64 \vert 5}(F_{6 \vert 5}, F_{4 \vert 5}) \, c_{65}(F_{6}, F_{5}) \, f_{5} \, c_{54}(F_{5}, F_{4}) \, f_{4} \diff x_{4}.
\end{equation}
\end{exmp}

\subsubsection*{Pair-copula decompositions for conditional cdfs}

Summing up, a pair-copula decomposition for the conditional cdf $F_{v \vert K}(\, \cdot \, \vert \, \boldsymbol{x}_{K})$ in Equation \eqref{eq:dagcondcdf} is obtained in three steps. First, we apply Theorem \ref{thm:dagpdf} to $f_{\lbrace v \rbrace \cup K}$ and $f_{K}$. Second, we possibly apply the inverse chain rule to the integral with respect to $x_{v}$ in the numerator. Last, we cancel common factors like $\prod_{w \in K} f_{w}(x_{w})$ in the numerator and the denominator. The procedure is summarised in Algorithm \ref{alg:dagcondcdf}.

\smallskip

\begin{algorithm}[htb]
	\begin{algorithmic}[1]
		\Require Well-ordered DAG $\mathcal{D}$; set of parent orderings $\mathcal{O}$; vertex $v \in V$ (conditioned variable), vertex set $K \subseteq V_{-v}$ (conditioning variables).
		\Ensure Factorisation $F$. \Comment{conditional cdf $F_{v \vert K}(y \, \vert \, \boldsymbol{x}_{K})$}
		\State \Comment{Exploit global $\mathcal{D}$-Markov property:}
		\While{$\exists \, w \in K \colon \csep{\lbrace v \rbrace}{\lbrace w \rbrace}{K_{-w}}{(\mathcal{D}_{\An(\lbrace v \rbrace \cup K)})^{m}}$}
			\State $K \leftarrow K_{-w}$;
		\EndWhile
		\State \Comment{Numerator:}
		\State $n \leftarrow \text{Algorithm\_\ref{alg:dagpdf}} \bigl( \mathcal{D}, \mathcal{O}, \lbrace v \rbrace \cup K \bigr)$;
		\State $n \leftarrow \int_{-\infty}^{y} n \diff x_{v}$; 
		\State simplify $n$ with inverse chain rule for variable $x_{v}$ if possible;
		\State \Comment{Denominator:}
		\State $d \leftarrow \text{Algorithm\_\ref{alg:dagpdf}} \bigl( \mathcal{D}, \mathcal{O}, K \bigr)$;
		\State \Comment{Conditional cdf $F_{v \vert K}(y \, \vert \, \boldsymbol{x}_{K})$:}
		\State cancel common factors in $n$ and $d$;
		\State $F \leftarrow \frac{n}{d}$;
	\end{algorithmic}
	\caption{Pair-copula decomposition of a conditional cdf.}
	\label{alg:dagcondcdf}
\end{algorithm}

As can be seen from Theorem \ref{thm:dagpdf} and Equation \eqref{eq:dagchainruleh}, the factorisation for $F_{v \vert K}(\, \cdot \, \vert \, \boldsymbol{x}_{K})$ obtained from Algorithm \ref{alg:dagcondcdf} may contain some new conditional cdfs. This problem can, however, be solved inductively. Let $w$ denote the maximal vertex in $\lbrace v \rbrace \cup K$ by the well-ordering of $\mathcal{D}$. Since Algorithm \ref{alg:dagcondcdf} only adds ancestors of $\lbrace v \rbrace \cup K$ as integration variables, all vertices involved in the new conditional cdfs are smaller than or equal to $w$ by the well-ordering of $\mathcal{D}$. In particular, those conditional cdfs involving $w$ are of the special form $F_{w \vert \! \pa(w; u)}(\, \cdot \, \vert \, \boldsymbol{x}_{\pa(w; u)})$ for some $u \in \pa(w)$, and can by Equation \eqref{eq:condcdf} iteratively be expressed as
\begin{equation*}
	F_{v \vert \! \pa(w; u)}(x_{v} \hspace{.2ex} \vert \hspace{.2ex} \boldsymbol{x}_{\pa(w; u)}) = h_{\underline{v}, u \vert \! \pa(v; u)} \bigl( F_{v \vert\! \pa(v; u)}(x_{v} \hspace{.2ex} \vert \hspace{.2ex} \boldsymbol{x}_{\pa(v; u)}), F_{u \vert\! \pa(v; u)}(x_{u} \hspace{.2ex} \vert \hspace{.2ex} \boldsymbol{x}_{\pa(v; u)}) \hspace{.2ex} \big\vert \hspace{.2ex} \boldsymbol{x}_{\pa(v; u)} \bigr).
\end{equation*}
Hence, all vertices involved in the algorithmically more demanding new conditional cdfs are strictly smaller than $w$ by the well-ordering of $\mathcal{D}$. Corresponding pair-copula decompositions for the new conditional cdfs can thus be computed inductively by again applying Algorithm \ref{alg:dagcondcdf}. Since $V$ is finite, the whole procedure terminates after finitely many steps, and the desired decomposition in terms of only univariate marginals and (conditional) pair copulas is obtained.

\bigskip

Overall, we observed that the problems of deriving pair-copula decompositions for a conditional cdf and a marginal pdf are deeply intertwined and can be solved by alternating iteration. Note that it is sufficient for our purposes to exploit only those conditional independence properties of $P$ which follow directly from graph separation in $\mathcal{D}$ via the global $\mathcal{D}$-Markov property. Once a complete decomposition for $f$ is obtained, the evaluation at $\boldsymbol{x} \in \Real^{d}$ can be performed vertex-by-vertex and parent-by-parent along the well-ordering of $\mathcal{D}$. That is, given $v^{\ast} \in V$ and $w^{\ast} \in \pa(v^{\ast})$, we first evaluate all terms corresponding to the marginals $F_{v}$ and the pair copulas $C_{v, w \vert \! \pa(v; w)}$ for $v$ smaller than $v^{\ast}$ by the well-ordering of $\mathcal{D}$ and $w <_{v} w^{\ast}$ if $w^{\ast} \in \pa(v)$, before evaluating the terms corresponding to $F_{v^{\ast}}$ and $C_{v^{\ast}, w^{\ast} \vert \! \pa(v^{\ast}; w^{\ast})}$.

\begin{exmp}[continued]
For the DAG $\mathcal{D}$ in Figure \ref{fig:dagcdf}, we sketch how to apply Algorithm~\ref{alg:dagcondcdf} to obtain a pair-copula decomposition for $F_{3 \vert 56}$. Note that $\mathcal{D}$ contains the edges $3 \rightarrow 5$ and $3 \rightarrow 6$, which is why neither $5$ nor $6$ can be removed from the conditioning set. We get $F_{3 \vert 56} = \int_{-\infty}^{\cdot} \frac{f_{356}}{f_{56}} \diff x_{3}$. Applying our previous results for $f_{356}$ and $f_{56}$, respectively, we further have
\begin{equation*}
	F_{3 \vert 56} = \frac{\int_{\Real} f_{6} \, h_{6\underline{3} \vert 54}(F_{6 \vert 54}, F_{3 \vert 54}) \, c_{64 \vert 5}(F_{6 \vert 5}, F_{4 \vert 5}) \, c_{65}(F_{6}, F_{5}) \, f_{5} \, c_{54}(F_{5}, F_{4}) \, f_{4} \diff x_{4}}{f_{6} \, c_{65}(F_{6}, F_{5}) \, f_{5}},
\end{equation*}
see Equation \eqref{eq:dagexample2}. Thus, by cancelling common factors, we finally obtain
\begin{align*}
	F_{3 \vert 56} & = \int_{\Real} h_{6\underline{3} \vert 54}(F_{6 \vert 54}, F_{3 \vert 54}) \, c_{64 \vert 5}(F_{6 \vert 5}, F_{4 \vert 5}) \, c_{54}(F_{5}, F_{4}) \, f_{4} \diff x_{4} .
\end{align*}
\end{exmp}

\subsection{Simulation, ML estimation, and model selection in PCBNs}\label{sec:pcbnstat}

Given (conditional) pair copulas $C_{v, w \vert \! \pa(v; w)}(\, \cdot \, , \cdot \, ; \boldsymbol{\theta}_{v, w \vert \! \pa(v; w)})$, $v \in V$, $w \in \pa(v)$, with joint parameter vector $\boldsymbol{\theta} \coloneqq (\boldsymbol{\theta}_{v, w \vert \! \pa(v; w)})_{v \in V, w \in \pa(v)} \in \boldsymbol{\Theta}$, above construction yields a $d$-variate copula model, which we will denote by $\lbrace C_{\mathcal{D}, \mathcal{O}, \boldsymbol{\theta}} \, \vert \, \boldsymbol{\theta} \in \boldsymbol{\Theta} \rbrace$. Note that for computational convenience, we again make the \emph{simplifying assumption} of constant conditional copulas described in Section \ref{sec:vinestat}. Together with families of univariate marginals, $\lbrace C_{\mathcal{D}, \mathcal{O}, \boldsymbol{\theta}} \, \vert \, \boldsymbol{\theta} \in \boldsymbol{\Theta} \rbrace$ constitutes a statistical model which merges the advantages of graphical Markov modelling with the distributional flexibility of the pair-copula approach. We will refer to such a model as a \emph{pair-copula Bayesian network} (PCBN). We want to mention that PCBNs were first introduced in \citet{Kurowicka.Cooke:2005}. The analyses therein were, however, restricted to pair-copula families with the property that zero rank correlation implies independence. 

\subsubsection*{Simulation}

Write $V = \lbrace v_{1}, \ldots, v_{d} \rbrace$ according to the well-ordering of $\mathcal{D}$ and set $V_{-i} \coloneqq \lbrace v_{1}, \ldots, v_{i - 1} \rbrace$ for all $i \in \lbrace 1, \ldots, d \rbrace$. A sample $\boldsymbol{u} = (u_{v_{1}}, \ldots, u_{v_{d}}) \in [0, 1]^{d}$ from a fully specified PCBN with uniform $[0, 1]$ univariate margins is obtained by simulating $d$ independent uniform $[0, 1]$ variables $x_{1}, \ldots, x_{d}$ and applying the quantile transformations
\begin{align*}
	u_{v_{1}} & \coloneqq x_{1},\\
	u_{v_{2}} & \coloneqq F_{v_{2} \vert v_{1}}^{-1}(x_{2} \, \vert \, u_{v_{1}}; \boldsymbol{\theta}),\\
	u_{v_{3}} & \coloneqq F_{v_{3} \vert V_{-3}}^{-1}(x_{3} \, \vert \, \boldsymbol{u}_{V_{-3}}; \boldsymbol{\theta}),\\
	& \vdots\\
	u_{v_{d}} & \coloneqq F_{v_{d} \vert V_{-d}}^{-1}(x_{d} \, \vert \, \boldsymbol{u}_{V_{-d}}; \boldsymbol{\theta}) .
\end{align*}
The order in which the components of $\boldsymbol{u}$ are generated is given by the well-ordering of $\mathcal{D}$. Solving transformation equation $i$ for $x_{i}$, we have by the local $\mathcal{D}$-Markov property in Equation \eqref{eq:localmarkov}
\begin{equation}\label{eq:dagsiminv}
	x_{i} = F_{v_{i} \vert V_{-i}}(u_{v_{i}} \, \vert \, \boldsymbol{u}_{V_{-i}}; \boldsymbol{\theta}) = F_{v_{i} \vert \! \pa(v_{i})}(u_{v_{i}} \, \vert \, \boldsymbol{u}_{\pa(v_{i})}; \boldsymbol{\theta}), \quad i \in \lbrace 1, \ldots, d \rbrace.
\end{equation}
Now assume that $\pa(v_{i}) \neq \emptyset$, and let $w$ denote the largest vertex in $\pa(v_{i})$ by the parent ordering $<_{v_{i}}$. Then Equations \eqref{eq:dagsiminv} and \eqref{eq:condcdf} yield
\begin{equation}\label{eq:dagsim}
	x_{i} = h_{\underline{v_{i}}, w \vert \! \pa(v_{i}; w)} \bigl( F_{v_{i} \vert \! \pa(v_{i}; w)}(u_{v_{i}} \, \vert \, \boldsymbol{u}_{\pa(v_{i}; w)}; \boldsymbol{\theta}), F_{w \vert \! \pa(v_{i}; w)}(u_{w} \, \vert \, \boldsymbol{u}_{\pa(v_{i}; w)}; \boldsymbol{\theta}); \boldsymbol{\theta} \bigr).
\end{equation}
Since $u_{v_{i}}$ is only contained in the first h-function argument $F_{v_{i} \vert \! \pa(v_{i}; w)}(u_{v_{i}} \, \vert \, \boldsymbol{u}_{\pa(v_{i}; w)}; \boldsymbol{\theta})$ on the right hand side of Equation \eqref{eq:dagsim}, we obtain by induction that the only inverse functions needed in the computation of $u_{v_{i}}$ are the inverse h-functions $h_{\underline{v_{i}}, w^{\ast} \vert \! \pa(v_{i}, w^{\ast})}^{-1}$, $w^{\ast} \in \pa(v_{i})$.

\subsubsection*{ML estimation}

ML estimation for PCBNs was first considered in \citet{Bauer.Czado.Klein:2012}. Let $\boldsymbol{u} = \left( \boldsymbol{u}^{1}, \ldots, \boldsymbol{u}^{n} \right)$, $n \in \Nat$, be a realisation of a sample of i.i.d.\ observations $\boldsymbol{U}^{1}, \ldots, \boldsymbol{U}^{n}$ from a random variable $\boldsymbol{U}$ on $[0, 1]^{d}$ with copula family $\lbrace C_{\mathcal{D}, \mathcal{O}, \boldsymbol{\theta}} \, \vert \, \boldsymbol{\theta} \in \boldsymbol{\Theta} \rbrace$ and uniform univariate margins. The restriction to uniform univariate margins is made along the same lines as in Section \ref{sec:vinestat} for vine copula models. Equation \eqref{eq:dagpdf} yields the log-likelihood function
\begin{equation}\label{eq:dagll}
	l(\boldsymbol{\theta}; \boldsymbol{u}) = \sum_{k = 1}^{n} \sum_{v \in V} \sum_{w \in \pa(v)} \!\! \log c_{v, w \vert \! \pa(v; w)} \Bigr( F_{v \vert \! \pa(v; w)} \bigl( u_{v}^{k} \, \big\vert \, \boldsymbol{u}_{\pa(v; w)}^{k}; \boldsymbol{\theta} \bigr), F_{w \vert \! \pa(v; w)} \bigl( u_{w}^{k} \, \big\vert \, \boldsymbol{u}_{\pa(v; w)}^{k}; \boldsymbol{\theta} \bigr); \boldsymbol{\theta} \Bigr) .
\end{equation}
ML estimation of the parameters in Equation \eqref{eq:dagll} can be performed using a stepwise approach similar to the one discussed in Section \ref{sec:vinestat} for vine copula models. The only difference to vine copula models is that we iterate over the vertices of $\mathcal{D}$ and their respective parents instead of over the trees of an R-vine. Hence again, in a first step, \emph{sequential ML estimates} are computed and in a second step, using the sequential ML estimates as starting values, \emph{joint ML estimates} $\boldsymbol{\widehat{\theta}}_{v, w \vert \! \pa(v; w)}$, $v \in V$, $w \in \pa(v)$, are inferred.

\subsubsection*{Model selection}

Model selection for PCBNs involves estimation of the DAG $\mathcal{D}$, selection of the set $\mathcal{O}$ of parent orderings, and selection of the pair-copula families for $C_{v, w \vert \! \pa(v; w)}$, $v \in V$, $w \in \pa(v)$. Estimation of $\mathcal{D}$ will be the subject of Section \ref{sec:pc}. Given $\mathcal{D}$ and $\mathcal{O}$, the selection of pair-copula families can be performed in a similar way as in Section \ref{sec:vinestat} for vine copula models, again with the difference that the iteration is vertex-by-vertex and parent-by-parent instead of tree-by-tree.

\bigskip

For the selection of $\mathcal{O}$ we propose a greedy-type procedure inspired by the structure selection algorithm for vine copula models outlined in Section \ref{sec:vinestat}. Clearly, an ordering of the parents of a vertex $v \in V$ is only required if $\pa(v) \neq \emptyset$. We assume $\mathcal{D}$ is well-ordered. Let $v \in V$ and assume that $k \coloneqq \abs{\pa(v)} \ge 1$. Moreover, let $i \in \lbrace 1, \ldots, k \rbrace$ and assume that we have already selected the $i - 1$ smallest parents of $v$, denoted by $w_{1} <_{v} \cdots <_{v} w_{i - 1}$. This implies that we have already selected pair-copula families for $C_{v^{\ast}, w \vert \! \pa(v^{\ast}; w)}$, $v^{\ast}$ smaller than $v$ by the well-ordering of $\mathcal{D}$, $w \in \pa(v^{\ast})$, and $C_{v, w_{j} \vert \! \pa(v; w_{j})}$, $j < i$. Also, this implies that we have inferred corresponding ML parameter estimates, which we summarise in the vector $\boldsymbol{\widehat{\theta}}$. Let $W_{-i} \coloneqq \lbrace w_{1}, \ldots, w_{i - 1} \rbrace$. The selection of $w_{i}$ is performed in three steps. First, we compute the pseudo-observations $F_{v \vert W_{-i}} \bigl( u_{v}^{k} \, \big\vert \, \boldsymbol{u}_{W_{-i}}^{k}; \boldsymbol{\widehat{\theta}} \bigr)$ and $F_{w \vert W_{-i}} \bigl( u_{w}^{k} \, \big\vert \, \boldsymbol{u}_{W_{-i}}^{k}; \boldsymbol{\widehat{\theta}} \bigr)$, $k \in \lbrace 1, \ldots, n \rbrace$, for all $w \in \pa(v) \setminus W_{-i}$. Note that for $i = 1$, nothing needs to be done since all univariate marginals are uniform on $[0, 1]$. Second, we assign a weight $\omega_{v, w}$ to every edge $w \rightarrow v$, $w \in \pa(v) \setminus W_{-i}$, based on the previously calculated pseudo-observations, and choose $w_{i}$ such that $w_{i} \rightarrow v$ has optimal edge weight. Suitable weights are, for instance, the absolute values of estimates of Kendall's $\tau$, or AIC or BIC values of selected pair-copula families with estimated parameters. Last, we select a pair-copula family for $C_{v, w_{i} \vert \! \pa(v; w_{i})}$ and compute an ML estimate of the corresponding parameter(s). Again, this last step may have already been performed when computing the edge weights $\omega_{v, w}$.

\section{Structure estimation in Bayesian networks using the PC algorithm}\label{sec:pc}

The first task of modelling the joint distribution of a given set of variables with a Bayesian network is to identify the DAG $\mathcal{D} = (V, E)$ specifying the Markov structure of the variables. A convenient approach to defining $\mathcal{D}$ is the use of \emph{expert knowledge}. However, the scope of this approach is rather limited since expert knowledge is often incomplete or unavailable. Data-driven \emph{structure estimation algorithms} provide a computer-based alternative to elicited expert knowledge. \citet{Robinson:1973} has shown that the number $n_{d}$ of DAGs on $d \coloneqq \abs{V}$ labelled vertices is given by the recurrence equation
\begin{equation*}
	n_{0} = 1, \qquad n_{d} = \sum_{k = 1}^{d} (-1)^{k - 1} \dbinom{d}{k} \, 2^{k (d - k)} \, n_{d - k}.
\end{equation*}
Since $n_{d}$ grows super-exponentially in $d$, a systematic trial of all possible DAGs on $V$ is infeasible, and thus efficient searching algorithms are required. A considerable number of structure estimation algorithms has been proposed over the last two decades, see \citet[Chapters $8$ -- $11$]{Neapolitan:2003} and \citet[Chapter $18$]{Koller.Friedman:2009} for an overview. The majority of these algorithms follow one of the two estimation approaches predominant in the literature: the \emph{constraint-based} and the \emph{score-and-search-based} approach. In the constraint-based approach, $\mathcal{D}$ is inferred from a series of conditional independence tests. In the score-and-search-based approach, $\mathcal{D}$ is found by optimising a given scoring function---like AIC or BIC---over a suitable search space, for instance the space of all DAGs or the space of all Markov-equivalence classes. Besides, there exist hybrid algorithms which combine both approaches. Unfortunately, available implementations of aforementioned algorithms are mainly confined to discrete or Gaussian models and are hence not suited for our non-Gaussian continuous Bayesian networks.

\subsection{The PC algorithm}

We will provide a structure estimation algorithm that is particularly suited to finding the DAG $\mathcal{D} = (V, E)$ underlying a non-Gaussian continuous Bayesian network. Our algorithm is a version of one of the most popular constraint-based estimation algorithms, the \emph{PC algorithm} (named after its inventors \emph{P}eter Spirtes and \emph{C}lark Glymour), see \citet{Spirtes.Glymour:1991} and \citet[Section $5.4.2$]{Spirtes.Glymour.Scheines:2000}. To fix notation and for the reader's convenience, we will now recall the PC algorithm. Let $P$ be an absolutely continuous $\mathcal{D}$-Markovian probability measure on $[0, 1]^{d}$ with uniform univariate margins. The restriction to uniform univariate margins is made along the same lines as in Section \ref{sec:pcbnstat}. Moreover, let $\boldsymbol{u} = \left( \boldsymbol{u}^{1}, \ldots, \boldsymbol{u}^{n} \right)$, $n \in \Nat$, be a realisation of a sample of i.i.d.\ observations $\boldsymbol{U}^{1}, \ldots, \boldsymbol{U}^{n}$ from a random variable $\boldsymbol{U}$ distributed as $P$. The PC algorithm for estimating $\mathcal{D}$ from $\boldsymbol{u}$ involves three major steps in which the complete UG $\mathcal{G}$ on $V$ is gradually transformed into a CG $\mathcal{G}^{\ast}$ on $V$, which is supposed to be the essential graph $\mathcal{D}^{e}$ corresponding to the Markov-equivalence class $[\mathcal{D}]$ of $\mathcal{D}$. The resulting CG $\mathcal{G}^{\ast}$ can then be extended to a DAG as outlined in Section \ref{sec:bn}.

\bigskip

In the first step of the PC algorithm, a series of tests for conditional independence is performed on $\boldsymbol{u}$. More precisely, for all distinct vertices $i, j \in V$ and chosen vertex sets $K \subseteq V \setminus \lbrace i, j \rbrace$, the null hypothesis $\text{H}_{0} \colon \cind{U_{i}}{U_{j}}{\boldsymbol{U}_{K}}$ is tested against the general alternative $\text{H}_{1} \colon \ncind{U_{i}}{U_{j}}{\boldsymbol{U}_{K}}$ of conditional dependence. Given a suitable independence test of choice, we denote the test decision at significance level $\alpha \in (0,1)$ by $T_{\alpha}(\boldsymbol{u}_{i}, \boldsymbol{u}_{j}; \boldsymbol{u}_{K}) \in \lbrace \text{H}_{0}, \text{H}_{1} \rbrace$. We will later introduce a novel class of conditional independence tests that is particularly tailored to the algorithm and applicable to non-Gaussian continuous data. If $T_{\alpha}(\boldsymbol{u}_{i}, \boldsymbol{u}_{j}; \boldsymbol{u}_{K}) = \text{H}_{0}$, the edge $i \edge j$ is removed from $\mathcal{G}$ and the conditioning set $K$ is stored in two variables $S_{i j}$ and $S_{j i}$ for later use. As a result of the first step, $\mathcal{G}$ is turned into the skeleton of $\mathcal{G}^{\ast}$. Step one is given in Algorithm \ref{alg:pcskeleton}.

\bigskip

\begin{algorithm}[htb]
	\begin{algorithmic}[1]
		\Require Data set $\boldsymbol{u}$; significance level $\alpha \in (0,1)$; conditional independence test with test decision $T_{\alpha}(\boldsymbol{u}_{i}, \boldsymbol{u}_{j}; \boldsymbol{u}_{K})$ for the null hypothesis $\text{H}_{0} \colon \cind{U_{i}}{U_{j}}{\boldsymbol{U}_{K}}$, $i \neq j \in V$, $K \subseteq V \setminus \lbrace i, j \rbrace$.
		\Ensure Skeleton $\mathcal{G} = (V, E_{\mathcal{G}})$; separation sets $S_{i j}$, $i \neq j \in V$, $(i, j) \notin E_{\mathcal{G}}, (j, i) \notin E_{\mathcal{G}}$.
		\State $\mathcal{G} \leftarrow$ complete UG on $V$;
		\State $k \leftarrow 0$;
		\Repeat
			\For{$i \in V$ \textbf{and} $j \in \ad(i)$} \Comment{$i$ and $j$ are adjacent in $\mathcal{G}$}
				\If{$T_{\alpha}(\boldsymbol{u}_{i}, \boldsymbol{u}_{j}; \boldsymbol{u}_{K}) = \text{H}_{0}$ for any $K \subseteq \ad(i) \setminus \lbrace j \rbrace$ with $\abs{K} = k$}
					\State delete $i \edge j$ from $\mathcal{G}$;
					\State $S_{i j} \leftarrow K$;
					\State $S_{j i} \leftarrow K$;
				\EndIf
			\EndFor
			\State $k \leftarrow k + 1$.
		\Until{$\abs{\ad(i)} \le k$ for all $i \in V$.}
	\end{algorithmic}
	\caption{PC algorithm: finding the skeleton.}
	\label{alg:pcskeleton}
\end{algorithm}

In the second step, $\mathcal{G}$ is transformed into a CG by introducing a v-structure $i \rightarrow k \leftarrow j$ whenever $i$ and $j$ are non-adjacent, $k \in \ad(i) \cap \ad(j)$, and $k \notin S_{i j}$. In the last step, $\mathcal{G}$ is transformed into $\mathcal{G}^{\ast}$ by directing further edges of $\mathcal{G}$ to prevent new v-structures and directed cycles, until no more edges need direction. Steps two and three are given in Algorithm \ref{alg:pcedgedir}, where the third step was taken from \citet[Section $2.5$]{Pearl:2009}. If $P$ is faithful to $\mathcal{D}$ and if all statistical test decisions made in Algorithm \ref{alg:pcskeleton} are correct, then Algorithm \ref{alg:pcedgedir} will return the correct graph $\mathcal{D}^{e}$, see \citet{Meek:1995}. Due to the finite sample size or the existence of hidden variables, the application of Algorithm \ref{alg:pcskeleton} to empirical data may sometimes, however, lead to conflicting information about edge directions. That is, it may be possible in a given situation that Algorithm \ref{alg:pcedgedir}, while introducing v-structures, first orients an undirected edge $i \edge j$ into $i \rightarrow j$, and later tries to introduce $i \leftarrow j$. In such a situation, we keep $i \rightarrow j$ and skip the new v-structure including $i \leftarrow j$. We can test whether the resulting CG can still be extended to a DAG without introducing new v-structures or directed cycles using the algorithm by \citet{Dor.Tarsi:1992}. The PC algorithm can also be adapted to incorporate existing expert knowledge, see \citet{Meek:1995} and \citet{Moole.Valtorta:2004}. We will henceforth assume that $P$ is faithful to $\mathcal{D}$ and that there are no hidden variables.

\begin{algorithm}[htb]
	\begin{algorithmic}[1]
		\Require Skeleton $\mathcal{G} = (V, E_{\mathcal{G}})$; separation sets $S_{i j}$, $i \neq j \in V$, $(i, j) \notin E_{\mathcal{G}}, (j, i) \notin E_{\mathcal{G}}$.
		\Ensure Chain graph $\mathcal{G}$.
		\State\Comment{Introduce v-structures:}
		\For{$i \in V$ \textbf{and} $j \notin \ad(i)$ \textbf{and} $k \in \ad(i) \cap \ad(j)$}
			\If{$k \notin S_{i j}$}
				\State replace $i \edge k \edge j$ by $i \rightarrow k \leftarrow j$ in $\mathcal{G}$;
			\EndIf
		\EndFor 
		\State\Comment{Orient as many undirected edges as possible by repeated application of the following rules:}
		\Repeat
			\State \textbf{R1} orient $j \edge k$ into $j \rightarrow k$ whenever $\mathcal{G}$ contains $i \rightarrow j$ and $k \notin \ad(i)$;
			\State \textbf{R2} orient $i \edge j$ into $i \rightarrow j$ whenever $\mathcal{G}$ contains $i \rightarrow k \rightarrow j$;
			\State \textbf{R3} orient $i \edge j$ into $i \rightarrow j$ whenever $\mathcal{G}$ contains $i \edge k \rightarrow j$ and $i \edge l \rightarrow j$, and $l \notin \ad(k)$;
		\Until{no more edges can be directed;}
	\end{algorithmic}
	\caption{PC algorithm: introducing edge directions}
	\label{alg:pcedgedir}
\end{algorithm}

\subsubsection*{Testing conditional independence using partial correlations}

The centrepiece of the PC algorithm---as of any constraint-based estimation algorithm---is the test for conditional independence. In a Gaussian framework, the test of choice is usually a test for zero partial correlation $\rho_{i j \cdot K}$, see, for instance, \citet[Section $4.3$]{Anderson:2003}. The null hypothesis then translates into $\text{H}_{0} \colon \rho_{i j \cdot K}(X_{i}, X_{j}; X_{K}) = 0$, where $X_{k} \coloneqq \Phi^{-1}(U_{k})$ for all $k \in V$, and $\Phi$ denotes the univariate standard normal cdf. Here, the quantile function $\Phi^{-1}$ is applied to $U$ in order to transform the uniform univariate copula margins to standard normal margins. The conditional independence test is based on the asymptotic normality
\begin{equation*}
	\sqrt{d - \abs{K} - 3} \ \widehat{z}_{n} \xrightarrow[n \to \infty]{\mathcal{L}} \Norm(0, 1), \qquad \widehat{z}_{n} \coloneqq \frac{1}{2} \log{\left( \frac{1 + \widehat{\rho}_{i j \cdot K}(\boldsymbol{X}_{i}^{n}, \boldsymbol{X}_{j}^{n}; \boldsymbol{X}_{K}^{n})}{1 - \widehat{\rho}_{i j \cdot K}(\boldsymbol{X}_{i}^{n}, \boldsymbol{X}_{j}^{n}; \boldsymbol{X}_{K}^{n})} \right)},
\end{equation*}
of the Fisher's $z$-transformed partial-correlation estimator $\widehat{\rho}_{i j \cdot K}$ under $\text{H}_{0}$, see again \citet[Section $4.3$]{Anderson:2003}. Here, $\xrightarrow[]{\mathcal{L}}$ denotes convergence in distribution, $\Norm(0, 1)$ is the univariate standard normal distribution, and $\boldsymbol{X}_{k}^{n} \coloneqq \left( \Phi^{-1}(U_{k}^{1}), \ldots, \Phi^{-1}(U_{k}^{n}) \right)$ for all $k \in V$. \citet{Kalisch.Buehlmann:2007} have proven uniform convergence of the PC algorithm under joint normality and a mild sparsity assumption for the underlying DAG, cf.\ also \citet{Harris.Drton:2012}. An implementation of the PC algorithm with above partial correlation test is available in the R package \texttt{pcalg} \citep{Kalisch.Maechler.Colombo:2012}. The \texttt{pcalg} package also provides an interface for self-implemented conditional independence tests.

\subsection{Testing conditional independence using vine copulas and the Rosenblatt transform}

Above test for zero partial correlation was derived under the assumption of joint normality. We will now introduce a copula-based alternative test for conditional independence that is also applicable to non-Gaussian continuous data. Assume $K \neq \emptyset$. Otherwise, the problem reduces to testing ordinary (unconditional) stochastic independence. Let $F_{i, j \vert K}(\, \cdot \, , \cdot \, \vert \, \boldsymbol{v}_{K})$ denote the conditional cdf of $U_{i}$ and $U_{j}$ given $\boldsymbol{U}_{K} = \boldsymbol{v}_{K}$, and let $C_{i, j \vert K}(\, \cdot \, , \cdot \, \vert \, \boldsymbol{v}_{K})$ be the corresponding conditional copula. Moreover, let $C_{\dperp}$ denote the independence copula on $[0, 1]^{2}$. The conditional independence $\cind{U_{i}}{U_{j}}{\boldsymbol{U}_{K}}$ holds if and only if
\begin{equation*}
	F_{i, j \vert K}(v_{i}, v_{j} \, \vert \, \boldsymbol{v}_{K}) = C_{i, j \vert K} \bigl( F_{i \vert K}(v_{i} \, \vert \, \boldsymbol{v}_{K}), F_{j \vert K}(v_{j} \, \vert \, \boldsymbol{v}_{K}) \, \big\vert \, \boldsymbol{v}_{K} \bigr) = F_{i \vert K}(v_{i} \, \vert \, \boldsymbol{v}_{K}) \, F_{j \vert K}(v_{j} \, \vert \, \boldsymbol{v}_{K})
\end{equation*}
for all $v_{i}, v_{j} \in [0, 1]$ and $P_{K}$-almost all $\boldsymbol{v}_{K} \in [0, 1]^{\abs{K}}$, where $\boldsymbol{U}_{K} \sim P_{K}$. Hence, the null hypothesis of the conditional independence test can be stated as $H_{0} \colon C_{i, j \vert K}(\, \cdot \, , \cdot \, \vert \, \boldsymbol{v}_{K}) = C_{\dperp}(\, \cdot \, , \cdot \,)$ for $P_{K}$-almost all $\boldsymbol{v}_{K} \in [0, 1]^{\abs{K}}$. Using the simplifying assumption that $C_{i, j \vert K}(\, \cdot \, , \cdot \, \vert \, \boldsymbol{v}_{K})$ depends on $\boldsymbol{v}_{K}$ only through $F_{i \vert K}(\, \cdot \, , \cdot \, \vert \, \boldsymbol{v}_{K})$ and $F_{j \vert K}(\, \cdot \, , \cdot \, \vert \, \boldsymbol{v}_{K})$ discussed in Section \ref{sec:vinestat}, we drop $\boldsymbol{v}_{K}$ from $C_{i, j \vert K}(\, \cdot \, , \cdot \, \vert \, \boldsymbol{v}_{K})$ and approximate $H_{0}$ by the more accessible null hypothesis $H_{0}^{\ast} \colon C_{i, j \vert K}(\, \cdot \, , \cdot \,) = C_{\dperp}(\, \cdot \, , \cdot \,)$. The new null hypothesis $H_{0}^{\ast}$ can be tested using any test for ordinary (unconditional) stochastic independence of two continuous random variables applied to the transformed observations $W_{i \vert K}^{1}, \ldots, W_{i \vert K}^{n}$ and $W_{j \vert K}^{1}, \ldots, W_{j \vert K}^{n}$, where
\begin{equation}\label{eq:vartrafo}
	W_{i \vert K}^{k} \coloneqq F_{i \vert K} \bigl( U_{i}^{k} \, \big\vert \, \boldsymbol{U}_{K}^{k} \bigr) \quad \text{and} \quad W_{j \vert K}^{k} \coloneqq F_{j \vert K} \bigl( U_{j}^{k} \, \big\vert \, \boldsymbol{U}_{K}^{k} \bigr)
\end{equation}
for all $k \in \lbrace 1, \ldots, n \rbrace$. \citet{Song:2009} called Equation \eqref{eq:vartrafo} the \emph{Rosenblatt transform} after \citet{Rosenblatt:1952}, while \citet{Bergsma:2011} called it the \emph{partial copula transform}. Given a realisation $\boldsymbol{u}$ of $\left( \boldsymbol{U}^{1}, \ldots, \boldsymbol{U}^{n} \right)$, the difficulty of this approach lies in the computation of the transformed realisations $\boldsymbol{w}_{i \vert K}$ and $\boldsymbol{w}_{j \vert K}$, where $w_{i \vert K}^{k} \coloneqq F_{i \vert K} \bigl( u_{i}^{k} \, \big\vert \, \boldsymbol{u}_{K}^{k} \bigr)$ and $w_{j \vert K}^{k} \coloneqq F_{j \vert K} \bigl( u_{j}^{k} \, \big\vert \, \boldsymbol{u}_{K}^{k} \bigr)$ for all $k \in \lbrace 1, \ldots, n \rbrace$. Note that the conditional cdfs $F_{i \vert K}(\, \cdot \, \vert \, \boldsymbol{v}_{K})$ and $F_{j \vert K}(\, \cdot \, \vert \, \boldsymbol{v}_{K})$ are typically unknown and need to be estimated in the course of the testing procedure. \citet{Bergsma:2011} suggested the use of non-parametric kernel estimators for this task. By contrast, we propose a parametric estimation method that is based on vine copula models.

\subsubsection*{Estimating conditional cdfs using vine copula models}

Taking another look at vine copula models as described in Section \ref{sec:vinestat}, we observe that transformed realisations like $\boldsymbol{w}_{i \vert K}$ and $\boldsymbol{w}_{j \vert K}$ naturally emerge in the log-likelihood function. In fact, given any distinct $i, j \in V$ and $K \subseteq V \setminus \lbrace i, j \rbrace$, it is always possible to construct a regular vine $\mathcal{V} = (T_{1}, \ldots, T_{p})$, $p \coloneqq 1 + \abs{K}$, in which tree $T_{1}$ has vertex set $V_{1} = \lbrace i \rbrace \cup \lbrace j \rbrace \cup K$ and tree $T_{p}$ is of the form $i, l \vert K_{-l} \stackrel{i, j \vert K}{\dash} j, m \vert K_{-m}$ for some $l, m \in K$. The corresponding log-likelihood function $l \bigl(\boldsymbol{\theta}; \boldsymbol{u}_{\lbrace i \rbrace \cup \lbrace j \rbrace \cup K} \bigr)$, $\boldsymbol{\theta} \in \Theta$, contains the pair-copula pdf $c_{i, j \vert K}$ with arguments $F_{i \vert K} \bigl( u_{i}^{k} \, \big\vert \, \boldsymbol{u}_{K}^{k}; \boldsymbol{\theta} \bigr)$ and $F_{j \vert K} \bigl( u_{j}^{k} \, \big\vert \, \boldsymbol{u}_{K}^{k}; \boldsymbol{\theta} \bigr)$ for all $k \in \lbrace 1, \ldots, n \rbrace$. Thus, by computing an ML estimate $\widehat{\boldsymbol{\theta}}$ of $\boldsymbol{\theta}$ and subsequently evaluating $l$ at $\widehat{\boldsymbol{\theta}}$, we obtain estimates $\widehat{w}_{i \vert K}^{k} \coloneqq F_{i \vert K} \bigl( u_{i}^{k} \, \big\vert \, \boldsymbol{u}_{K}^{k}; \widehat{\boldsymbol{\theta}} \bigr)$ and $\widehat{w}_{j \vert K}^{k} \coloneqq F_{j \vert K} \bigl( u_{j}^{k} \, \big\vert \, \boldsymbol{u}_{K}^{k}; \widehat{\boldsymbol{\theta}} \bigr)$ of $w_{i \vert K}^{k}$ and $w_{j \vert K}^{k}$, respectively, as a welcome side effect.

\bigskip

We call a vertex $v$ in Tree $T_{q}$, $q \in \lbrace 1, \ldots, p - 1 \rbrace$, an \emph{inner vertex} if $\abs{\ad(v)} \ge 2$. In order to construct such a vine $\mathcal{V}$, we have to follow one simple rule:

\medskip

\begin{tabular*}{\textwidth}{@{\extracolsep{\fill}}llp{0.90\textwidth}}
	& \textbf{R} & Neither $i$ nor $j$ may be part of an inner vertex in the trees $T_{1}, \ldots, T_{p - 1}$ of $\mathcal{V}$.
\end{tabular*}

\medskip

Following \textbf{R}, it is even possible to restrict the class of R-vines to C- or D-vines. The only inner vertices of a C-vine are the root vertices of the trees $T_{1}, \ldots, T_{p - 1}$. Thus, in a C-vine obeying \textbf{R}, $i$ and $j$ do not appear in the root vertices of the respective trees. Similarly, in a D-vine obeying \textbf{R}, $i$ and $j$ only appear in the boundary vertices of trees $T_{1}, \ldots, T_{p - 1}$. Figures \ref{fig:rvine} and \ref{fig:cdvine} give an example of a C-, a D-, and an R-vine, respectively, having the same edge label in tree $T_{p}$.

\bigskip

\begin{figure}[htb]
	\centering
		\begin{subfigure}[c]{.55\textwidth}
			\centering \includegraphics[width=\textwidth]{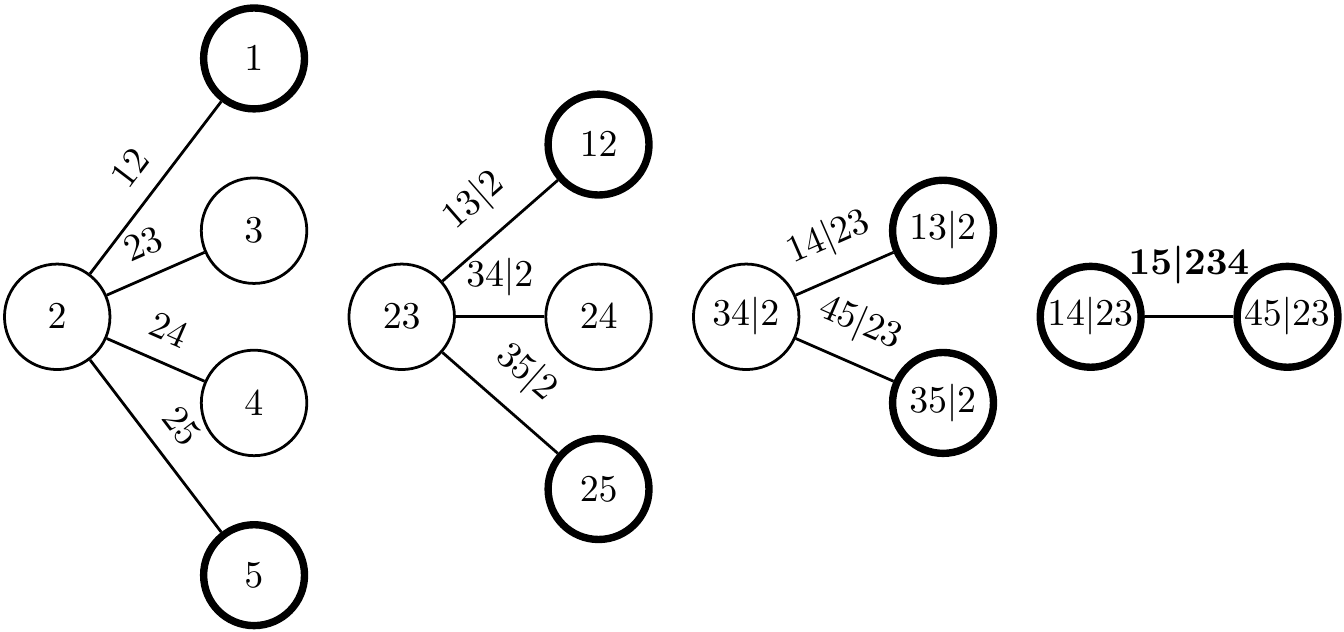}
		\end{subfigure}\qquad
		\begin{subfigure}[c]{.35\textwidth}
			\centering \includegraphics[width=\textwidth]{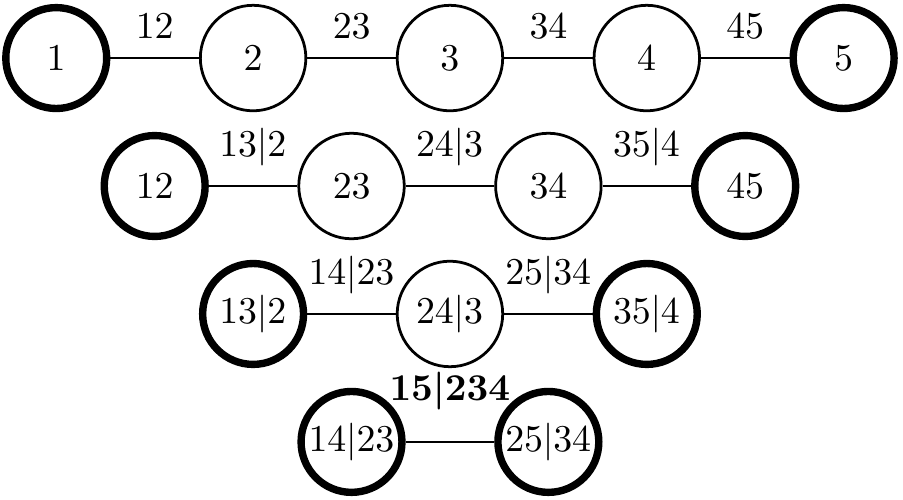}
		\end{subfigure}
	\caption{A C- (left) and a D-vine (right) on five vertices having the same edge label $15 \vert 234$ in tree $T_{4}$. A corresponding R-vine is given in Figure \ref{fig:rvine}. The three vines were constructed according to rule \textbf{R} with $i = 1$ and $j = 5$. Boundaries of nodes including either $1$ or $5$ appear in bold.}
	\label{fig:cdvine}
\end{figure}

The tree structure of $\mathcal{V}$ can be estimated from $\boldsymbol{u}_{\lbrace i \rbrace \cup \lbrace j \rbrace \cup K}$ by adapting the greedy search strategies described in Section \ref{sec:vinestat} to the new constraint \textbf{R}. An optimal C-vine obeying \textbf{R} is found by restricting the sets of possible root vertices for trees $T_{1}, \ldots, T_{p - 1}$ to vertices containing neither $i$ nor $j$, respectively. In order to find an optimal D-vine obeying \textbf{R}, the unconstrained TSP usually solved has to be replaced by a constrained TSP with fixed source vertex $i$ and destination vertex $j$. Finally, an optimal R-vine obeying \textbf{R} is found by first estimating a smaller R-vine $\mathcal{V}_{K}$ with first tree vertices $K$. Having found $\mathcal{V}_{K}$, vertex $i$ is then connected to a vertex $l \in K$ in tree $T_{1}$ such that the new edge $i \edge l$ has optimal edge weight amongst all possible edges $i \edge m$ for $m \in K$. The same is done for vertex $j$. Note that this way, $j$ cannot be connected to $i$. The newly formed structure is then sequentially transformed into $\mathcal{V}$ by analogously extending the remaining trees $T_{2}, \ldots, T_{p}$, such that the proximity condition and \textbf{R} are always satisfied and the corresponding edge weights are optimised. Copula selection and ML estimation in the resulting vine copula model is then performed as usual, see Section \ref{sec:vinestat}.

\subsubsection*{Vine-copula-based conditional independence tests}

Summing up, we test the conditional independence $\cind{U_{i}}{U_{j}}{\boldsymbol{U}_{K}}$ in three steps. In the first step, we construct a vine $\mathcal{V}$ on the vertices $\lbrace i \rbrace \cup \lbrace j \rbrace \cup K$ by applying a modified version of one of the structure estimation algorithms described in Section \ref{sec:vinestat} to $\boldsymbol{u}_{\lbrace i \rbrace \cup \lbrace j \rbrace \cup K}$. In the second step, we select corresponding pair-copula families, perform ML estimation in the resulting model, and evaluate the log-likelihood function $l$ at the estimated parameter vector $\widehat{\boldsymbol{\theta}}$ to obtain transformed realisations $\widehat{\boldsymbol{w}}_{i \vert K} \coloneqq \bigl( \widehat{w}_{i \vert K}^{k} \bigr)_{1 \le k \le n}$ and $\widehat{\boldsymbol{w}}_{j \vert K} \coloneqq \bigl( \widehat{w}_{j \vert K}^{k} \bigr)_{1 \le k \le n}$, respectively. In the last step, we apply a test for ordinary stochastic independence of two continuous random variables to $\widehat{\boldsymbol{w}}_{i \vert K}$ and $\widehat{\boldsymbol{w}}_{j \vert K}$. Note that in the first iteration step of Algorithm \ref{alg:pcskeleton}, only unconditional independences, that is $K = \emptyset$, are tested, and thus the independence test of choice is directly applied to $\boldsymbol{u}$.

\bigskip

We will examine the performance of our novel testing procedure in a simulation study in Section \ref{sec:simstudy} using three different tests for ordinary stochastic independence. Recycling notation, consider the null hypothesis $H_{0} \colon U_{i} \ind U_{j}$ vs.\ $H_{1} \colon U_{i} \nind U_{j}$. The first test used is a test for zero Kendall's $\tau$ with null hypothesis $H_{0}^{\ast} \colon \tau(U_{i}, U_{j}) = 0$ vs.\ $H_{1}^{\ast} \colon \tau(U_{i}, U_{j}) \neq 0$. Under $H_{0}$, the Kendall's $\tau$ estimator $\widehat{\tau}_{n}$ exhibits the asymptotic normality
\begin{equation*}
	\sqrt{\frac{9 n (n - 1)}{2 (2 n + 5)}} \ \widehat{\tau}_{n} \! \left( \boldsymbol{U}_{i}, \boldsymbol{U}_{j} \right) \xrightarrow[n \to \infty]{\mathcal{L}} \Norm(0, 1),
\end{equation*}
where $\boldsymbol{U}_{i} \coloneqq \bigl( U_{i}^{1}, \ldots, U_{i}^{n} \bigr)$ and $\boldsymbol{U}_{j} \coloneqq \bigl( U_{j}^{1}, \ldots, U_{j}^{n} \bigr)$, see \citet[Section $8.1$]{Hollander.Myles:1999}. In general, $\tau(U_{i}, U_{j}) = 0$ does not imply $U_{i} \ind U_{j}$. However, for many popular copula families like the Clayton, the Gaussian, and the Gumbel copula families, $H_{0}$ and $H_{0}^{\ast}$ are equivalent. The family of Student's t copulas serves as a counterexample. We then consider $H_{0}^{\ast}$ an approximation for $H_{0}$. The other two independence tests used in Section \ref{sec:simstudy} are of Cramér-von Mises type. More precisely, independence test number two is the test for zero Hoeffding's $D$ proposed by \citet{Hoeffding:1948}. P-values of the sample test statistic $\widehat{D}_{n}$ are computed using the asymptotically equivalent sample test statistic $\widehat{B}_{n}$ by \citet{Blum.Kiefer.Rosenblatt:1961}, see also \citet[Section $8.6$]{Hollander.Myles:1999}. Independence test number three is the test by \citet{Genest.Remillard:2004} based on the empirical copula process.

\section{Simulation study}\label{sec:simstudy}

We conducted an extensive simulation study to examine the small sample performance of the PC algorithm in finding the true Markov structure underlying a PCBN. To this end, we drew samples from various PCBNs based on the conditional independence properties represented by the DAG $\mathcal{D} = (V, E)$ in Figure \ref{fig:dagessgraph}. These PCBNs emerged from various choices of pair-copula families for $C_{12}$, $C_{13}$, $C_{24}$, and $C_{34 \vert 2}$, cf.\ Section \ref{sec:pcbn}. More precisely, we chose from the Clayton, Gumbel, Gaussian, and Student's t pair-copula families. These copula families exhibit considerable differences in their dependence structures and tail behaviours, see the simulation study in \citet{Bauer.Czado.Klein:2012} for an overview. We considered four PCBNs with all four pair copulas $C_{12}$, $C_{13}$, $C_{24}$, and $C_{34 \vert 2}$ coming from the same copula family, respectively. Additionally, we considered $24$ PCBNs with each pair copula $C_{12}$, $C_{13}$, $C_{24}$, and $C_{34 \vert 2}$ coming from a different copula family. Our choices of pair-copula families are given in Table \ref{tab:copulas}. For each choice of pair-copula families we then considered $16$ different parameter configurations arising from a selection of two different parameter values for each pair copula. The parameter values for each pair copula were chosen to correspond to values of Kendall's $\tau$ of $0.25$ and $0.75$, that is one low and one high rank-correlation specification. These Kendall's $\tau$ configurations are summarised in Table \ref{tab:tau}.

\bigskip

\begin{table}[htb]
	\centering
	\begin{tabular*}{\textwidth}{@{\extracolsep{\fill}}lcccccccccccccc}
		\toprule
		Copula & $1$ & $2$ & $3$ & $4$ & $5$ & $6$ & $7$ & $8$ & $9$ & $10$ & $11$ & $12$ & $13$ & $14$\\
		\midrule
		$C_{12}$ & C & G & N & t & C & C & C & C & C & C & G & G & G & G\\
		$C_{13}$ & C & G & N & t & G & G & N & N & t & t & C & C & N & N\\
		$C_{24}$ & C & G & N & t & N & t & G & t & G & N & N & t & C & t\\
		$C_{34 \vert 2}$ & C & G & N & t & t & N & t & G & N & G & t & N & t & C\\
		\midrule
		Copula & $15$ & $16$ & $17$ & $18$ & $19$ & $20$ & $21$ & $22$ & $23$ & $24$ & $25$ & $26$ & $27$ & $28$\\
		\midrule
		$C_{12}$ & G & G & N & N & N & N & N & N & t & t & t & t & t & t\\
		$C_{13}$ & t & t & C & C & G & G & t & t & C & C & G & G & N & N\\
		$C_{24}$ & C & N & G & t & C & t & C & G & G & N & C & N & C & G\\
		$C_{34 \vert 2}$ & N & C & t & G & t & C & G & C & N & G & N & C & G & C\\
		\bottomrule
	\end{tabular*}
	\caption{Selected pair-copula families for $C_{12}$, $C_{13}$, $C_{24}$, $C_{34 \vert 2}$. Copulas were chosen from the Clayton (C), Gumbel (G), Gaussian (N), and Student's t (t) pair-copula families. See Tables \ref{tab:tau} and \ref{tab:param} for further details on the pair-copula families used.} 
	\label{tab:copulas}
\end{table}

\begin{table}[htb]
	\centering
	\small
	\begin{tabular*}{\textwidth}{@{\extracolsep{\fill}}l@{\!}c@{\!}c@{\!}c@{\!}c@{\!}c@{\!}c@{\!}c@{\!}c@{\!}c@{\!}c@{\!}c@{\!}c@{\!}c@{\!}c@{\!}c@{\!}c}
		\toprule
		Copula & $1$ & $2$ & $3$ & $4$ & $5$ & $6$ & $7$ & $8$ & $9$ & $10$ & $11$ & $12$ & $13$ & $14$ & $15$ & $16$\\
		\midrule
		$C_{12}$ & $0.25$ & $0.75$ & $0.25$ & $0.25$ & $0.25$ & $0.75$ & $0.75$ & $0.75$ & $0.25$ & $0.25$ & $0.25$ & $0.75$ & $0.75$ & $0.75$ & $0.25$ & $0.75$\\
		$C_{13}$ & $0.25$ & $0.25$ & $0.75$ & $0.25$ & $0.25$ & $0.75$ & $0.25$ & $0.25$ & $0.75$ & $0.75$ & $0.25$ & $0.75$ & $0.75$ & $0.25$ & $0.75$ & $0.75$\\
		$C_{24}$ & $0.25$ & $0.25$ & $0.25$ & $0.75$ & $0.25$ & $0.25$ & $0.75$ & $0.25$ & $0.75$ & $0.25$ & $0.75$ & $0.75$ & $0.25$ & $0.75$ & $0.75$ & $0.75$\\
		$C_{34 \vert 2}$ & $0.25$ & $0.25$ & $0.25$ & $0.25$ & $0.75$ & $0.25$ & $0.25$ & $0.75$ & $0.25$ & $0.75$ & $0.75$ & $0.25$ & $0.75$ & $0.75$ & $0.75$ & $0.75$\\
		\bottomrule
	\end{tabular*}
	\caption{Selected values of Kendall's $\tau$ for each choice of pair-copula families for $C_{12}$, $C_{13}$, $C_{24}$, $C_{34 \vert 2}$. See Tables \ref{tab:copulas} and \ref{tab:param} for further details on the pair-copula families used.} 
	\label{tab:tau}
\end{table}

Our selection of copula parameters is based on the bijective relationship between the parameters of the Clayton, Gumbel, and Gaussian pair-copula families and the corresponding Kendall's $\tau$. For the Student's t copula, such a bijective relationship exists only between the correlation parameter and Kendall's $\tau$, which is why we set the degrees-of-freedom parameter of each Student's t copula to $\nu = 5$ in order to allow for heavy-tailed dependence. Table \ref{tab:param} summarises the parameters $\theta$, the corresponding Kendall's correlation coefficients $\tau(\theta)$, and the respective tail-dependence coefficients $\lambda_{\text{L}}(\theta) = \lim\limits_{u \to 0} \frac{C_{\theta}(u, u)}{u}$ and $\lambda_{\text{U}}(\theta) = \lim\limits_{u \to 1} \frac{1 - 2 u + C_{\theta}(u, u)}{1 - u}$ for each pair copula $C_{\theta}$, $\theta \in \Theta$, used in the simulation study.

\bigskip

\begin{table}[htb]
	\centering
	\begin{tabular*}{\textwidth}{@{\extracolsep{\fill}}lcccccccc}
		\toprule
		Copula & \multicolumn{2}{@{}c}{Clayton} & \multicolumn{2}{@{}c}{Gumbel} & \multicolumn{2}{@{}c}{Gauss} & \multicolumn{2}{@{}c}{Student}\\
		\midrule
		Parameter(s) & $0.67$ & $6.00$ & $1.33$ & $4.00$ & $0.38$ & $0.92$ & $0.38$, $5$ & $0.92$, $5$\\
		Kendall's $\tau$ & $0.25$ & $0.75$ & $0.25$ & $0.75$ & $0.25$ & $0.75$ & $0.25$ & $0.75$\\
		Lower TDC $\lambda_{\text{L}}$ & $0.35$ & $0.89$ & $0.00$ & $0.00$ & $0.00$ & $0.00$ & $0.15$ & $0.15$\\
		Upper TDC $\lambda_{\text{U}}$ & $0.00$ & $0.00$ & $0.32$ & $0.81$ & $0.00$ & $0.00$ & $0.64$ & $0.64$\\
		\bottomrule
	\end{tabular*}
	\caption{Parameters, Kendall's correlation coefficients, and tail-dependence coefficients (TDCs) of the pair copulas used in the simulation study.}
	\label{tab:param}
\end{table}

Summing up, we have $28$ different PCBNs with $16$ different parameter configurations each, that is $448$ simulation scenarios. In each of the $448$ simulation scenarios we performed $N = 100$ simulation runs, and in each simulation run we generated $n = 1,\!000$ i.i.d.\ observations. The sampling procedure used was described in Section \ref{sec:pcbnstat}.

\bigskip

For each of the $44,\!800$ runs we applied the PC algorithm with the ten different conditional independence tests described in Section \ref{sec:pc}. Those were the widely used test for zero partial correlation (COR) and our novel vine-copula-based tests using either only C-vines (C), or only D-vines (D), or more generally R-vines (R), respectively, together with one of the Kendall's $\tau$ (K), Hoeffding's $D$ (H), or Genest and Rémillard (GR) tests for ordinary stochastic independence. Since zero partial correlation is generally a weaker property than conditional independence, we consider COR only an approximate conditional independence test serving as a benchmark. In a Gaussian framework, however, zero partial correlation is equivalent to conditional independence. This equivalence holds in particular in the scenarios featuring only Gaussian pair copulas, in which case the respective joint copula families are also Gaussian. The corresponding correlation matrices were derived in \citet{Bauer.Czado.Klein:2012}. Each test was performed at the $5 \%$ significance~level.

\subsubsection*{Results}

Let $\mathcal{G}_{f, p, r, t}$ denote the CG obtained from applying the PC algorithm with conditional-in\-de\-pen\-dence test $t \in \lbrace \text{COR}, \text{C-GR}, \text{C-H}, \text{C-K}, \text{D-GR}, \text{D-H}, \text{D-K}, \text{R-GR}, \text{R-H}, \text{R-K} \rbrace$ to the data simulated in run $r \in \lbrace 1, \ldots, 100 \rbrace$ of pair-copula scenario $f \in \lbrace 1, \ldots, 28 \rbrace$ (see Table \ref{tab:copulas}) and parameter configuration $p \in \lbrace 1, \ldots, 16 \rbrace$ (see Tables \ref{tab:tau} and \ref{tab:param}). We compared each CG $\mathcal{G}_{f, p, r, t}$ to the true essential graph $\mathcal{D}^{e}$ in Figure \ref{fig:dagessgraph}, and set $\pi_{f, p, r, t} \coloneqq 1$ if $\mathcal{G}_{f, p, r, t}$ equalled $\mathcal{D}^{e}$ and $\pi_{f, p, r, t} \coloneqq 0$ otherwise. For each pair-copula scenario $f$ and each conditional independence test $t$, we then computed the relative frequency of recovering the correct structure over all parameter configurations $p$ and all runs $r$, which we will denote by $\pi_{f, t} \coloneqq \frac{1}{1600} \sum_{p = 1}^{16} \sum_{r = 1}^{100} \pi_{f, p, r, t}$. Moreover, we determined the \emph{structural Hamming distance} (SHD) \citep{Tsamardinos.Brown.Aliferis:2006} $\delta_{f, p, r, t}$ between each CG $\mathcal{G}_{f, p, r, t}$ and $\mathcal{D}^{e}$. In short, $\delta_{f, p, r, t}$ counts the number of edges that need to be added to, removed from, directed in, or flipped in $\mathcal{G}_{f, p, r, t}$ in order to obtain $\mathcal{D}^{e}$. Hence, $\delta_{f, p, r, t}$ takes a value between zero and $\binom{\abs{V}}{2} = 6$. We again took the average over all parameter configurations $p$ and all runs $r$, yielding the mean SHD $\delta_{f, t} \coloneqq \frac{1}{1600} \sum_{p = 1}^{16} \sum_{r = 1}^{100} \delta_{f, p, r, t}$ for each pair-copula scenario $f$ and each conditional independence test $t$. The results are given in Figures \ref{fig:simstudyrec} and \ref{fig:simstudyshd}, respectively.

\bigskip

\begin{figure}[!htb]
	\centering
		\begin{subfigure}[t]{.49\textwidth}
			\centering \includegraphics[width=.99\textwidth]{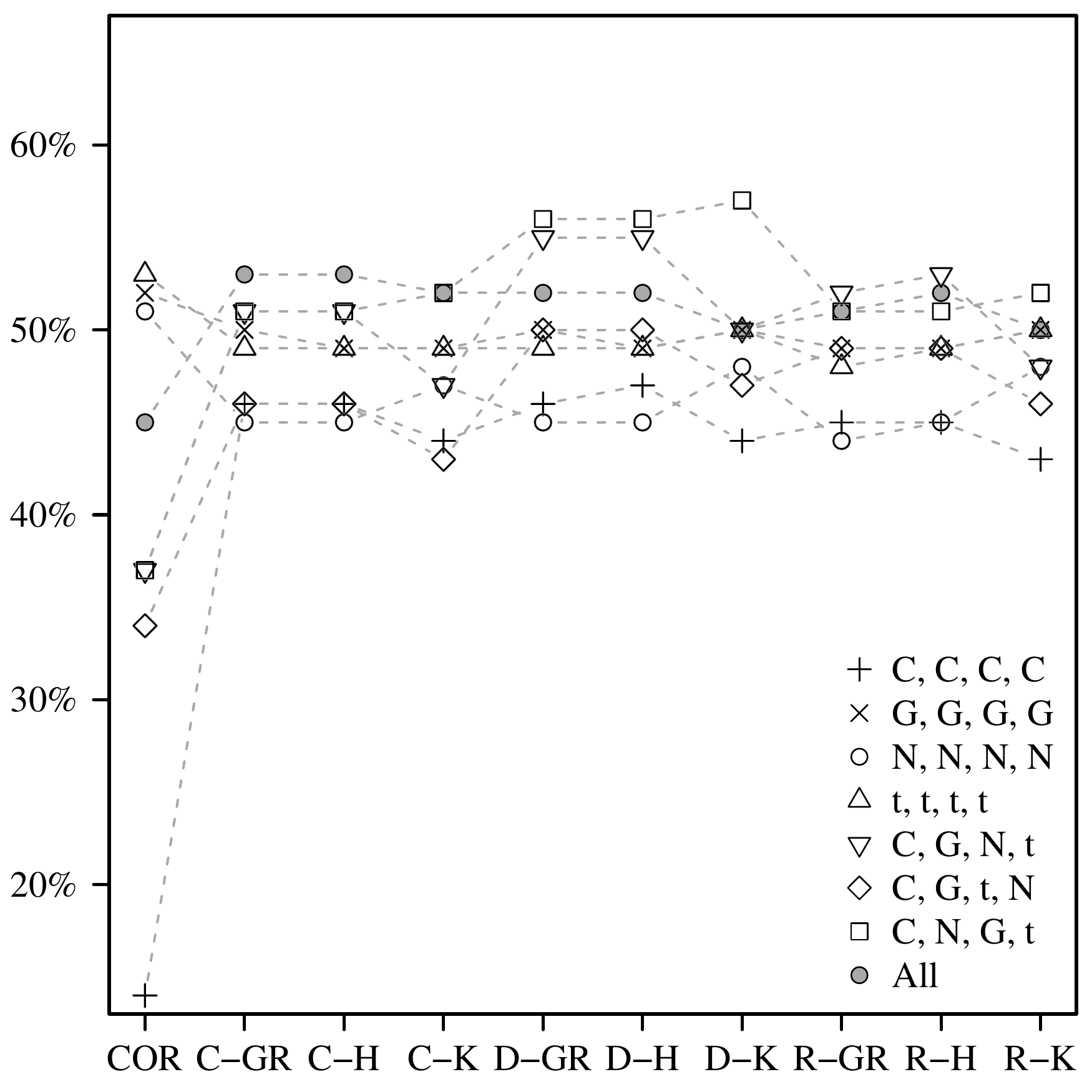}
		\end{subfigure}
		\begin{subfigure}[t]{.49\textwidth}
			\centering \includegraphics[width=.99\textwidth]{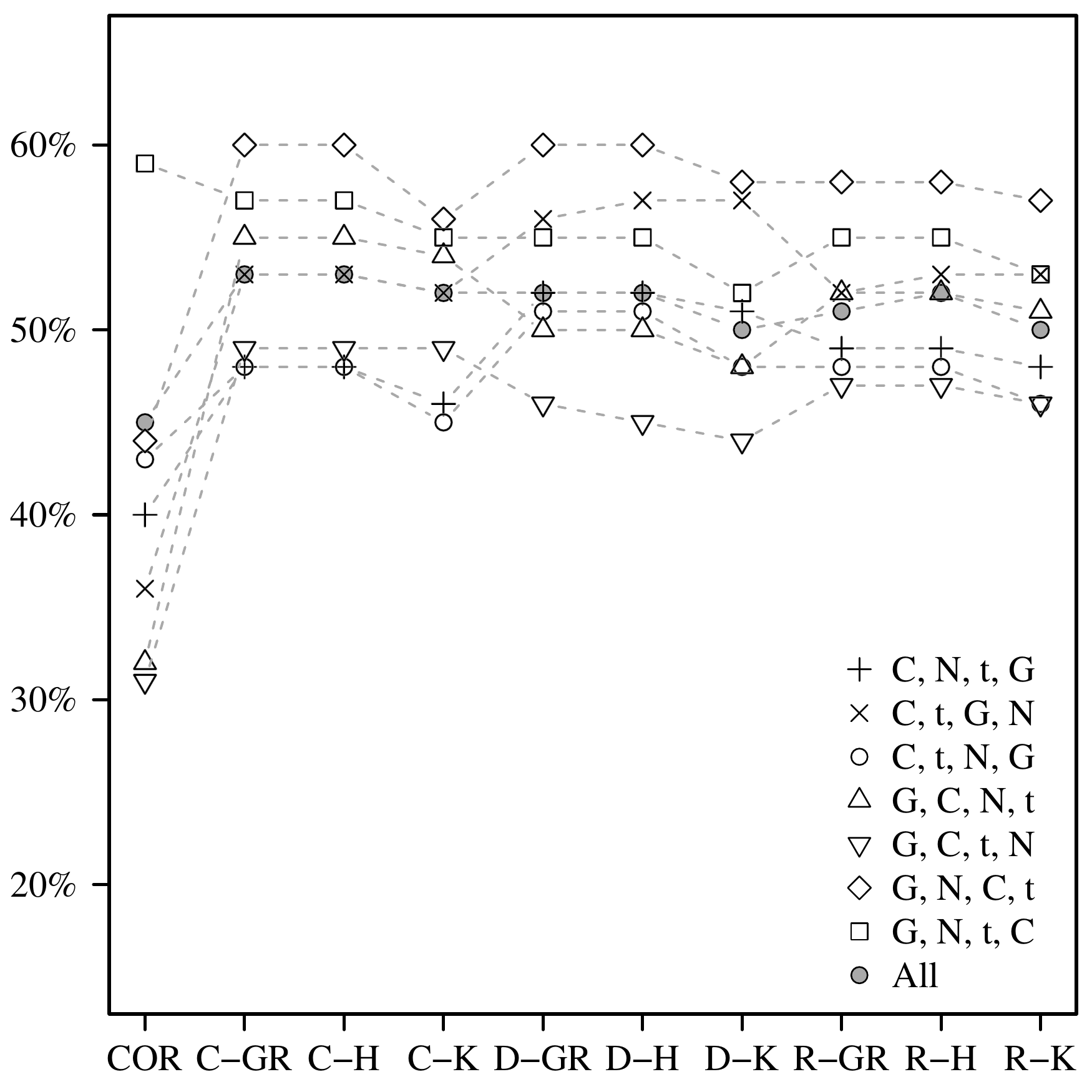}
		\end{subfigure}\\[1ex]
			\begin{subfigure}[t]{.49\textwidth}
			\centering \includegraphics[width=.99\textwidth]{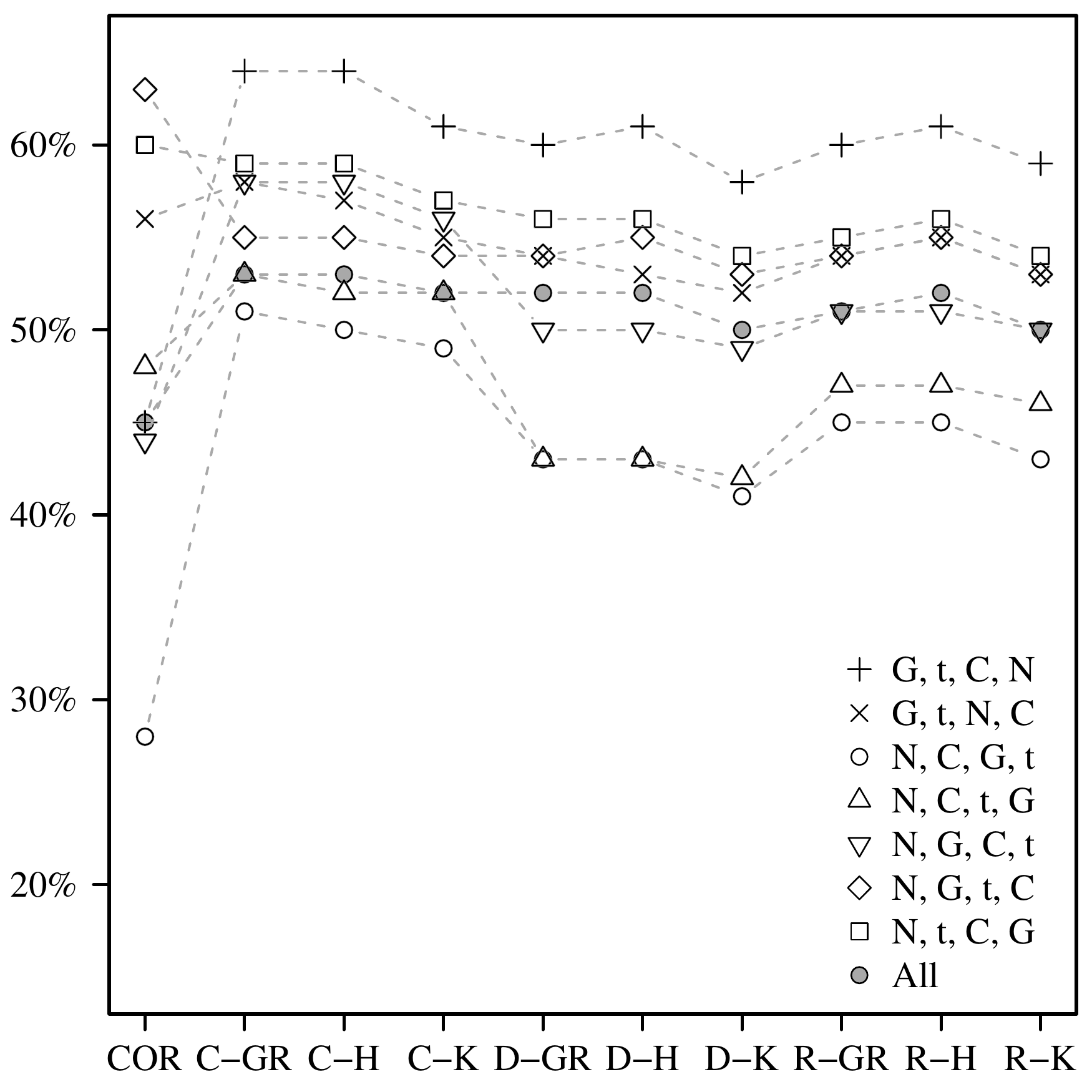}
		\end{subfigure}
		\begin{subfigure}[t]{.49\textwidth}
			\centering \includegraphics[width=.99\textwidth]{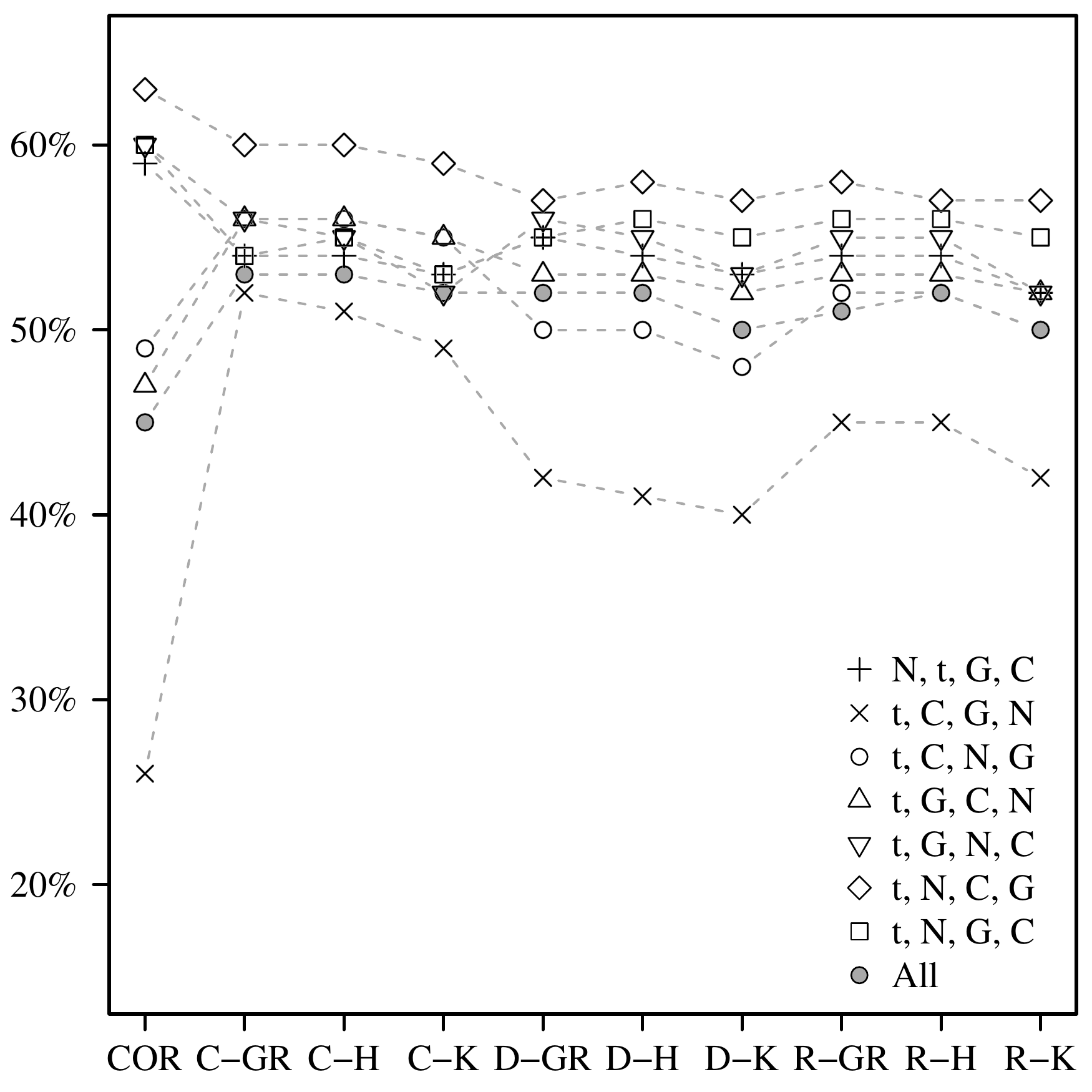}
		\end{subfigure}
	\caption{Percentage $\pi_{f, t}$ of runs in which the PC algorithm returned the correct Markov structure for each choice $f$ of pair-copula families for $C_{12}$, $C_{13}$, $C_{24}$, $C_{34 \vert 2}$ (legends) and each conditional independence test $t$ (horizontal axes) ($1600$ runs each). Copulas were chosen from the Clayton (C), Gumbel (G), Gaussian (N), and Student's t (t) pair-copula families. The percentage of correct recoveries out of all $28$ copula scenarios is given in solid grey.}
	\label{fig:simstudyrec}
\end{figure}

\begin{figure}[!htb]
	\centering
		\begin{subfigure}[t]{.49\textwidth}
			\centering \includegraphics[width=.99\textwidth]{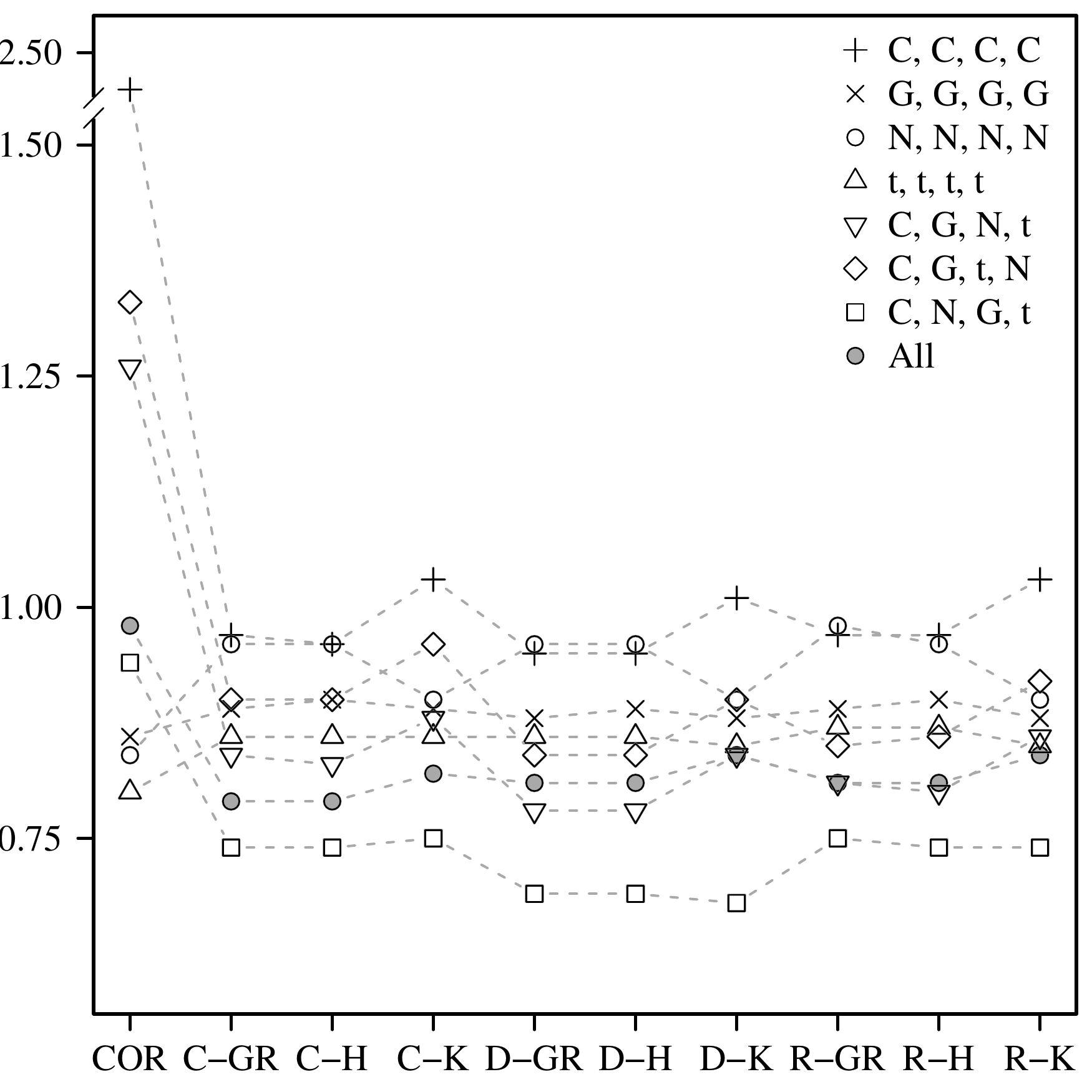}
		\end{subfigure}
		\begin{subfigure}[t]{.49\textwidth}
			\centering \includegraphics[width=.99\textwidth]{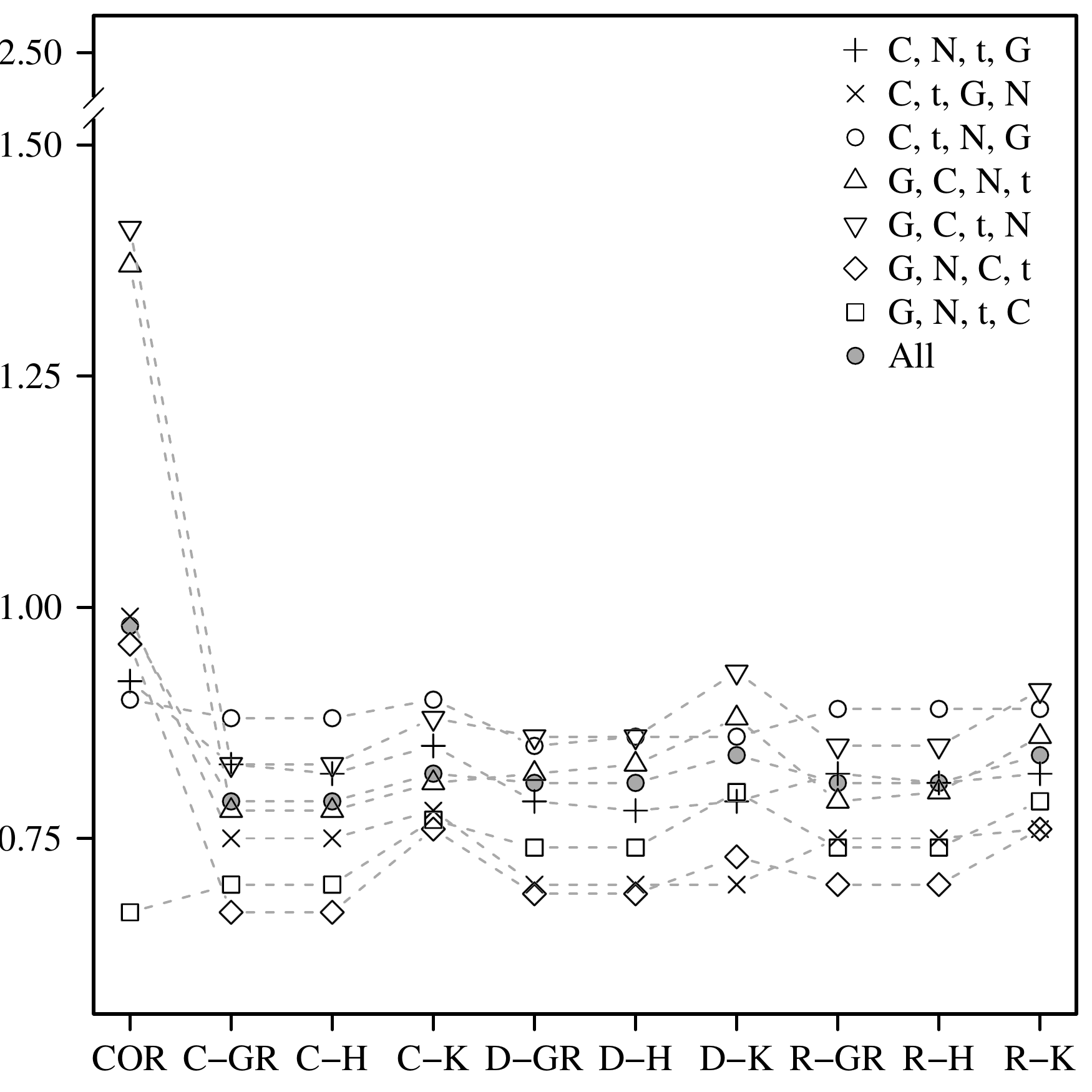}
		\end{subfigure}\\[1ex]
			\begin{subfigure}[t]{.49\textwidth}
			\centering \includegraphics[width=.99\textwidth]{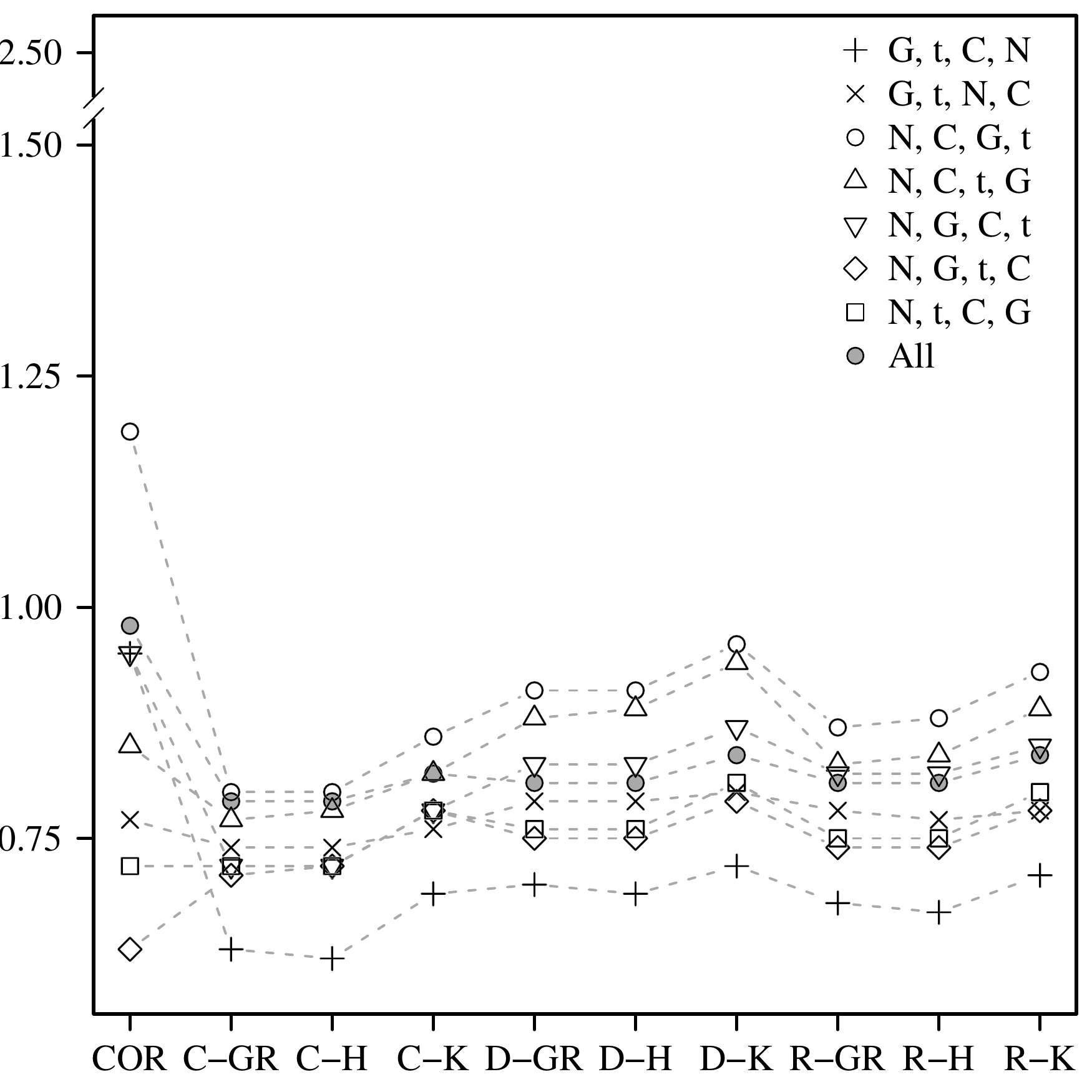}
		\end{subfigure}
		\begin{subfigure}[t]{.49\textwidth}
			\centering \includegraphics[width=.99\textwidth]{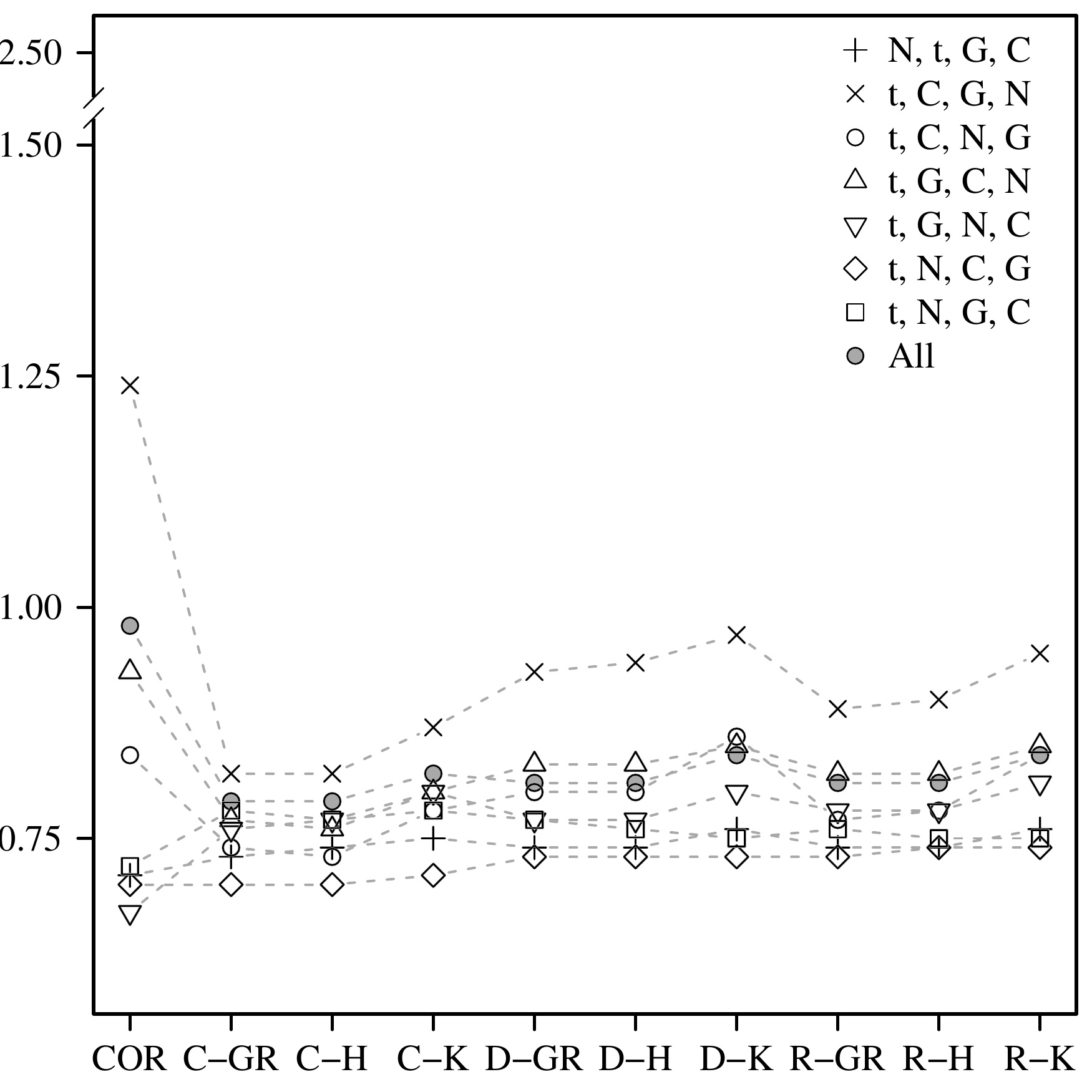}
		\end{subfigure}
	\caption{Average structural Hamming distance (SHD) $\delta_{f, t}$ between the true essential graph $\mathcal{D}^{e}$ and the chain graph $\mathcal{G}_{f, p, r, t}$ returned by the PC algorithm for each choice $f$ of pair-copula families for $C_{12}$, $C_{13}$, $C_{24}$, $C_{34 \vert 2}$ (legends) and each conditional independence test $t$ (horizontal axes) ($1600$ runs each). Copulas were chosen from the Clayton (C), Gumbel (G), Gaussian (N), and Student's t (t) pair-copula families. The average SHD over all $28$ copula scenarios is given in solid grey.}
	\label{fig:simstudyshd}
\end{figure}

Let us first consider Figure \ref{fig:simstudyrec}. The relative frequencies $\pi_{f, \text{COR}}$ range between $14 \%$ and $63 \%$, whereas for the vine-copula-based tests, $\pi_{f, t}$ ranges between $40 \%$ and $64 \%$. COR was outperformed by at least one vine-copula-based test in $18$, and by all vine-copula-based tests in $15$ out of the $28$ copula scenarios. The lowest frequency of $14 \%$ was obtained when applying the PC algorithm with COR to the data sets generated in copula scenario $1$ (numbering as in Table \ref{tab:copulas}), which features only Clayton, that is non-elliptical, copulas. By contrast, COR showed a solid performance in the elliptical-copulas-only scenarios $3$ and $4$, which is not surprising given that COR is based on the partial correlation. In $9$ out of the $28$ copula scenarios, $\pi_{f, \text{COR}}$ is lower than $40 \%$, which is the minimum frequency obtained for the vine-copula-based tests. Also, in these $9$ scenarios, the difference in relative frequencies between COR and the vine-copula-based tests ranges between $9$ and $33$ percentage points. The highest frequency of $64 \%$ was obtained in copula scenario $15$ both for the PC algorithm with C-GR and C-H, respectively. Taking means over all $28$ copula scenarios, we obtain the overall relative frequencies $\pi_{t} \coloneqq \frac{1}{28} \sum_{f = 1}^{28} \pi_{f, t}$ for all tests $t$. These overall frequencies range between $50 \%$ and $53 \%$ for the vine-copula-based tests, while $\pi_{\text{COR}} = 45 \%$. The best performances were again achieved by C-GR and C-H. However, we recommend using the R-vine-based conditional independence tests in higher dimensions since these offer more general tree structures than their C- and D-vine counterparts. Moreover, we observe that choosing H instead of GR as test for unconditional stochastic independence has only little effect on the performance of the vine-copula-based tests. By contrast, relative frequencies were, on average, slightly worse when using K instead of GR and H, respectively. Since zero Kendall's $\tau$ is generally also not equivalent to stochastic independence, we recommend using GR and H. Note that in a given copula scenario $f$ and a given parameter scenario $p$, the relative frequencies $\pi_{f, p, t} \coloneqq \frac{1}{100} \sum_{r = 1}^{100} \pi_{f, p, r, t}$ can be a lot higher than the averages displayed in Figure \ref{fig:simstudyrec}. We observed frequencies $\pi_{f, p, t}$ of up to $98 \%$. To sum up, using a vine-copula-based conditional independence test instead of COR leads to more reliable structure estimates, in particular when the data exhibit non-Gaussian, asymmetric dependence.

\bigskip

Considering only the correctly recovered Markov structures may be a too crude performance measure. Hence, the mean SHDs $\delta_{f, t}$ in Figure \ref{fig:simstudyshd} illustrate how much the results of the PC algorithm differ from the true essential graph $\mathcal{D}^{e}$. For the vine-copula-based tests, $\delta_{f, t}$ ranges between $0.62$ and $1.03$. The respective overall means $\delta_{t} \coloneqq \frac{1}{28} \sum_{f = 1}^{28} \delta_{f, t}$ lie between $0.79$ and $0.84$. Thus, on average, the results of the PC algorithm differ by less than one edge from $\mathcal{D}^{e}$. That is, if the PC algorithm yields a CG that is not equivalent to $\mathcal{D}^{e}$, then, with a high probability, CG and $\mathcal{D}^{e}$ are not too different. The lowest values of $\delta_{t}$ were again obtained for C-GR and C-H. Similarly, $\delta_{f, \text{COR}}$ ranges between $0.63$ and $2.44$, and $\delta_{\text{COR}} = 0.98$, which again shows the superiority of the vine copula approach. The worst mean SHD of $2.44$ was obtained in copula scenario $1$. Overall, we can say that the PC algorithm with either of the $9$ vine-copula-based conditional independence tests provides a suitable procedure for structure estimation in PCBNs.

\bigskip
 
We repeated the simulation study both for a significance level $\alpha$ of $1 \%$ and for a sample size $n$ of $500$. For $\alpha = 1\%$, we obtained results similar to the ones described above for $\alpha = 5 \%$. The overall relative frequencies $\pi_{t}$ were slightly lower, ranging from $44 \%$ to $47 \%$ for the vine-copula-based tests, while $\pi_{\text{COR}}$ was $43 \%$. Also, the overall mean SHDs $\delta_{t}$ ranged between $0.86$ and $0.94$ for the vine-copula-based tests, while $\pi_{\text{COR}}$ was $0.99$. The reduction in sample size to $n = 500$, on the other hand, lead to a slightly stronger decrease in the overall relative frequencies $\pi_{t}$, which then ranged between $39 \%$ and $41 \%$ for the vine-copula-based tests, while $\pi_{\text{COR}}$ was $37 \%$. Similarly, the overall mean SHDs $\delta_{t}$ ranged between $1.07$ and $1.11$ for the vine-copula-based tests, while $\pi_{\text{COR}}$ was $1.17$. Yet, both for $\alpha = 1 \%$ and for $n = 500$, the CGs returned by the PC algorithm differed on average from $\mathcal{D}^{e}$ by only one edge. The performance of the PC algorithm can thus be deemed reliable and robust.

\section{Application: Stock market indices}\label{sec:finance}

As a real-world application, we applied PCBNs to a financial data set comprising ten major international stock market indices. More precisely, we modelled the joint distribution of a portfolio of daily log-returns of the Australian All Ordinaries (AUS), the Canadian S\&P/TSX Composite Index (CAN), the Swiss Market Index (CH), the German DAX (DEU), the French CAC $40$ (FRA), the Hong Kong Hang Seng Index (HK), the Japanese Nikkei $225$ (JPN), the Singapore Straits Times Index (SGP), the UK's FTSE $100$ (UK), and the US S\&P $500$ (USA) from $1$ April $2008$ to $29$ July $2011$ ($n = 733$ observations).

\subsubsection*{Univariate time series models}\enlargethispage{\baselineskip}

Using the inference functions for margins method outlined in Section \ref{sec:vinestat}, we modelled univariate marginal distributions without regard to the dependence structure between variables. We first removed serial correlation in the ten time series of log-returns by applying an AR($1$)-GARCH($1$,$1$) filter, which accounts for conditional heteroskedasticity present in the data, see \citet{Bollerslev:1986}. The log-return $r_{i,t}$ of stock index $i \in \lbrace \text{AUS}, \text{CAN}, \text{CH}, \text{DEU}, \text{FRA}, \text{HK}, \text{JPN}, \text{SGP}, \text{UK}, \text{USA} \rbrace$ at time $t$ can thus be written as
\begin{equation*}
	r_{i, t} = \mu_{i} + a_{i} \, r_{i, t - 1} + \varepsilon_{i, t}, \quad \varepsilon_{i, t} = \sigma_{i, t} \, z_{i, t}, \quad \sigma_{i, t}^{2} = \omega_{i} + \alpha_{i} \, \varepsilon_{i, t - 1}^{2} + \beta_{i} \, \sigma_{i, t - 1}^{2},
\end{equation*}
with parameters $\omega_{i} > 0$, $\alpha_{i}, \beta_{i} \ge 0$ such that $\alpha_{i} + \beta_{i} < 1$, $\abs{a_{i}} < 1$, and $\mu_{i} \in \Real$, where $\EW{[z_{t, i}]} = 0$ and $\Var{[z_{t, i}]} = 1$. The standardised residuals $z_{i, t}$ are assumed to follow a skewed Student's t distribution with $\nu_{i}$ degrees of freedom and skewness parameter $\gamma_{i}$, see \citet[Section $3.2$]{McNeil.Frey.Embrechts:2005}. The corresponding cdf will be denoted by $\text{t}_{\nu_{i}, \gamma_{i}}$. ML parameter estimates and corresponding standard errors derived from numerical evaluation of the Hessian of the AR($1$)-GARCH($1$,$1$) parameters are given in Appendix \ref{app:argarch}. We assessed model fit using the following statistical tests: the Ljung-Box test \citep{Ljung.Box:1978} with null hypothesis that there is no autocorrelation left in the residuals and squared residuals, the Langrange-multiplier ARCH test \citep{Engle:1982} with null hypothesis that the residuals exhibit no conditional heteroskedasticity, and the Kolmogorov-Smirnov test \citep[Section $6.2$]{Conover:1999} with null hypothesis that the residuals follow a skewed Student's t distribution. None of these null hypotheses could be rejected at the $5 \%$ significance level. We then transformed the standardised residuals to uniformly distributed observations $u_{i,t} \coloneqq \text{t}_{\nu_{i}, \gamma_{i}}{\left( \sqrt{\frac{\nu_{i}}{\nu_{i} - 2} + \frac{2 \, \nu_{i}^{2} \, \gamma_{i}^{2}}{(\nu_{i} - 2)^{2} \, (\nu_{i} - 4)}} \, z_{i, t} \right)}$, before modelling the joint dependence structure of the ten time series of log-returns by a PCBN.

\subsubsection*{Estimating the conditional independence structure with the PC algorithm}

We estimated the conditional independence structure of the ten time series of log-returns by applying the PC algorithm with either of the ten conditional independence tests COR, C-GR, C-H, C-K, D-GR, D-H, D-K, R-GR, R-H, and R-K described in Section \ref{sec:pc} (with notation as in Section \ref{sec:simstudy}) to the transformed observations $u_{i,t}$. All tests were performed at the $5 \%$ significance level. As a result, we obtained three different essential graphs $\mathcal{D}_{\text{COR}}^{e}$, $\mathcal{D}_{\text{GR}, \text{H}}^{e}$, and $\mathcal{D}_{\text{K}}^{e}$, of which the first was returned by the PC algorithm with COR, the second was returned by the PC algorithm with either of C-GR, C-H, D-GR, D-H, R-GR, and R-H, and the third was returned by the PC algorithm with either of C-K, D-K, and R-K, respectively. Obviously, a restriction of the class of R-vines to C- or D-vines had not influence on the resulting essential graph. We then oriented undirected edges in the obtained essential graphs, as described in Section \ref{sec:bn}, in order to obtain DAGs $\mathcal{D}_{\text{COR}}$, $\mathcal{D}_{\text{GR}, \text{H}}$, and $\mathcal{D}_{\text{K}}$ from the Markov-equivalence classes represented by $\mathcal{D}_{\text{COR}}^{e}$, $\mathcal{D}_{\text{GR}, \text{H}}^{e}$, and $\mathcal{D}_{\text{K}}^{e}$, respectively. More precisely, $\mathcal{D}_{\text{COR}}^{e}$ contained the two undirected edges $\text{AUS} \edge \text{HK}$ and $\text{CH} \edge \text{DEU}$, which we replaced by $\text{AUS} \rightarrow \text{HK}$ and $\text{CH} \rightarrow \text{DEU}$, respectively, based on the heuristic rule that $\mathcal{D}_{\text{GR}, \text{H}}^{e}$ and $\mathcal{D}_{\text{K}}^{e}$ already contained $\text{AUS} \rightarrow \text{HK}$ and $\text{CH} \rightarrow \text{DEU}$. Similarly, we oriented $\text{AUS} \edge \text{JPN}$ into $\text{AUS} \leftarrow \text{JPN}$ in $\mathcal{D}_{\text{GR}, \text{H}}$ and $\mathcal{D}_{\text{K}}$ since $\mathcal{D}_{\text{COR}}^{e}$ already contained $\text{AUS} \leftarrow \text{JPN}$. The DAGs $\mathcal{D}_{\text{COR}}$, $\mathcal{D}_{\text{GR}, \text{H}}$, and $\mathcal{D}_{\text{K}}$ are given in Figure \ref{fig:financedag}.

\bigskip

\begin{figure}[htb]
	\centering
		\begin{subfigure}[t]{.47\textwidth}
			\centering \includegraphics[width=\textwidth]{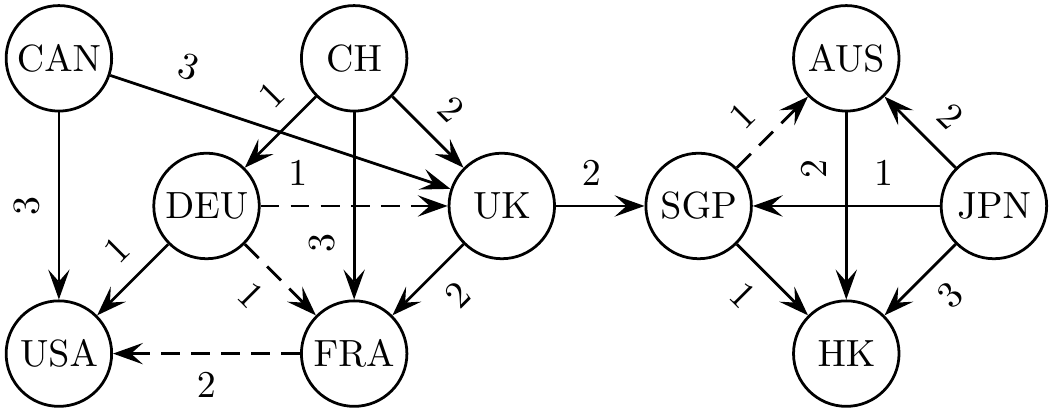}
		\end{subfigure}\qquad
		\begin{subfigure}[t]{.47\textwidth}
			\centering \includegraphics[width=\textwidth]{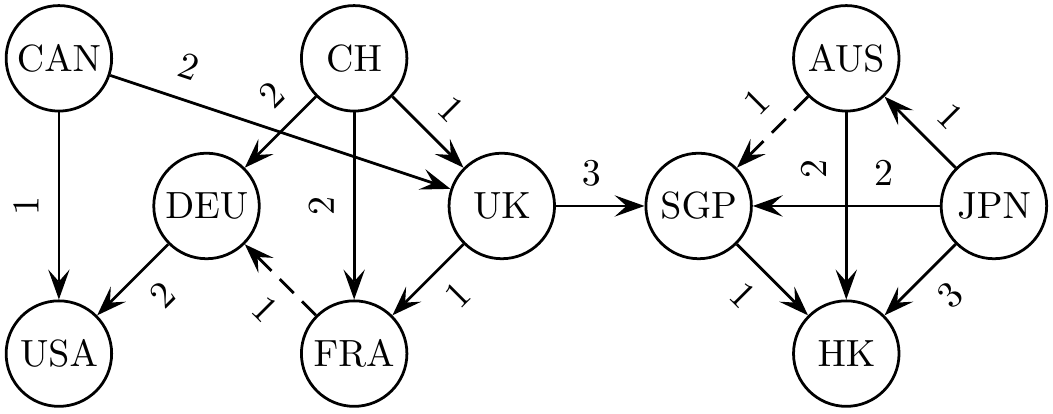}
		\end{subfigure}\\[15pt]
		\begin{subfigure}[t]{.47\textwidth}
			\centering \includegraphics[width=\textwidth]{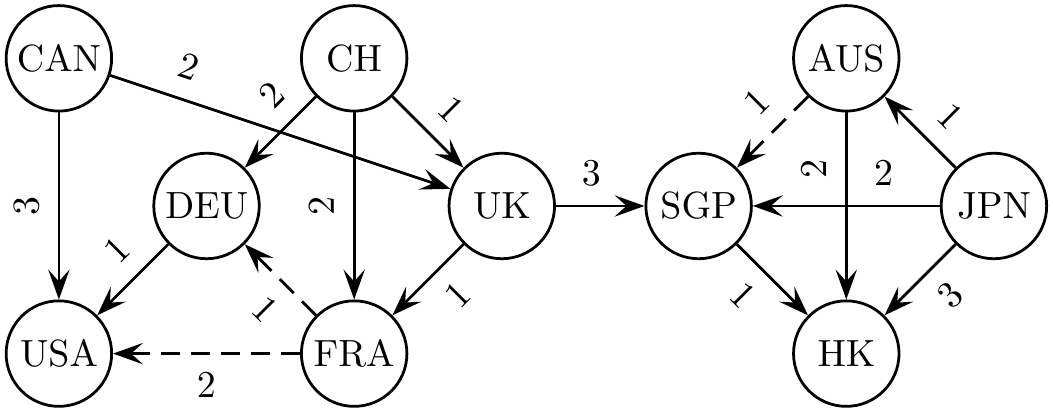}
		\end{subfigure}
	\caption{DAGs $\mathcal{D}_{\text{COR}}$ (top left), $\mathcal{D}_{\text{GR}, \text{H}}$ (top right), $\mathcal{D}_{\text{K}}$ (bottom) returned by the PC algorithm with different conditional independence tests when estimating the Markov structure of the ten time series AUS, CAN, CH, DEU, FRA, HK, JPN, SGP, UK, USA of daily log-returns. Solid edges appear in all three DAGs. Edge labels indicate parent orderings, that is, for instance, $\text{CAN} <_{\text{USA}} \text{DEU} <_{\text{USA}} \text{FRA}$ in $\mathcal{D}_{\text{COR}}$.}
	\label{fig:financedag}
\end{figure}

In all three DAGs in Figure \ref{fig:financedag}, the Asian-Pacific indices AUS, HK, JPN, and SGP are mutually adjacent, and so are the two North American indices CAN and USA. The same holds true for the European indices CH, DEU, FRA, and UK in DAG $\mathcal{D}_{\text{COR}}$, while DEU and UK are non-adjacent in $\mathcal{D}_{\text{GR}, \text{H}}$ and $\mathcal{D}_{\text{K}}$. A probability measure satisfying the Markov properties represented by either $\mathcal{D}_{\text{GR}, \text{H}}$ or $\mathcal{D}_{\text{K}}$, respectively, observes the conditional independence restriction $\cind{\text{DEU}}{\text{UK}}{\lbrace \text{CH}, \text{FRA} \rbrace}$. All further conditional independence restrictions represented by the DAGs in Figure \ref{fig:financedag} involve indices in at least two of the above given regions Asia-Pacific, Europe, and North America. We hence observe a strong geographical clustering of dependences. Moreover, all three DAGs in Figure \ref{fig:financedag} represent the conditional independence restriction $\cind{\lbrace \text{AUS}, \text{HK}, \text{JPN}, \text{SGP} \rbrace}{\lbrace \text{CAN}, \text{USA} \rbrace}{\lbrace \text{CH}, \text{DEU}, \text{FRA}, \text{UK} \rbrace}$, that is, $\cind{\text{Asia-Pacific}}{\text{North America}}{\text{Europe}}$. Note that Markov properties alone are not sufficient for deriving causal relations within the analysed data (see, for instance, the undirected edges in an essential graph), but they can be used as a starting point for further research in that direction.

\bigskip

A well-ordering for $\mathcal{D}_{\text{COR}}$ is given by $1 \mapsto \text{CAN}$, $2 \mapsto \text{CH}$, $3 \mapsto \text{DEU}$, $4 \mapsto \text{UK}$, $5 \mapsto \text{FRA}$, $6 \mapsto \text{USA}$, $7 \mapsto \text{JPN}$, $8 \mapsto \text{SGP}$, $9 \mapsto \text{AUS}$, $10 \mapsto \text{HK}$. Similarly, we obtain a well-ordering for $\mathcal{D}_{\text{GR}, \text{H}}$ and $\mathcal{D}_{\text{K}}$, respectively, by mapping $1 \mapsto \text{CAN}$, $2 \mapsto \text{CH}$, $3 \mapsto \text{UK}$, $4 \mapsto \text{FRA}$, $5 \mapsto \text{DEU}$, $6 \mapsto \text{USA}$, $7 \mapsto \text{JPN}$, $8 \mapsto \text{AUS}$, $9 \mapsto \text{SGP}$, $10 \mapsto \text{HK}$. We determined parent orderings for the three DAGs in Figure \ref{fig:financedag} in two steps. First, we applied the greedy-type procedure with Kendall's $\tau$ edge weights described in Section \ref{sec:pcbnstat}, and second, we permuted some of the orderings obtained in step one to reduce the number of integrals in the corresponding pair-copula decompositions and thus the computational complexity. More precisely, we changed $\text{JPN} <_{\text{AUS}} \text{SGP}$ and $\text{CAN} <_{\text{USA}} \text{DEU} <_{\text{USA}} \text{FRA}$ in DAG $\mathcal{D}_{\text{COR}}$ into $\text{SGP} <_{\text{AUS}} \text{JPN}$ and $\text{DEU} <_{\text{USA}} \text{FRA} <_{\text{USA}} \text{CAN}$, respectively, and $\text{CAN} <_{\text{USA}} \text{DEU} <_{\text{USA}} \text{FRA}$ in DAG $\mathcal{D}_{\text{K}}$ into $\text{DEU} <_{\text{USA}} \text{FRA} <_{\text{USA}} \text{CAN}$. The resulting parent orderings for $\mathcal{D}_{\text{COR}}$,  $\mathcal{D}_{\text{GR}, \text{H}}$, and $\mathcal{D}_{\text{K}}$, respectively, are displayed in Figure \ref{fig:financedag}.

\subsubsection*{Pair-copula selection and ML estimation}

Having fixed the parent orderings for the three PCBNs corresponding to $\mathcal{D}_{\text{COR}}$,  $\mathcal{D}_{\text{GR}, \text{H}}$, and $\mathcal{D}_{\text{K}}$, respectively, we next selected parametric copula families using the AIC as a selection criterion. We considered the Clayton, Frank, Gaussian, Gumbel, and Student's t copula families as well as reflected versions of the Clayton and Gumbel copula families in order to account for negative correlations. We then computed sequential ML estimates of the parameters of the so specified PCBNs. Selected pair-copula families, corresponding sequential ML estimates, bootstrapped standard errors, and estimates of Kendall's $\tau$ are given in Table \ref{tab:financeparam}. The respective maximised log-likelihoods and AIC values are summarised in Table \ref{tab:financeloglik}. Moreover, we compared model fit to the respective Gaussian PCBNs comprising only Gaussian pair copulas. Corresponding ML estimates, standard errors, and estimates of Kendall's $\tau$ are again found in Table \ref{tab:financeparam}, while maximised log-likelihoods and AIC values are given in Table \ref{tab:financeloglik}.

\bigskip

\begin{table}[!htb]
	\centering
	\footnotesize
	\begin{tabular*}{\linewidth}{@{\extracolsep{\fill}}l@{}c@{}c@{}r@{}c@{}c@{}r@{}c@{}c@{}r@{}c@{}}
		\toprule
		DAG & & \multicolumn{3}{@{}c}{$\mathcal{D}_{\text{COR}}$} & \multicolumn{3}{@{}c}{$\mathcal{D}_{\text{GR}, \text{H}}$} & \multicolumn{3}{@{}c}{$\mathcal{D}_{\text{K}}$}\\
		\midrule
		& & & Parameters & $\widehat{\tau}$ & & Parameters & $\widehat{\tau}$ & & Parameters & $\widehat{\tau}$\\
		\midrule
		JPN $\rightarrow$ AUS & $\text{n}\mathcal{G}$ & t & $0.56, 10.3$ $(0.03, 4.1)$ & $0.38$ & t & $0.73, 8.5$ $(0.02, 3.4)$ & $0.52$ & t & $0.73, 8.5$ $(0.02, 3.4)$ & $0.52$\\
		& $\mathcal{G}$ & N & $0.55$ $(0.03)$ & $0.37$ & N & $0.72$ $(0.02)$ & $0.52$ & N & $0.72$ $(0.02)$ & $0.52$\\
		SGP $\rightarrow$ AUS & $\text{n}\mathcal{G}$ & t & $0.64, 9.4$ $(0.02, 4.1)$ & $0.44$ &  &  &  &  &  & \\
		& $\mathcal{G}$ & N & $0.64$ $(0.02)$ & $0.44$ &  &  &  &  &  & \\
		\midrule
		CH $\rightarrow$ DEU & $\text{n}\mathcal{G}$ & t & $0.83, 5.4$ $(0.01, 1.8)$ & $0.63$ & F & $0.98$ $(0.23)$ & $0.11$ & F & $0.98$ $(0.23)$ & $0.11$\\
		& $\mathcal{G}$ & N & $0.82$ $(0.01)$ & $0.62$ & N & $0.14$ $(0.04)$ & $0.09$ & N & $0.14$ $(0.04)$ & $0.09$\\
		FRA $\rightarrow$ DEU & $\text{n}\mathcal{G}$ &  &  &  & t & $0.94, 3.8$ $(0.01, 1.2)$ & $0.78$ & t & $0.94, 3.8$ $(0.01, 1.2)$ & $0.78$\\
		& $\mathcal{G}$ &  &  &  & N & $0.93$ $(0.01)$ & $0.76$ & N & $0.93$ $(0.01)$ & $0.76$\\
		\midrule
		CH $\rightarrow$ FRA & $\text{n}\mathcal{G}$ & F & $1.55$ $(0.24)$ & $0.17$ & t & $0.44, 9.0$ $(0.03, 4.4)$ & $0.29$ & t & $0.44, 9.0$ $(0.03, 4.3)$ & $0.29$\\
		& $\mathcal{G}$ & N & $0.28$ $(0.04)$ & $0.18$ & N & $0.44$ $(0.03)$ & $0.29$ & N & $0.44$ $(0.03)$ & $0.29$\\
		DEU $\rightarrow$ FRA & $\text{n}\mathcal{G}$ & t & $0.94, 3.8$ $(0.01, 1.1)$ & $0.78$ &  &  &  &  &  & \\
		& $\mathcal{G}$ & N & $0.93$ $(0.01)$ & $0.76$ &  &  &  &  &  & \\
		UK $\rightarrow$ FRA & $\text{n}\mathcal{G}$ & t & $0.57, 6.9$ $(0.03, 3.3)$ & $0.38$ & t & $0.92, 7.3$ $(0.01, 2.8)$ & $0.74$ & t & $0.92, 7.3$ $(0.01, 2.8)$ & $0.74$\\
		& $\mathcal{G}$ & N & $0.59$ $(0.04)$ & $0.40$ & N & $0.92$ $(0.01)$ & $0.74$ & N & $0.92$ $(0.01)$ & $0.74$\\
		\midrule
		AUS $\rightarrow$ HK & $\text{n}\mathcal{G}$ & t & $0.35, 13.4$ $(0.03, 4.8)$ & $0.22$ & t & $0.35, 13.4$ $(0.03, 4.8)$ & $0.22$ & t & $0.35, 13.4$ $(0.03, 4.8)$ & $0.22$\\
		& $\mathcal{G}$ & N & $0.36$ $(0.04)$ & $0.24$ & N & $0.36$ $(0.04)$ & $0.24$ & N & $0.36$ $(0.04)$ & $0.24$\\
		JPN $\rightarrow$ HK & $\text{n}\mathcal{G}$ & t & $0.20, 20.0$ $(0.04, 2.8)$ & $0.13$ & N & $0.20$ $(0.04)$ & $0.13$ & N & $0.20$ $(0.04)$ & $0.13$\\
		& $\mathcal{G}$ & N & $0.20$ $(0.04)$ & $0.13$ & N & $0.20$ $(0.04)$ & $0.13$ & N & $0.20$ $(0.04)$ & $0.13$\\
		SGP $\rightarrow$ HK & $\text{n}\mathcal{G}$ & t & $0.78, 5.7$ $(0.01, 1.9)$ & $0.57$ & t & $0.78, 5.7$ $(0.01, 2.0)$ & $0.57$ & t & $0.78, 5.7$ $(0.01, 2.0)$ & $0.57$\\
		& $\mathcal{G}$ & N & $0.78$ $(0.02)$ & $0.57$ & N & $0.78$ $(0.02)$ & $0.57$ & N & $0.78$ $(0.02)$ & $0.57$\\
		\midrule
		AUS $\rightarrow$ SGP & $\text{n}\mathcal{G}$ &  &  &  & t & $0.64, 9.4$ $(0.02, 4.0)$ & $0.44$ & t & $0.64, 9.4$ $(0.02, 4.0)$ & $0.44$\\
		& $\mathcal{G}$ &  &  &  & N & $0.64$ $(0.02)$ & $0.44$ & N & $0.64$ $(0.02)$ & $0.44$\\
		JPN $\rightarrow$ SGP & $\text{n}\mathcal{G}$ & t & $0.60, 11.6$ $(0.02, 4.5)$ & $0.41$ & N & $0.27$ $(0.03)$ & $0.17$ & N & $0.27$ $(0.03)$ & $0.17$\\
		& $\mathcal{G}$ & N & $0.60$ $(0.02)$ & $0.41$ & N & $0.26$ $(0.03)$ & $0.17$ & N & $0.26$ $(0.04)$ & $0.17$\\
		UK $\rightarrow$ SGP & $\text{n}\mathcal{G}$ & N & $0.34$ $(0.03)$ & $0.22$ & N & $0.27$ $(0.03)$ & $0.18$ & N & $0.27$ $(0.03)$ & $0.18$\\
		& $\mathcal{G}$ & N & $0.34$ $(0.03)$ & $0.22$ & N & $0.27$ $(0.03)$ & $0.18$ & N & $0.27$ $(0.03)$ & $0.18$\\
		\midrule
		CAN $\rightarrow$ UK & $\text{n}\mathcal{G}$ & SG & $1.13$ $(0.03)$ & $0.12$ & N & $0.31$ $(0.03)$ & $0.20$ & N & $0.31$ $(0.03)$ & $0.20$\\
		& $\mathcal{G}$ & N & $0.21$ $(0.04)$ & $0.13$ & N & $0.31$ $(0.03)$ & $0.20$ & N & $0.31$ $(0.03)$ & $0.20$\\
		CH $\rightarrow$ UK & $\text{n}\mathcal{G}$ & t & $0.36, 9.6$ $(0.03, 4.5)$ & $0.23$ & t & $0.83, 8.6$ $(0.01, 3.7)$ & $0.62$ & t & $0.83, 8.6$ $(0.01, 3.7)$ & $0.62$\\
		& $\mathcal{G}$ & N & $0.39$ $(0.04)$ & $0.25$ & N & $0.83$ $(0.01)$ & $0.62$ & N & $0.83$ $(0.01)$ & $0.62$\\
		DEU $\rightarrow$ UK & $\text{n}\mathcal{G}$ & t & $0.88, 7.3$ $(0.01, 3.0)$ & $0.69$ &  &  &  &  &  & \\
		& $\mathcal{G}$ & N & $0.88$ $(0.01)$ & $0.68$ &  &  &  &  &  & \\
		\midrule
		CAN $\rightarrow$ USA & $\text{n}\mathcal{G}$ & N & $0.48$ $(0.03)$ & $0.32$ & t & $0.75, 6.7$ $(0.02, 3.0)$ & $0.54$ & N & $0.52$ $(0.03)$ & $0.35$\\
		& $\mathcal{G}$ & N & $0.48$ $(0.03)$ & $0.32$ & N & $0.75$ $(0.01)$ & $0.54$ & N & $0.52$ $(0.03)$ & $0.35$\\
		DEU $\rightarrow$ USA & $\text{n}\mathcal{G}$ & t & $0.71, 6.8$ $(0.02, 3.0)$ & $0.51$ & t & $0.47, 12.7$ $(0.03, 4.9)$ & $0.31$ & t & $0.71, 6.8$ $(0.02, 3.0)$ & $0.51$\\
		& $\mathcal{G}$ & N & $0.71$ $(0.02)$ & $0.50$ & N & $0.47$ $(0.03)$ & $0.31$ & N & $0.71$ $(0.02)$ & $0.50$\\
		FRA $\rightarrow$ USA & $\text{n}\mathcal{G}$ & t & $0.19, 9.2$ $(0.04, 4.5)$ & $0.12$ &  &  &  & t & $0.19, 9.2$ $(0.04, 4.4)$ & $0.12$\\
		& $\mathcal{G}$ & N & $0.21$ $(0.05)$ & $0.14$ &  &  &  & N & $0.21$ $(0.05)$ & $0.14$\\
 		\bottomrule
	\end{tabular*}
	\caption{Selected pair-copula families, sequential ML estimates, standard errors (parentheses), and estimates of Kendall's $\tau$ for the Gaussian ($\mathcal{G}$) and non-Gaussian ($\text{n}\mathcal{G}$) PCBNs corresponding to the DAGs in Figure \ref{fig:financedag}. Copulas include the Frank (F), Gaussian (N), Survival-Gumbel (SG), and Student's t (t) pair-copula families.} 
	\label{tab:financeparam}
\end{table}

\begin{table}[htb]
	\centering
	\begin{tabular*}{.75\textwidth}{@{\extracolsep{\fill}}lcrrr@{}}
		\toprule
		DAG & & LL & \# Parameters & AIC\\
		\midrule
		$\mathcal{D}_{\text{COR}}$ & $\text{n}\mathcal{G}$ & $3397.0$ & $30$ & $-6734.0$\\
		& $\mathcal{G}$ & $3264.6$ & $17$ & $-6495.3$\\
		\midrule
		$\mathcal{D}_{\text{GR}, \text{H}}$ & $\text{n}\mathcal{G}$ & $3412.8$ & $25$ & $-6775.6$\\
		& $\mathcal{G}$ & $3285.5$ & $15$ & $-6540.9$\\
		\midrule
		$\mathcal{D}_{\text{K}}$ & $\text{n}\mathcal{G}$ & $3401.9$ & $26$ & $-6751.8$\\
		& $\mathcal{G}$ & $3264.1$ & $16$ & $-6496.3$\\
		\bottomrule
	\end{tabular*}
	\caption{Maximised log-likelihoods, numbers of parameters, and AIC values for the Gaussian ($\mathcal{G}$) and non-Gaussian ($\text{n}\mathcal{G}$) PCBNs corresponding to the DAGs in Figure \ref{fig:financedag}. Sequential ML estimates of the corresponding parameters are given in Table \ref{tab:financeparam}.} 
	\label{tab:financeloglik}
\end{table}

According to the AIC, the best fit was obtained by the non-Gaussian PCBN with DAG $\mathcal{D}_{\text{GR}, \text{H}}$, followed by the non-Gaussian PCBNs associated to $\mathcal{D}_{\text{K}}$ and $\mathcal{D}_{\text{COR}}$, respectively. Applying the Vuong test with AIC correction \citep{Vuong:1989} to the non-Gaussian PCBNs at the $5 \%$ level, we cannot reject the null hypothesis that all three models  are equally close to the true model. A similar statement holds for the Gaussian PCBNs. However, using the Vuong test for model selection between a Gaussian and a non-Gaussian PCBN will always decide in favor of the non-Gaussian model, which again shows the latter models' superiority. \enlargethispage{\baselineskip}% Remove!
This is, of course, to be expected since financial returns often exhibit heavy-tailed dependence, which is validated here by the low estimates of the degrees-of-freedom parameters of the Student's t copulas.

\section{Conclusion}\label{sec:conclusion} \enlargethispage{\baselineskip}% Remove!

We have investigated a novel procedure for constructing non-Gaussian continuous Bayesian networks that uses bivariate copulas as building blocks. The resulting models can accommodate a great variety of distributional features to be modelled such as tail-dependence and non-linear, asymmetric dependence. We have provided an algorithm for deriving explicit representations of the corresponding log-likelihoods, as well as routines for random sampling and model selection.

\bigskip

Depending on the underlying DAG and the corresponding parent orderings, the evaluation of the log-likelihood of a PCBN may involve high-dimensional numerical integration and hence considerable computational effort. We have presented a greedy procedure for selecting the parent orderings of the vertices of the underlying DAG, which is based on the idea of modelling strongest dependences in the unconditional pair-copulas. In Section \ref{sec:finance}, we introduced an additional selection step, in which some of the parent sets were rearranged in order to reduce the number of integrals in the corresponding likelihood decompositions. It would be desirable to have theoretical results on the relationship between parent orderings and the number and complexity of integrals. \citet{Bauer.Czado.Klein:2012} suggested to replace some or all of the integrals by non-parametric kernel conditional cdf estimators. Another way of reducing computational complexity is to consider sequential instead of joint ML estimates.

\bigskip

We used vine copula models to derive a novel test for conditional independence of continuous random variables. The quality of the test, by design, greatly benefits from the ongoing research on vine copulas. In combination with the PC algorithm, we obtained a structure estimation procedure for non-Gaussian PCBNs, which proved to be reliable in the simulation study in Section \ref{sec:simstudy}. One may investigate the performance of other conditional independence tests like \citet{Zhang.Peters.Janzing.Schoelkopf:2011}, as well as of other structure estimation algorithms. Also, recall that by \citet{Meek:1995}, constraint-based estimation algorithms can be adapted to incorporate existing expert knowledge. The distributional flexibility of pair-copula Bayesian networks may become even more apparent in application areas other than finance.

\section*{Acknowledgements}

The authors are very grateful to Peter Hepperger for his help in implementing the algorithms of Section \ref{sec:pcbn} in \texttt{C++}. The computer programs were tested on a Linux cluster supported by the DFG (German Research Foundation). Alexander Bauer acknowledges the support of the TUM Graduate School's Faculty Graduate Center ISAM (International School of Applied Mathematics) at the Technische Universität München.

\appendix
		
\section{Estimated AR-GARCH parameters and standard errors}\label{app:argarch}

\begin{table}[htb]
	\centering
	\small
	\begin{tabular*}{\textwidth}{@{\extracolsep{\fill}}lrrrrrrr@{}}
		\toprule
		& $\mu$ $[\times 10^{3}]$ & $a$ & $\omega$ $[\times 10^{5}]$ & $\alpha$ & $\beta$ & $\nu$ & $\gamma$\\
		\midrule
		AUS & $0.26$ $(0.41)$ & $0.01$ $(0.04)$ & $0.18$ $(3.77)$ & $0.08$ $(0.02)$ & $0.91$ $(0.02)$ & $10.33$ $(3.46)$ & $0.90$ $(0.05)$\\ 
		CAN & $0.50$ $(0.38)$ & $-0.01$ $(0.04)$ & $0.17$ $(3.83)$ & $0.09$ $(0.02)$ & $0.90$ $(0.02)$ & $10.74$ $(3.73)$ & $0.80$ $(0.05)$\\ 
		CH & $0.08$ $(0.38)$ & $0.02$ $(0.04)$ & $0.38$ $(4.01)$ & $0.12$ $(0.03)$ & $0.86$ $(0.02)$ & $7.64$ $(2.04)$ & $0.92$ $(0.05)$\\ 
		DEU & $0.66$ $(0.49)$ & $-0.03$ $(0.04)$ & $0.31$ $(4.03)$ & $0.08$ $(0.02)$ & $0.91$ $(0.02)$ & $7.77$ $(2.39)$ & $0.93$ $(0.05)$\\ 
		FRA & $0.23$ $(0.54)$ & $-0.02$ $(0.04)$ & $0.59$ $(4.24)$ & $0.09$ $(0.02)$ & $0.89$ $(0.02)$ & $8.29$ $(2.44)$ & $0.94$ $(0.05)$\\ 
		HK & $0.36$ $(0.55)$ & $-0.02$ $(0.04)$ & $0.23$ $(4.07)$ & $0.07$ $(0.02)$ & $0.93$ $(0.01)$ & $8.23$ $(2.25)$ & $0.97$ $(0.05)$\\ 
		JPN & $0.27$ $(0.53)$ & $-0.06$ $(0.04)$ & $0.85$ $(4.39)$ & $0.12$ $(0.03)$ & $0.85$ $(0.02)$ & $15.33$ $(7.85)$ & $0.87$ $(0.05)$\\ 
		SGP & $0.57$ $(0.40)$ & $-0.01$ $(0.04)$ & $0.24$ $(3.98)$ & $0.09$ $(0.02)$ & $0.90$ $(0.02)$ & $4.92$ $(0.89)$ & $1.03$ $(0.05)$\\ 
		UK & $0.56$ $(0.44)$ & $-0.01$ $(0.04)$ & $0.37$ $(4.02)$ & $0.09$ $(0.02)$ & $0.89$ $(0.02)$ & $8.42$ $(2.60)$ & $0.92$ $(0.05)$\\ 
		USA & $0.76$ $(0.43)$ & $-0.07$ $(0.04)$ & $0.23$ $(4.03)$ & $0.10$ $(0.02)$ & $0.89$ $(0.02)$ & $7.22$ $(2.15)$ & $0.84$ $(0.04)$\\ 
		\bottomrule
	\end{tabular*}
	\caption{ML estimates and standard errors (in parentheses) of AR($1$)-GARCH($1$,$1$) parameters for the ten time series of daily log-returns analysed in Section \ref{sec:finance}.} 
	\label{tab:argarch}
\end{table}

\end{document}